\documentclass[11pt,letterpaper]{article}

\usepackage{tikz}
\definecolor{darkgreen}{rgb}{0,0.5,0}
\usepackage{hyperref}
\hypersetup{
    unicode=false,          % non-Latin characters in Acrobat’s bookmarks
    colorlinks=true,        % false: boxed links; true: colored links
    linkcolor=red,          % color of internal links (change box color with linkbordercolor)
    citecolor=darkgreen,        % color of links to bibliography
    filecolor=magenta,      % color of file links
    urlcolor=cyan           % color of external links
}
\usepackage{amssymb,amsmath,amsfonts,fullpage,mdwlist,enumerate,amsthm, graphicx,algpseudocode, algorithm, algorithmicx, cite, xspace, thm-restate, framed}
\usepackage[capitalize, nameinlink]{cleveref}

\newcommand{\remove}[1]{}

\newtheorem{theorem}{Theorem}[section]
\newtheorem{claim}[theorem]{Claim}
\newtheorem{lemma}[theorem]{Lemma}
\newtheorem{corollary}[theorem]{Corollary}
\newtheorem{definition}[theorem]{Definition}
\newtheorem{question}[theorem]{Question}
\newcommand{\poly}{\mathop\mathrm{poly}}

\newcommand{\local}{\textsc{Local}\xspace}
\newcommand{\congest}{\textsc{Congest}\xspace}

\newcommand{\mtodo}[1]{\textcolor{blue}{[TODO, Michal: #1]}}

\newcommand{\cov}{\mathsf{Cov}}
\newcommand{\cut}{\mathsf{Cut}}

\newcommand{\cutInfo}{\mathsf{CutInfo}}
\newcommand{\fragInfo}{\mathsf{FragInfo}}
\newcommand{\sketch}{\mathsf{Sketch}}

\makeatletter
\newcommand\footnoteref[1]{\protected@xdef\@thefnmark{\ref{#1}}\@footnotemark}
\makeatother

\begin{document}

\begin{titlepage}

\title{A Nearly Time-Optimal Distributed Approximation of \\ Minimum Cost $k$-Edge-Connected Spanning Subgraph}
\author{Michal Dory\\ University of Haifa \\ {mdory@ds.haifa.ac.il} \\ \and Mohsen Ghaffari \\ MIT \\ {ghaffari@mit.edu}}
%\author{}
\date{}
\maketitle

\begin{abstract}
The minimum-cost $k$-edge-connected spanning subgraph ($k$-ECSS) problem is a generalization and strengthening of the well-studied minimum-cost spanning tree (MST) problem. While the round complexity of distributedly computing the latter has been well-understood, the former remains mostly open, especially as soon as $k\geq 3$. 

In this paper, we present the first distributed algorithm that computes an approximation of %the minimum-cost $k$-edge-connected spanning subgraph ($k$-ECSS) 
$k$-ECSS in sublinear time for general $k$. Concretely, we describe a randomized distributed algorithm that, in $\tilde{O}(k(D+k\sqrt{n}))$ rounds, computes a $k$-edge-connected spanning subgraph whose cost is within an $O(\log n\log k)$ factor of optimal. Here, $n$ and $D$ denote the number of vertices and diameter of the graph, respectively. This time complexity is nearly optimal for any $k=\poly(\log n)$, almost matching an $\tilde{\Omega}(D+\sqrt{n/k})$ lower bound.  Our algorithm is the first to achieve a sublinear round complexity for $k\geq 3$. We note that this case is considerably more challenging than the well-studied and well-understood $k=1$ case---better known as MST---and the closely related $k=2$ case. 

Our algorithm is based on reducing the $k$-ECSS problem to $k$ set cover instances, in which we gradually augment the connectivity of the spanning subgraph. To solve each set cover instance, we combine new structural observations on minimum cuts with graph sketching ideas. One key ingredient in our algorithm is a novel structural lemma that allows us to compress the information about \emph{all minimum cuts} in a graph into a succinct representation, which is computed in a decentralized fashion. We hope that this succinct representation may find applications in other computational settings or for other problems.
\end{abstract}
\thispagestyle{empty}
\newpage

\tableofcontents
\thispagestyle{empty}
\newpage

\end{titlepage}

\section{Introduction}
%\subsection{Background}  
The minimum cost $k$-edge connected spanning subgraph problem is one of the most central problems in the area of network design~\cite{jain2001factor, khuller1994biconnectivity,cheriyan2000approximating, gabow2009approximating,gabow2012iterated,goemans1994improved,chalermsook2022approximating}. Consider a network of computers with a prescribed list of possible connection links, where using each particular link incurs a certain cost---e.g., a monetary cost that should be paid to the Internet service provider. In the area of \textit{network design}, the task is to choose a set of links (e.g., which will be purchased from the provider) that minimizes our cost while meeting certain (connectivity) requirements. For instance, we might want that the chosen links form a connected spanning subgraph; this is exactly the well-known minimum-cost spanning tree (MST) problem. Or we might want a stronger subgraph, which remains connected even if a small number of edges fail. The latter is known as the \emph{minimum-cost $k$-edge-connected spanning subgraph} ($k$-ECSS) problem and is the focus of this paper.  More precisely, in the $k$-ECSS problem, the objective is to find a minimum cost $k$-edge-connected spanning subgraph $H$ of the given graph $G$. As standard, we say that a graph $H$ is $k$-edge-connected if it remains connected after the removal of any $k-1$ edges. Notice that MST is a special case of $k$-ECSS where $k=1$.

\paragraph{Centralized Algorithms.} Centralized approximation algorithms for $k$-ECSS are based on various techniques including a reduction of the problem to a matroid intersection problem in directed graphs \cite{khuller1994biconnectivity}, the primal-dual algorithm of Goemans et al. \cite{goemans1994improved}, and the iterative rounding algorithm of Jain \cite{jain2001factor}. At best these algorithms obtain a 2-approximation for weighted $k$-ECSS \cite{jain2001factor, khuller1994biconnectivity}. We note that these centralized algorithms are somewhat slow:  The fastest 2-approximation algorithm~\cite{khuller1994biconnectivity} takes $O(mnk)$ time. A recent work of Chalermsook et al.~\cite{chalermsook2022approximating} provides a $2+\epsilon$ approximation, for constant $\epsilon>0$, in $\tilde{O}(m+k^2 n^{1.5})$-time. More importantly for this paper, none of these algorithms is suitable for the distributed setting.

\paragraph{Distributed Algorithms.} Since a primary motivation for the $k$-ECSS problem comes from the reliability of networks--- e.g., with the goal of building a low-cost (communication) backbone resistant to failures---it is imperative to study the problem also from the perspective of distributed algorithms that allow the network of computers to autonomously compute the solution to the $k$-ECSS problem. We focus on the distributed \congest model: The communication network is equal to the input graph $G$, and vertices communicate by exchanging $O(\log{n})$ bit messages with their neighbors in synchronous rounds. The input and output are local. In the beginning, each vertex only knows the weights of the edges adjacent to it, and at the end, it should know which of the edges adjacent to it are in the spanning subgraph constructed. The main complexity measure is the number of communication rounds. 

While there has been some recent progress in understanding the complexity of $k$-ECSS in the \congest model, currently there are efficient algorithms only for very special cases of the problem, as we overview next. These special cases are significantly easier from a technical perspective. 

\begin{itemize}
    \item\textbf{Unweighted Graphs:} First, if one wants to solve the \emph{unweighted} version of the $k$-ECSS problem, a very simple algorithm based on iteratively constructing spanners gives an $O(1)$-approximation in $O(k \log^{1+o(1)}{n})$ time \cite{un_kECSS}. This was recently improved to a $(2+\epsilon)$-approximation in $\poly(\log n)$ rounds by Bezdirghin et al~\cite{bezdrighin2022deterministic}. Unfortunately, these approaches do not extend to weighted graphs at all, because in the weighted case even adding one redundant edge can be too expensive and might destroy the approximation guarantee. 
    
    \item\textbf{Weighted Graphs with $k=1$:} The case $k=1$ in weighted graphs is just the minimum weight spanning tree problem (MST), which is a widely studied and well-understood problem in distributed computing. See, e.g., \cite{gallager1983distributed,garay1998sublinear,kutten1998fast,elkin2006unconditional, DBLP:conf/podc/Elkin17, pandurangan2017time,fischer2021distributed}. It is well-known that the MST problem can be solved in $O(D + \sqrt{n} \log^*{n})$ time \cite{kutten1998fast} for graphs with $n$ vertices and diameter $D$, and that this time complexity is nearly tight, essentially matching an $\Omega(D + \sqrt{\frac{n}{\log{n}}})$ lower bound \cite{peleg2000near, elkin2006unconditional, sarma2012distributed}. 
    \item\textbf{Weighted Graphs with $k=2$:} For the case $k=2$, a recent line of research \cite{censor2020fast,dory2018distributed,dory2019improved} leads to a $(5+\epsilon)$-approximation in $\tilde{O}(D+\sqrt{n})$ time \cite{dory2019improved}, nearly matching the complexity of the MST problem. A 3-approximation can be obtained in $O(n)$ time \cite{censor2020fast}.
    
    \item\textbf{Weighted Graphs with $k\geq 3$:} For larger values of $k$, the fastest algorithms require at least \emph{linear} time. The first algorithm for the problem is an $O(knD)$-round $O(\log{k})$-approximation algorithm \cite{shadeh2009distributed} based on the primal-dual algorithm of Goemans et al. \cite{goemans1994improved}. A later result shows an $O(k \log{n})$-approximation in $\tilde{O}(kn)$ rounds \cite{dory2018distributed}. This complexity remains at least linear and also it is still way above the $\tilde{O}(D+\sqrt{n})$ complexity obtained for $k \leq 2$. 
\end{itemize}
Despite the above progress, the general question (general $k$ and arbitrary weights) remains mostly open. In particular, the following remains unanswered:

\begin{question} \label{question_kECSS}
Is it possible to solve the $k$-ECSS problem in sublinear time for $k \geq 3$?
\end{question}

As we will discuss later in-depth, the core of the technical challenge in answering this question lies in the following issue: in solving the $k$-ECSS problem, we somehow have to deal with \textit{all} the minimum size cuts of the hitherto selected subgraph and make sure to include edges so that no cut of size  $k'=k-1$ remains. The complexity of this task changes significantly once $k\geq 3$. The intuitive reason is that when the minimum cut size (in the subgraph selected so far) is $k'=k-1\geq 2$, the minimum cuts can have much more complex structures and overlaps, in comparison with the special cases where $k'\in\{0, 1\}$. In those two special cases, there are fairly simple characterizations of all min cuts, which open the road for efficient algorithms: In particular, when $k'=0$, minimum cuts are exactly defined by the connected components. When $k'=1$, if we fix any spanning tree $T$, minimum-cuts are exactly those edges $e$ of $T$ for which there is no other edge between the two components of $T\setminus \{e\}$. When $k\geq 3$ and thus $k'=k-1\geq 2$, minimum cuts can have much more complicated structures. We will discuss this case in \Cref{subsec:technique}.

\subsection{Our Results}

In this paper, we answer \Cref{question_kECSS} in the affirmative, proving the following.

\begin{restatable}{theorem}{mainThm} \label{thm:main}
There is a randomized distributed algorithm in the \congest model that computes an $O(\log{k}\log{n})$-approximation of the minimum cost $k$-edge-connected spanning subgraph problem in $\tilde{O}(k(D+k\sqrt{n}))$ time, with high probability.\footnote{As standard, the phrase \lq\lq with high probability\rq\rq\, (w.h.p) indicates probability at least $1-\frac{1}{n^c}$ for a constant $c\geq 1$.}
\end{restatable}

\paragraph{Remark on near-optimality  of time complexity.} For $k=\poly(\log n)$, we get a complexity of $\tilde{O}(D+\sqrt{n})$, nearly matching the complexity of previous algorithms for $k \leq 2$. The complexity is nearly tight due to a lower bound we show in \Cref{sec:lowerBounds}. Concretely, by extending lower bounds from \cite{sarma2012distributed,censor2020fast,DBLP:conf/wdag/GhaffariK13}, we prove that obtaining any polynomial approximation requires $\tilde{\Omega}(D+\sqrt{n})$ rounds for graphs that allow parallel edges, and $\tilde{\Omega}(D+\sqrt{n/k})$ rounds for simple graphs. Moreover, our algorithm is the first sublinear time algorithm for a wide range of parameters $k$.

\paragraph{Remark on approximation factor.} We remark that while there are centralized algorithms that obtain a 2-approximation for $k$-ECSS, obtaining such an approximation for non-centralized algorithms seems a highly challenging task. All the previous distributed approximation  algorithms for general $k$-ECSS had higher approximation ratios of $O(k\log{n})$ \cite{dory2018distributed} or $O(\log{k})$ \cite{shadeh2009distributed} (with a much higher $O(knd)$ time complexity). Such high approximation have been studied in other modern computation settings, e.g., in the online centralized setting, where an $O(k \log^3{n})$-approximation was obtained by Gupta et al.~\cite{gupta2012online}.
For a comparison with prior work, see \Cref{table_results}.

\begin{table}[h!]
\centering
\begin{tabular}{ |p{5.5cm}|p{2cm}|p{2.5cm}|p{3cm}|  }
 \hline
 %\multicolumn{4}{|c|}{Algorithms for weighted $k$-ECSS} \\
 %\hline
 Reference & Variant & Approximation & Time complexity\\
 \hline
  Kutten and Peleg \cite{kutten1998fast} & $k=1$ & exact & $\tilde{O}(D+\sqrt{n})$\\
  Censor-Hillel and Dory \cite{censor2020fast} & $k=2$  &  $3$ &  $O(n)$\\
  Dory and Ghaffari \cite{dory2019improved} & $k=2$  &  $5+\epsilon$ &  $\tilde{O}(D+\sqrt{n})$\\
  Shadeh \cite{shadeh2009distributed} & general $k$ & $O(\log{k})$ & $O(knD)$\\
  Dory \cite{dory2018distributed} & general $k$ & $O(k\log{n})$ & $\tilde{O}(kn)$\\
  \textbf{Here} & general $k$ & $O(\log{k}\log{n})$ & $\tilde{O}(k(D+k\sqrt{n}))$\\
 \hline
\end{tabular}
 \caption{\small Algorithms for weighted $k$-ECSS. Note that for $k\geq 3$, prior to this work, no sublinear time algorithm was known.  Also, we believe that cases with small values of $k$---e.g., constant or at most polylogarithmic---are arguably the most relevant cases, from a practical perspective.}
\label{table_results}
\end{table}

\paragraph{Remark on bicriteria approximation for larger $k$.} We also note that by using a simple combination of \Cref{thm:main} with Karger's well-known uniform edge sampling results~\cite{karger1993global}, we can obtain the following bicriteria algorithm, which has a better round complexity and approximation guarantee for larger $k$. For the proof, see \Cref{sec:puttingEverything}.

\begin{restatable}{corollary}{corBic}\label{crl:bicriteria}
There is a randomized distributed algorithm in the \congest model that computes a $k(1-o(1))$-edge-connected spanning subgraph whose cost is within an $O(\log n \log\log n)$ factor of the minimum cost $k$-edge-connected spanning subgraph, in $\tilde{O}(k(D+\sqrt{n}))$ time, with high probability.
\end{restatable}

\subsection{Our Techniques}\label{subsec:technique}
While there are many centralized algorithms for $k$-ECSS, they do not seem applicable in the distributed setting. Instead, the high-level outline of our algorithm is to use a simple and natural approach, based on \cite{dory2018distributed}, where we iteratively augment the connectivity of the chosen subgraph, gradually increasing connectivity from $0$ to $k$. In particular, we start with an empty subgraph $H$, and we augment its connectivity in $k$ iterations, where in iteration $i$ we augment the connectivity from $i-1$ to $i$; this is indeed the technical core of the problem. This \textit{augmentation problem} can be seen as a special case of the \emph{set cover} problem, where the goal is to \emph{cover} cuts of size $i-1$ in $H$ by a low-cost set of edges, according to the following definition. 
Here we say that an edge $e$ covers a cut $C$ in $H$ if $(H \setminus C) \cup \{e\}$ is connected. 
To obtain a good approximation, our goal is to add a low-cost set of edges that augments the connectivity. 

To decide which edges to add, we compare edges according to their \emph{cost-effectiveness}, defined as $\rho(e)=\frac{|C_e|}{w(e)}$, where $C_e$ is the set of minimum cuts that $e$ covers in $H$, and $w(e)$ is the weight of $e$. Intuitively, adding edges with large values of cost-effectiveness allows us to cover many cuts while paying a low cost. To exploit this, we use an approach similar to greedy set cover solutions, in which per step we add to the augmentation a set of edges with near-maximal cost-effectiveness.

The main challenge in obtaining a fast algorithm using this approach is to design an efficient algorithm for \textit{estimating the cost-effectiveness of edges}, and that is exactly the main technical contribution of our paper. In \cite{dory2018distributed}, it is shown that in the special case that the minimum cut size is 1, the cost-effectiveness can be computed in $\tilde{O}(D+\sqrt{n})$ time. That leads to an $\tilde{O}(D+\sqrt{n})$-round algorithm for 2-ECSS. However, dealing with cuts of size at least 2 is much more challenging, due to the possibility of more complex interactions and overlaps between these cuts, and no sublinear-time or even linear time algorithm is known for general $k$.

%In \cite{dory2018distributed}, instead of directly dealing with this challenge in a distributed setting, the author resorts to a simple $\tilde{O}(kn)$-round algorithm that broadcasts the subgraph $H$ to the whole graph. However, obtaining a sublinear algorithm is fundamentally more challenging and requires a new approach.   

\paragraph{Our Approach for Approximating Cost-Effectivenesses.} Our main technical contribution in this paper is to provide a sublinear-time algorithm for approximating cost-effectivenesses. This involves presenting a structural lemma about min-cuts and providing a novel succinct representation of them, which can also be computed efficiently using a relatively simple and local scheme. We hope that these properties may find applications in other problems and/or computational settings.

Intuitively, to compute the cost-effectiveness, we may want to be able to learn about all minimum cuts in the graph $H$, and then let each edge $e$ check which ones of them it covers. However, a direct implementation of this approach seems quite expensive. We do something different. Our algorithm has two main parts. 
\begin{enumerate}
    \item \textbf{Succinct Representation of Minimum Cuts.} First, we show a structural lemma, which allows us to compress efficiently information about minimum cuts. Using it, each edge learns a small amount of information that represents the minimum cuts it participates in.
    \item \textbf{Estimating the Cost-Effectiveness via Graph Sketches.} Second, we show that we can use the information computed above combined with graph sketches to estimate for each edge how many cuts it covers.
\end{enumerate}
In \Cref{sec:overview}, we explain in detail the high-level ideas behind these ingredients.

\subsection{Additional Related Work}

\noindent\textbf{Minimum Cut.} The minimum cut problem is a well-studied problem. In the distributed \congest model, a rich line of work \cite{pritchard2011fast,DBLP:conf/wdag/GhaffariK13,NanongkaiS14,daga2019distributed,ghaffari2020faster,un_kECSS,DBLP:journals/corr/abs-2004-09129,ghaffari2022universally}, led to a near-optimal algorithm taking $\tilde{O}(D+\sqrt{n})$ time for computing the minimum cut in a weighted graph \cite{DBLP:journals/corr/abs-2004-09129}. We remark that distributed min cut algorithms usually focus on finding one minimum cut, while in our work here we had to work with all the minimum cuts in a graph. For this reason, we believe that our approach can find future applications for other problems related to minimum cuts, for example for estimating the number of minimum cuts in a graph.\\[-7pt]

\noindent\textbf{Connectivity Certificates.} Another related problem studied in the distributed setting is the problem of finding sparse connectivity certificates \cite{un_kECSS,bezdrighin2022deterministic,daga2019distributed}. A sparse connectivity certificate of a graph $G$, is a subgraph $H$ of size $O(kn)$ that is $k$-edge-connected if and only if $G$ is $k$-edge-connected. Sparse connectivity certificates can be used to obtain $O(1)$-approximation for the \emph{unweighted} version of $k$-ECSS. In particular, the construction in \cite{bezdrighin2022deterministic} gives a $(2+\epsilon)$-approximation for unweighted $k$-ECSS in $\poly(\log n)$ rounds. Obtaining this approximation in \emph{unweighted} graphs is based on showing that the subgraph built is sparse, and combine it with the fact that any $k$-edge-connected subgraph has at least $kn/2$ edges.
One common approach to build connectivity certificates is based on iteratively building spanning trees or spanners.  
While this approach can also be implemented in weighted graphs, unfortunately it does not lead to any meaningful approximation for weighted $k$-ECSS, the reason is that in the weighted case even adding one redundant edge can be too expensive, and it is not enough to count the number of edges in the subgraph constructed anymore.\\[-7pt]

\noindent\textbf{Fault-Tolerant Structures.} In the $k$-ECSS problem the goal is to build a low-cost subgraph that remains connected after edge failures. This is closely related to fault-tolerant graph structures, sparse graphs that maintain a certain functionality (for example, contain an MST or a spanner) after failures. Distributed algorithms for building fault-tolerant BFS and MST structures that are resilient to one edge failure appear in \cite{ghaffari2016near}, where distributed fault-tolerant BFS trees and distance preserves resilient to two edge failures appear in \cite{parter2020distributed}. Dealing with more failures is an open question.
Distributed algorithms for fault-tolerant spanners appear in \cite{dinitz2011fault,dinitz2020efficient,parter2022nearly}. 

\section{High-Level Description of the Algorithm} \label{sec:overview}

As explained in Section \ref{subsec:technique}, to solve $k$-ECSS, we follow a distributed greedy approach described in \cite{dory2018distributed}, that reduces the $k$-ECSS problem to the problem of estimating cost-effectiveness of edges. A detailed discussion of this reduction appears in \Cref{sec:setcover}. 
Our main technical contribution is providing a method that computes the cost-effectiveness of all edges efficiently. 

\begin{restatable}{theorem}{CostEfThm}[Cost-Effectiveness Approximation] \label{thm_cost_ef}
Let $H$ be an $i$-edge-connected spanning subgraph of a graph $G$. For an edge $e \not \in H$ define $C_e$ as the set of cuts of size $i$ in $H$ covered by $e$. Then in $\tilde{O}(D+i\sqrt{n})$ time each edge $e \not \in H$ learns an $O(1)$-approximation of the value $|C_e|$. In particular, edge $e$ learns an $O(1)$-approximation of its cost-effectiveness with respect to $H$, defined by $\rho(e)=\frac{|C_e|}{w(e)}.$ The algorithm works with high probability.
\end{restatable}

\paragraph{Reduction to 2-respecting cuts.} To estimate the cost-effectiveness, we should estimate for each edge the number of cuts it covers. Our first step is to simplify the structure of the problem. We use Karger's well-known reduction~\cite{karger2000minimum}, which allows us to work with only \emph{2-respecting cuts} with respect to a certain tree, instead of working with arbitrary minimum cuts. We say that a cut $C$ \emph{2-respects} a spanning tree $T$ if $|C \cap T| \leq 2$. 
Based on \cite{karger2000minimum}, we show in \Cref{sec:preliminaries} that in order to estimate the cost-effectiveness of edges it is enough to estimate for each edge the number of min 2-respecting cuts it covers in a certain spanning tree $T$. A min 2-respecting cut is a minimum cut in $G$ that has at most 2 edges in $T$. For full details and proofs see \Cref{sec:preliminaries}.
From now on we fix a spanning tree $T$ of a graph $H$ and focus on estimating the number of min 2-respecting cuts each edge $e \not \in T$ covers. %We denote by $i$ the size of the minimum cut in $H$. 
Note that any min 2-respecting cut has (at most) 2 tree edges $t,t'$, we identify the cut with these edges. A cut that has only one tree edge $t$ is called a 1-respecting cut. 

\subsection{Estimating the Cost-Effectiveness when $T$ is a Path}

To simplify the presentation, we start by focusing on the special case that $T$ is a path, which already allows us to present many of the key ingredients of the algorithm. 
For a non-tree edge $e=\{u,v\}$, we denote by $P_e$ the tree path between $u$ and $v$.
We start with the following simple observation that we prove in \Cref{sec:preliminaries}. See \Cref{pathCoverPic} for an illustration. 

\begin{restatable}{claim}{twoRespectCover} \label{claim_cut_cover}
A non-tree edge $e$ covers a 2-respecting cut $\{t,t'\}$ iff exactly one of $t,t'$ is in $P_e$.
\end{restatable}

We next give an intuition for the proof (for the case that $T$ is a path). First note that if $\{t,t'\}$ is a minimum 2-respecting cut, that represents a minimum cut $C$ in a graph $H$, then $\{t,t'\}$ are the only edges of $C$ in $T$, and $H \setminus C$ has exactly two connected components. We can identify them as follows. Removing $\{t,t'\}$ from $T$ disconnects $T$ to 3 sub-paths $P_1,P_2,P_3$, where the 2-respecting cut defined by $\{t,t'\}$ has 2 sides (as this is a minimum cut), one has the middle path $P_2$ and the other has the two other paths $P_1 \cup P_3$. The reason is that $\{t,t'\}$ are edges that cross the cut, and they already have one endpoint in $P_2$, hence their other endpoints are in the other side of the cut. Now the edges that cross the cut are exactly those that have one endpoint in $P_2$ and the other in $P_1 \cup P_3$, which are edges $e$ where exactly one of $\{t,t'\}$ is in $P_e$. The proof also generalizes to the case that $T$ is a general spanning tree (see \Cref{sec:preliminaries}).

\setlength{\intextsep}{0pt}
\begin{figure}[h]
\centering
\setlength{\abovecaptionskip}{0pt}
\setlength{\belowcaptionskip}{8pt}
\includegraphics[scale=0.5]{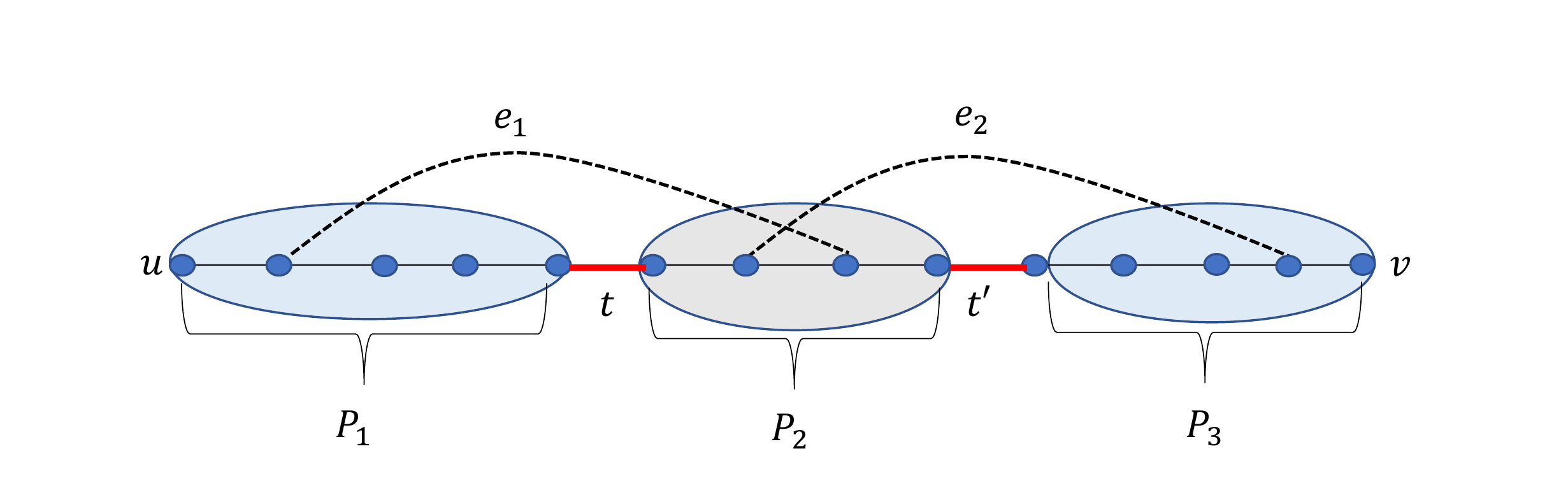}
 \caption{The tree $T$ is the path from $u$ to $v$. The edges $e_1,e_2$ are examples of non-tree edges that cover the 2-respecting cut $\{t,t'\}$. This 2-respecting cut is is the blue-grey cut in the figure, where one side is $P_1 \cup P_3$, and the other side is $P_2$.}
\label{pathCoverPic}
\end{figure}

As we show in \Cref{sec:preliminaries}, if the endpoints of $e$ already know about the min 2-respecting cut $\{t,t'\}$, then using \Cref{claim_cut_cover} and the tool of \emph{lowest common ancestor (LCA) labels} (see \Cref{sec:preliminaries}), we can check if $e$ covers the min 2-respecting cut $\{t,t'\}$. Basically we just need to compare the endpoints of $e,t,t'$ and identify if exactly one of $\{t,t'\}$ is in $P_e$, which can be done using LCA labels. Hence, a naive approach for estimating the cost-effectiveness could be first to learn about all the min 2-respecting cuts $\{t,t'\}$ in the graph, and then for each one of them to check if $e$ covers $\{t,t'\}$. 

%\paragraph{Estimating the Cost-Effectiveness: First Try.}
%It is easy to estimate the number of min 2-respecting cuts covered by $e$ if the whole graph knows the identity of all the min 2-respecting cuts. To do so, we can use the tool of \emph{lowest common ancestor (LCA) labels} (see \Cref{sec:preliminaries}) to identify for each edge which cuts it covers.

%We show later that if $e$ knows the edges $t,t'$ that define a min 2-respecting cut, there is a simple way to check if $e$ covers the cut. For this, we exploit the tool of \emph{lowest common ancestor (LCA) labels}.
%Hence, if we have a way to find all the min 2-respecting cuts in the graph and to let the whole graph learn about them, it would lead to a very simple algorithm for estimating the cost-effectiveness of edges. Each edge $e$ just checks which of the min 2-respecting cuts it covers. 

There are however two problems with this approach. It requires \emph{finding} all min 2-respecting cuts in the graph, and \emph{broadcasting} this information to the whole graph. It is not clear if it is possible to find  all min 2-respecting cuts efficiently. Very recently, it was shown that it is possible to find \emph{one} min 2-respecting cut in $\tilde{O}(D+\sqrt{n})$ time. This led to a near-optimal distributed min cut algorithm \cite{DBLP:journals/corr/abs-2004-09129}. However, their algorithm crucially relies on the fact that only one minimum cut should be found and it does not explore all possible options for min cuts. Moreover, even if there was a distributed algorithm that finds all possible cuts, the number of minimum cuts in the graph could be $\Omega(n^2)$, and hence broadcasting information about all min 2-respecting cuts to the whole graph cannot lead to an efficient algorithm.

We next describe our general approach to overcome the above problems. One of the key insights of this paper (which we hope may find applications in other computation settings and/or for other problems) is this: Instead of learning about all min 2-respecting cuts and broadcasting this information to the whole graph, we show that each tree edge $t$ can learn a \emph{compressed} description of the min 2-respecting cuts of the form $\{t,t'\}$. More concretely, we show that if the minimum cut size is $k$, then we can represent in $O(k \log{n})$ bits a compressed information of all minimum 2-respecting cuts $\{t,t'\}$ that a tree edge $t$ appears in. Second, we show that all tree edges $t$ can learn their compressed information efficiently, and we show how to use it to evaluate the cost-effectiveness of non-tree edges. We next provide more details on these ingredients.  

\subsubsection{Succinct Representation of Min Cuts}

In \Cref{sec:succinct_min_cuts}, we show that while there could be many minimum 2-respecting cuts, there is a succinct way to represent them. To get an intuition, think about a simple case that the graph $G$ is just a cycle, and the spanning tree $T$ is a path. Here any 2 edges in the graph are a cut, so listing all the cuts in the graph would take $\Omega(n^2)$ space. However, if we look at some tree edge $t$, there is a very simple way to represent all the 2-respecting cuts it participates in, these are just all the edges in $T$. We prove a structural lemma that generalizes this simple observation to a general tree with minimum cut of size $k$. 
Concretely, we show that each tree edge $t$ can learn $O(k \log{n})$ bits that represent all the minimum 2-respecting cuts of the from $\{t,t'\}$. We call this information $\cutInfo(t)$. This information contains a description of $O(k)$ subpaths in the tree, where in each one of them there is a simple criterion to check if a tree edge $t'$ is in a min 2-respecting cut with $t$. We emphasize that $t$ does not learn about all the min 2-respecting cuts $\{t,t'\}$, as this could be a lot of information. Instead, given $\cutInfo(t)$ and a tree edge $t'$, it is possible to deduce whether $\{t,t'\}$ is a min 2-respecting cut or not.
Let us fix some notation. We denote by $\cut(t,t')$ the size of the 2-respecting cut defined by the tree edges $t,t'$ (i.e., the number of edges that cross the cut). We denote by $\cov(t)$ the number of edges that \emph{cover} a tree edge $t$. Here we say that an edge $e=\{u,v\}$ covers a tree edge $t$ if $t$ is in the unique tree path $P_e$ between $u$ and $v$. In particular, $t$ covers $t$. We denote by $\cov(t,t')$ the number of edges that cover both $t$ and $t'$. From \Cref{claim_cut_cover}, we get that $\cut(t,t')=\cov(t)+\cov(t')-2\cov(t,t)$. This follows from the fact that the edges that cross (or cover) the cut are exactly the ones that cover exactly one of $t,t'$. We show the following.   

\begin{restatable}{lemma}{succinctCutsLemma}[\textbf{Succinct Representation of Min Cuts}] \label{min_cuts_lemma}
Let $G=(V,E)$ be a graph with minimum cut size $k$, and let $T$ be a spanning tree of $G$. Fix a tree edge $t \in T$. Then there exist $\ell = O(k)$ tuples $\{(S^t_i,c^t_i)\}_{1 \leq i \leq \ell}$ where each $S^t_i$ is a subpath in the tree, such that the edges $t'$ where $\cut(t,t')=k$ are exactly all edges in $E_t = \cup_{1 \leq i \leq \ell} \{ t' \in S^t_i | \cov(t') = c^t_i \}$.
\end{restatable}

Let us provide an intuition of the proof of this Succinct Representation Lemma when $T$ is a path (see \Cref{pathLemmaPicPrelim}). In this case, if we fix a tree edge $t$, any edge that covers $t$ has one endpoint to the left of $t$ and one endpoint to the right of $t$. We focus on 2-respecting cuts $\{t,t'\}$ where $t'$ is on the right of $t$. We sort all the edges that cover $t$ according to their right endpoint, from the closet to the furthest from $t$. Denote by $e_1,...,e_k$ the $k$ first edges in this sorted order (in particular, $e_1 = t$). Let $v_1,v_2,...,v_k$ be the right endpoints of $e_1,...,e_k$, respectively. The first observation is that $\cut(t,t')=k$ for an edge $t'$ to the right of $t$ only if $t'$ is on the subpath between $v_1$ to $v_k$, as otherwise, there are more than $k$ edges: $\{e_1=t,e_2,...,e_k,t'\}$ that cross the cut defined by $\{t,t'\}$. Moreover, we can show that all the edges $t'$ on the subpath $S_i$ between $v_i$ to $v_{i+1}$ are in a min 2-respecting cut with $t$ if and only if $\cov(t')=c^t_i = \cov(t)+k-2i$. The reason is as follows. First, if $\{t,t'\}$ is a min 2-respecting cut, then $\cut(t,t')=k$ (as the minimum cut has size $k$). Second, as explained above $\cut(t,t')=\cov(t)+\cov(t')-2\cov(t,t')$. Finally, by the definition of the subpath $S_i$, we get that for all edges $t' \in S_i$, $\cov(t,t')=\cov(t)-i$. The latter follows from the fact that the edges that cover $t$ and $t'$ are all the edges that cover $t$ expect of the $i$ edges $e_1,...,e_i$. This follows from the definition of $S_i$. From the above, a simple calculation shows that $t' \in S_i$ is in a min 2-respecting cut with $t$ if and only if $\cov(t') = \cov(t)+k-2i$. Hence we can divide the path between $v_1$ to $v_k$ to $k-1$ subpaths $S^t_i=[v_i,v_{i+1}]$, where in each one of them the edges $t'$ in a min 2-respecting cut with $t$ are exactly those where $\cov(t')=c^t_i$. We can use the same argument to represent cuts where $t'$ is on the left of $t$. Based on these ideas, we prove \Cref{min_cuts_lemma} for a path $T$. 

\setlength{\intextsep}{0pt}
\begin{figure}[h]
\centering
\setlength{\abovecaptionskip}{-4pt}
\setlength{\belowcaptionskip}{4pt}
\includegraphics[scale=0.4]{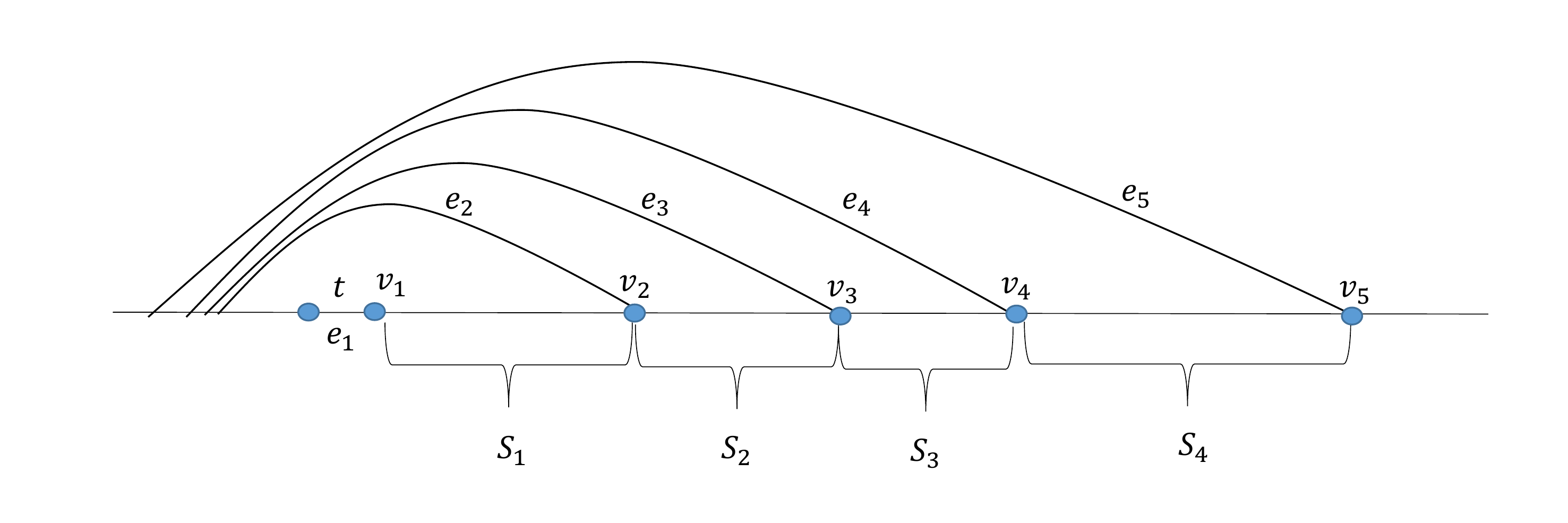}
 \caption{Illustration of the proof of \Cref{min_cuts_lemma}.}
\label{pathLemmaPicPrelim}
\end{figure}

\paragraph{Learning the Succinct Representation.} 
To allow all tree edges to learn the $O(k)$ tuples $\{(S^t_i,c^t_i)\}_{1 \leq i \leq \ell}$ that represent their cuts, the main task is to learn $O(k)$ edges that cover $t$ and are the \emph{best} according to a certain criterion. To explain the idea, assume first that each tree edge $t$ should learn only about one non-tree edge $e_t$, the non-tree edge that covers $t$ and has the right endpoint closest to $t$. To do so, we can break the tree $T$ to $O(\sqrt{n})$ subpaths of length $O(\sqrt{n})$. Let $P_t$ be the subpath that contains $t$. If the edge $e_t$ has at least one endpoint in $P_t$, then using communication inside $P_t$, the edge $t$ can learn about $e_t$. Actually all the tree edges $t$ in $P_t$ can learn about their edges $e_t$ in this case in $O(\sqrt{n})$ time, using the fact that $P_t$ has length $O(\sqrt{n})$ and the congestion is also bounded by $O(\sqrt{n})$ because there are just $O(\sqrt{n})$ tree edges in $P_t$. This can be done in all the $O(\sqrt{n})$ subpaths of $T$ simultaneously.  Another case is that the edge $e_t$ does not have any endpoint in $P_t$. Here the crucial observation is that the edge $e_t$ in this case does not depend on the specific choice of $t$, it is an edge that covers all tree edges in $P_t$ and has the right endpoint closest to $P_t$. There is only one such edge for the whole subpath $P_t$, and we can use global communication over a BFS tree to learn about this edge in $O(D)$ time. Since we have $O(\sqrt{n})$ subpaths $P_t$, then learning about the relevant edges for all subpaths takes $O(D+\sqrt{n})$ time. Following the same ideas, all the tree edges can learn about $O(k)$ best edges that cover them in $O(D+k\sqrt{n})$ time. Dealing with an arbitrary tree $T$ is more challenging and will be discussed later. 

\subsubsection{Estimating the Cost-Effectiveness via Graph Sketches}

We next explain how we use the cut information learned by edges to estimate the cost-effectiveness. To simplify the description, we first focus on a setting that each tree edge $t$ knows the full list of min 2-respecting cuts $\{t,t'\}$ it participates in. 
We focus on estimating the number of 2-respecting cuts that are covered by $e$ and have exactly two tree edges; handling 1-respecting cuts is easier. 

Recall that by \Cref{claim_cut_cover}, a non-tree edge $e$ covers a 2-respecting cut $\{t,t'\}$ iff exactly one of $t,t'$ is in $P_e$.
A naive algorithm to compute the cost-effectiveness using \Cref{claim_cut_cover} could be to go over all the tree edges $t$ in the path $P_e$, learn about all the min 2-respecting cuts $\{t,t'\}$ they participate in, and count in how many of them the second edge $t'$ in the cut is not in $P_e$. However, this requires collecting a lot of information  and is not efficient. We next show how using \emph{sketching} ideas we can estimate the number of edges that cross the cut. Graph sketches is a powerful tool that is used for example in connectivity algorithms in various settings (see for example \cite{ahn2012analyzing,kapron2013dynamic,kapralov2014spanners,king2015construction,mashreghi2021broadcast}). Many times it is used in order to find outgoing edges from a connected component. Here the property that distinguishes between an outgoing edge from a set $S$ and other edges is that the outgoing edges are exactly those that have exactly one endpoint in $S$. We show that similar type of ideas are helpful in our context. Intuitively, the reason is this: edge $e$ covers a cut $\{t,t'\}$ iff exactly one of $t $ and $t'$ is in the path $P_e$. This is similar to the property that outgoing edges have exactly one endpoint in a connected component. 

\paragraph{High-level Idea.} Each tree edge $t$ uses the information about the min 2-respecting cuts $\{t,t'\}$ to compute a \emph{sketch} $\sketch(t)$ of the min 2-respecting cuts, with the following properties. 

\begin{enumerate}
\item The sketch has poly-logarithmic size.
\item Given a tree path $P$, if we look at $\oplus_{t \in P} \sketch(t)$ we can obtain a constant approximation to the number of min 2-respecting cuts that have exactly one tree edge in $P$, with high probability. Here $\oplus$ represents the bitwise XOR of the sketches.\label{approxCuts}
\end{enumerate}

Given the sketches, we can obtain an efficient algorithm for approximating the cost-effectiveness of all non-tree edges, as follows.
Each non-tree edge $e=\{u,v\}$ should learn $\oplus_{t \in P_e} \sketch(t)$. This is an aggregate function of the values $\sketch(t)$ of edges in $P_e$. Such computations can be done using the tool of fragment decomposition as was shown in \cite{dory2018distributed}. 
From property \ref{approxCuts} of the sketches and from \Cref{claim_cut_cover}, this gives a constant approximation to the number of min 2-respecting cuts covered by $e$, as needed.

\paragraph{Estimating the Cost-Effectiveness without Full Information about Minimum Cuts.}

We next explain how to extend the above algorithm to a setting where tree edges do not know all their min 2-respecting cuts, but they have the succinct description of min 2-respecting cuts provided by \Cref{min_cuts_lemma}; we call this information $\cutInfo(t)$.
Note that for the above algorithm, tree edges do not need to know all their min 2-respecting cuts, they just need to be able to compute the values $\sketch(t)$. Our main goal is to show that this information can be computed efficiently.

The value $\sketch(t)$ is a compressed description of all the min 2-respecting cuts $\{t,t'\}$, think about it as an aggregate function of all the values $\{t,t'\}$. To be able to compute it we need that for each cut $\{t,t'\}$ either $t$ or a vertex close to $t$ will know the name of the cut $\{t,t'\}$. At a high-level, this is obtained as follows. First note that if a vertex $v$ knows $\cutInfo(t)$ and the value $\cov(t')$ for a tree edge $t'$ it can check if $\{t,t'\}$ is min 2-respecting cut (this just follows from the succinct description of cuts provided in \Cref{min_cuts_lemma}). Hence if for any tree edge $t'$ there is a vertex close by to $t$ that knows $\cutInfo(t),\cov(t')$, this vertex will know about the cut $\{t,t'\}$. %A naive way to use it is to broadcast the values $\cutInfo(t),\cov(t)$ of all tree edges to the whole graph. This allows the whole graph to learn about all min 2-respecting cuts, but requires $\Omega(kn)$ time for broadcasting the information.
Our approach is to broadcast the values $\cutInfo(t),\cov(t)$  to $O(\sqrt{n})$ vertices close to $t$, and additionally broadcast to the whole graph a small amount of information. We next provide more details on our approach. First, we again decompose the path $T$ into $O(\sqrt{n})$ subpaths of size $O(\sqrt{n})$. In each subpath, we let all the vertices of the subpath to learn the values $\{\cutInfo(t),\cov(t)\}$ of all tree edges in the subpath. As this is $O(k \log{n})$ bits per edge, this can be done in $O(k\sqrt{n})$ time. Now for any pair of subpaths $P_1,P_2$ where there is an edge $\{v,v'\}$ between them, the vertices $v$ and $v'$ can exchange between them the values $\{\cutInfo(t),\cov(t)\}$ of all tree edges in the subpaths $P_1,P_2$. By the definition of $\cutInfo(t)$, from this information $v$ and $v'$ can deduce the names of all cuts $\{t,t'\}$ where $t \in P_1, t' \in P_2$. Then for any edge $t \in P_1$, the vertex $v \in P_1$ can compute a sketch of all the cuts $\{t,t'\}$ it knows about. As sketches are aggregate functions, if different vertices in the subpath $P_1$ know about different cuts $\{t,t'\}$ learning the value $\sketch(t)$ can be done using a simple aggregate computation in the subpath. In this description we focused on the case that there is an edge between $P_1$ and $P_2$. The case there is no edge between them is slightly more delicate. Here we exploit structural properties of 2-respecting cuts to show that the min 2-respecting cuts have a very simple structure in this case, and we use a \emph{small amount} of global communication to let all tree edges learn enough information about these cuts to complete the computation of $\sketch(t)$. The full description of the algorithm appears in \Cref{sec:cost_ef}.

%We next provide more details on our approach. First, we decompose the graph into $O(\sqrt{n})$ fragments of size $O(\sqrt{n})$. In each fragment, we let all the vertices of the fragment learn the values $\{\cutInfo(t),\cov(t)\}$ of all tree edges in the fragment. As this is $O(k \log{n})$ bits per edge, this can be done in $O(k\sqrt{n})$ time. Now for any pair of fragments $F,F'$ where there is an edge $\{v,v'\}$ between them, the vertices $v$ and $v'$ can exchange between them the values $\{\cutInfo(t),\cov(t)\}$ of all tree edges in the fragments $F,F'$. By the definition of $\cutInfo(t)$, from this information $v$ and $v'$ can deduce the names of all cuts $\{t,t'\}$ where $t \in F, t' \in F'$. Then for any edge $t \in F$, the vertex $v \in F$ can compute a sketch of all the cuts $\{t,t'\}$ it knows about. As sketches are aggregate functions, if different vertices in the fragment $F$ know about different cuts $\{t,t'\}$ learning the value $\sketch(t)$ can be done using a simple aggregate computation in the fragment. In this description we focused on the case that there is an edge between $F$ and $F'$. The case there is no edge between them is slightly more delicate. Here we exploit structural properties of 2-respecting cuts to show that the min 2-respecting cuts have a very simple structure in this case, and we use a \emph{small amount} of global communication to let all tree edges learn enough information about these cuts to complete the computation of $\sketch(t)$. The full description of the algorithm appears in \Cref{sec:cost_ef}. 

\subsection{Estimating the Cost-Effectiveness for a General Spanning Tree $T$}

Most of the ideas described above also generalize to the case that the spanning tree $T$ is a general tree, but it requires additional ingredients and more complex proofs. We next discuss the main ingredients needed. 
First, Lemma \ref{min_cuts_lemma} about the succinct representation of min cuts holds also when $T$ is a general tree, but it requires a more involved proof. Note that when $T$ is a path we sorted edges according to their right endpoint, and our goal is to be able to do something similar for a general tree as well.
The high-level idea is to use an observation on the structure of 2-respecting cuts from \cite{mukhopadhyay2020weighted} that says that for each tree edge $t$, all the edges $t'$ where $\{t,t'\}$ defines a min 2-respecting cut are on one tree path $P_t$. Second, we sort the edges that cover $t$ according to the subpath they cover in $P_t$ and repeat a similar proof idea to the one we used in the case that $T$ is a path. See \Cref{sec:succinct_min_cuts} for the details.

To learn the succinct cut information, the main task again is to learn $O(k)$ edges that cover $t$ and are the \emph{best} according to a certain criterion.
Such computations can usually be done in $\tilde{O}(D+k\sqrt{n})$ time using the tool of \emph{fragment decomposition}. One challenge in our case is that when $T$ is a general tree the criterion defining the best edges is specific for each edge $t$ (as we sort the edges that cover $t$ according to a tree path $P_t$). To get an efficient algorithm, we use structural properties of cuts proved in \cite{DBLP:journals/corr/abs-2004-09129} that informally speaking imply that for \emph{many tree edges} the paths $P_t$ are the same. This allows us to bound the total amount of communication. For a more precise discussion and detailed description of the algorithm, see \Cref{sec:learn_cut}. 

Finally, to estimate the cost-effectiveness, we follow similar ideas to the ones described for the path case, where instead of breaking $T$ to $O(\sqrt{n})$ subpaths, we break the tree $T$ to $O(\sqrt{n})$ fragments of size $O(\sqrt{n})$ and distribute the cut information inside these fragments. As in the path case, it is easier to handle cuts $\{t,t'\}$ if there is an edge between the fragments of $t$ and $t'$. In the case that there is no such edge, we prove that the minimum cuts have a very special structure and we can broadcast a small amount of information to learn about these cuts. For full details see \Cref{sec:cost_ef}. 

\paragraph{Conclusion.} 
Based on the above ideas, we obtain an $O(1)$-approximation for the cost-effectiveness of edges in $\tilde{O}(D+k\sqrt{n})$ time. As we discuss in \Cref{sec:setcover}, this leads to an $O(\log{k}\log{n})$-approximation for $k$-ECSS in $\tilde{O}(k(D+k\sqrt{n}))$ time.
Intuitively, the $O(\log{n})$ term in the approximation comes from using a greedy approach to solve each one of the connectivity augmentation problems, as we add to the solution edges with near-optimal cost-effectiveness. The $O(\log{k})$ term comes from working in $k$ phases where we gradually augment the connectivity. The time complexity is based on showing that if we can estimate the cost-effectiveness in $\tilde{O}(D+k\sqrt{n})$ time, then we can augment the connectivity of $H$ from $i$ to $i+1$ in just poly-logarithmic number of iterations, where in each one of them the cost-effectiveness is computed. Summing up over all $k$ phases leads to $\tilde{O}(k(D+k\sqrt{n}))$ time complexity. For a detailed discussion and a proof see \Cref{sec:setcover}.

\section{Discussion}

In this paper, we provide a randomized algorithm with a near-optimal time complexity for the minimum cost $k$-edge-connected spanning subgraph problem for any constant $k$. Several interesting questions remain open. First, our algorithm obtains an $O(\log{k}\log{n})$-approximation, and a natural question is whether it is possible to obtain a constant approximation in a similar time complexity. While there are several centralized polynomial time algorithms that obtain a constant approximation to the problem \cite{goemans1994improved,jain2001factor,khuller1994biconnectivity,chalermsook2022approximating}, it seems challenging to implement them efficiently in a distributed setting (for example, a distributed implementation of the algorithm of Goemans et al. \cite{goemans1994improved} led to an $O(\log{k})$-approximation algorithm with complexity of $O(knD)$ rounds \cite{shadeh2009distributed}). Hence, obtaining an efficient algorithm with a constant approximation in the distributed setting seems to be a highly challenging task that will require a new approach for the problem. Another interesting question is whether it is possible to obtain a deterministic algorithm with similar guarantees.

\section*{Roadmap} The rest of the paper is organized as follows. \Cref{sec:preliminaries} discusses basic tools and claims used in our algorithm. \Cref{sec:setcover} describes the distributed greedy algorithm for augmenting the connectivity, that reduces the $k$-ECSS problem to the problem of estimating cost-effectiveness. The rest of the paper focuses on approximating the cost-effectiveness of edges. First, in \Cref{sec:succinct_min_cuts}, we prove \Cref{min_cuts_lemma} about the succinct representation of minimum cuts. In \Cref{sec:learn_cut}, we show how all tree edges learn the values $\cutInfo(t)$ and additional information useful for the algorithm. Finally, in  \Cref{sec:cost_ef}, we describe our algorithm for approximating the cost-effectiveness. Our lower bounds appear in \Cref{sec:lowerBounds}.

\section{Preliminaries} \label{sec:preliminaries}

\subsection{Cover Values and 2-respecting Cuts} \label{pre_cov_cut}

Given a spanning tree $T$, we say that an edge $e=\{u,v\}$ \emph{covers} a tree edge $t$ if $t$ in the unique tree path between $u$ and $v$.
We denote by $\cov(t)$ the number of edges that cover $t$.
Note that the edges that cover $t$ are exactly the edges that cross the 1-respecting cut defined by $t$ (see Figure \ref{covPic}). 
It follows that $\cut(t)=\cov(t)$, where $\cut(t)$ is the size of the 1-respecting cut defined by $t$. 

We next discuss 2-respecting cuts.
We denote by $\cov(t_1,t_2)$ the number of edges that cover both $t_1$ and $t_2$, and by $\cut(t_1,t_2)$ the value of the 2-respecting cut defined by $t_1,t_2$. The following is shown in \cite{DBLP:journals/corr/abs-2004-09129}, the proof is based on showing that the edges that cross the 2-respecting cut defined by $t_1,t_2$ are exactly the edges that cover exactly one of $t_1,t_2$. See Figure \ref{covPic} for illustration.

\begin{claim}[\cite{DBLP:journals/corr/abs-2004-09129}] \label{covClaim}
Let $t,t'$ be tree edges, then $\cut(t,t')=\cov(t)+\cov(t') - 2 \cov(t,t')$.
\end{claim}

\setlength{\intextsep}{0pt}
\begin{figure}[h]
\centering
\setlength{\abovecaptionskip}{-4pt}
\setlength{\belowcaptionskip}{4pt}
\includegraphics[scale=0.6]{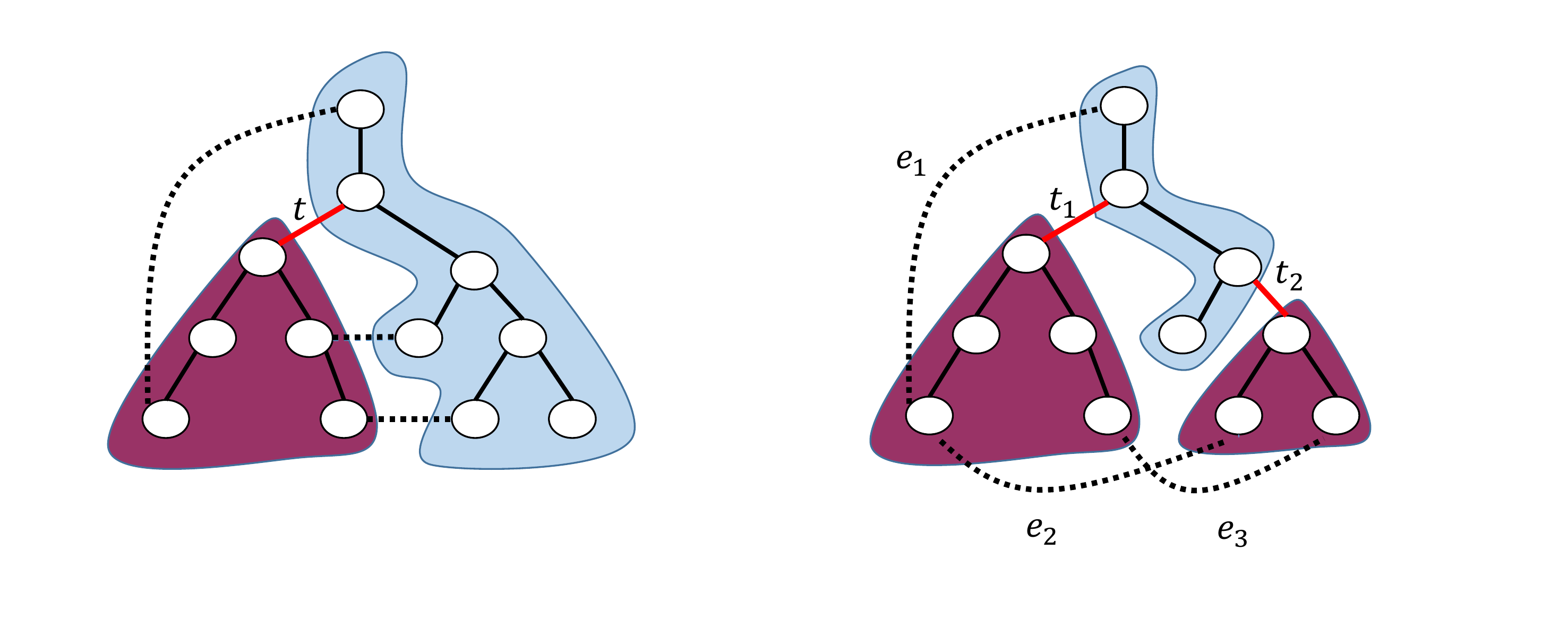}
 \caption{The left side illustrates the 1-respecting cut defined by $t$, the dotted edges are examples of edges that cover $t$. The right side illustrates the 2-respecting cut defined by $t_1,t_2$. The edges $e_2,e_3$ are edges that cover both $t_1,t_2$, where $e_1$ is an edge that covers $t_1$ but not $t_2$.}
\label{covPic}
\end{figure}

The following was shown in \cite{dory2018distributed}.

\begin{claim}[\cite{dory2018distributed}] \label{claimLearnCov}
In $\tilde{O}(D+\sqrt{n})$ time, all tree edges $t$ can learn $\cov(t)$.
\end{claim}

We next discuss a few useful observations used in our algorithm. 
For a non-tree edge $e=\{u,v\}$, we denote by $P_e$ the tree path between $u$ and $v$. 

%\begin{claim} \label{claim_cut_cover}
%A non-tree edge $e$ covers a 2-respecting cut $\{t,t'\}$ iff exactly one of $t,t'$ is in $P_e$.
%\end{claim}

\begin{claim} \label{claim_one_cut_cover}
A non-tree edge $e$ covers a 1-respecting cut defined by $t$ iff $t$ is in $P_e$.
\end{claim}

\begin{proof}
Removing $t$ from the tree disconnects $T$ to two connected components, one is the subtree below $t$, and the other is the rest of the tree. The edges that cross the cut are exactly the edges that have endpoints in these two components. Since $t$ is the only tree edge connecting these components, these are exactly all edges where $t$ in the tree path $P_e$.
\end{proof}

%\begin{restatable}{claim}{oneRespectCover} \label{claim_one_cut_cover}
%A non-tree edge $e$ covers a 1-respecting cut defined by $t$ iff $t$ is in $P_e$.
%\end{restatable}

\twoRespectCover*

\begin{proof}
Note that removing the edges of a minimum cut from a graph disconnects the graph into two connected components, we refer to them as the two sides of the cut.
Since $t,t'$ are edges in the cut, their endpoints are on the two different sides of the cut. If we take an edge $e=\{u,v\}$ that covers exactly one of $t,t'$, say $t$, then the path $P_e$ starts in the side of the cut of $u$, and then switches to the other side when it crosses $t$. It follows that $u$ and $v$ are on different sides of the cut, and that $e$ is an edge that crosses the cut. Any other edge either doesn't cover any of $t,t'$ or covers both of them. In the first case both $e$'s endpoints are on the same connected component in $T \setminus \{t,t'\}$, hence they are on the same side of the cut. In the second case, the path $P_e$ covered by $e$ starts on one side of the cut, and switches sides 2 times, when crossing $t$ and $t'$, hence both its endpoints are on the same side of the cut.
\end{proof}

%\begin{lemma}
%Let $H$ be a graph.
%If there is an $\alpha$-approximation algorithm for computing the number of min 2-respecting cuts in $H$ %covered by each edge $e \not \in H$ with respect to a spanning tree $T$ that takes $S$ time and works w.h.p, then there is an $O(\alpha)$-approximation algorithm for computing the number of minimum cuts in $H$ covered by each edge $e \not \in H$ that takes $\widetilde{O}(S+D+\sqrt{n})$ time and works w.h.p.
%\end{lemma}

We next show that it is enough to focus on 2-respecting cuts.
In a seminal work \cite{karger2000minimum}, Karger showed that one can find a small set of spanning trees $T_1,...,T_{\ell}$ of a graph $H$ such that for any minimum cut $C$ in $H$, there is a tree $T_j$ in the set such that $C$ 2-respects $T_j$. Actually, a stronger statement holds: cut $C$ 2-respects a \emph{constant fraction} of the trees. This construction was also implemented in a distributed setting \cite{daga2019distributed,DBLP:journals/corr/abs-2004-09129}, giving the following.\footnote{See Appendix A in \cite{DBLP:journals/corr/abs-2004-09129} for a proof. In that paper, the description of the theorem is slightly different, as the paper focuses on finding one minimum cut, but from the lemmas that appear before it based on \cite{karger2000minimum}, \Cref{thm_karger} follows.}

\begin{theorem}[\cite{karger2000minimum,daga2019distributed}] \label{thm_karger}
Given a graph $H$, it is possible to find a set of $\ell=\Theta(\log^{2.2}n)$ spanning trees $T_1,T_2,...,T_{\ell}$ such that any minimum cut in $H$ 2-respects 1/3 fraction of the trees. The construction time is $\tilde{O}(D+\sqrt{n})$ time, and the algorithm works with high probability.
\end{theorem}

This allows us to focus our attention on \emph{min 2-respecting cuts}, minimum cuts that have at most 2 edges in a spanning tree $T_j$. In our algorithm, to estimate the cost-effectiveness of edges, we estimate the number of min 2-respecting cuts they cover in all the spanning trees $T_1,...,T_{\ell}$. Using \Cref{thm_karger}, we show the following.

\begin{restatable}{lemma}{reductionLemma}[\textbf{Reduction from Min Cuts to Min 2-respecting Cuts}] \label{lem:reduction}
Let $H$ be a graph with a spanning tree $T$.
Let $A$ be an $\alpha$-approximation algorithm for computing the number of min 2-respecting cuts in $H$ covered by each edge $e \not \in H$. If $A$ takes $S$ time and works w.h.p, then there is an $O(\alpha)$-approximation algorithm for computing the number of minimum cuts in $H$ covered by each edge $e \not \in H$ that takes $\tilde{O}(S+D+\sqrt{n})$ time and works w.h.p.
\end{restatable}

\begin{proof}
From \Cref{thm_karger}, it is possible to find a set $T_1,...,T_{\ell}$ of spanning trees with $\ell = \Theta(\log^{2.2} n)$, such that for any minimum cut $C$ in the graph $G$, $C$ 2-respects 1/3 fraction of the spanning trees. The construction time is $\tilde{O}(D+\sqrt{n})$ time, and the algorithm works with high probability. By the end of the algorithm each vertex knows which edges are adjacent to it in each one of the spanning trees. 

Now we apply our $\alpha$-approximation algorithm to compute the number of min 2-respecting cuts covered by each edge $e \not \in H$ in each one of the spanning trees $T_1,...,T_{\ell}$. This takes $\tilde{O}(S)$ time. Denote by $c_i$ the number of min 2-respecting cuts covered by $e$ in the tree $T_i$, and denote by $c'_i$ the estimate of $c_i$ returned by the algorithm. Let $C$ be a minimum cut covered by $e$, then $C$ 2-respects at least $1/3$ of the trees $T_1,...,T_{\ell}$, it follows that in the sum $c_S = \sum_{i=1}^{\ell} c_i$ each minimum cut $C$ covered by $e$ is counted between $\ell/3$ to $\ell$ times. Also, any min 2-respecting cut is also a minimum cut in the original graph by definition. Hence the value $3c_S / \ell$ gives a 3-approximation to the number of minimum cuts covered by $e$. As we only have $\alpha$-approximations $c'_i$ to the values $c_i$, we can get $\alpha$-approximation to $c_S$ by computing $c'_S = \sum_{i=1}^{\ell} c'_i$. Now taking $3c'_S / \ell$ gives a $3 \alpha$-approximation to the number of minimum cuts covered by $e$, as needed.
\end{proof}

We say that $t$ and $t'$ are orthogonal tree edges if they are on different root to leaf paths in the tree. 

\begin{claim} \label{subtreeEdge}
Let $t$ and $t'$ be orthogonal tree edges.
Let $T_t$ be the subtree of $T$ below the edge $t$. If $\{t,t'\}$ is min 2-respecting cut, then there is an edge between the subtree $T_t$ and the subtree $T_{t'}$.
\end{claim}

\begin{proof}
It is shown in \cite{mukhopadhyay2020weighted,DBLP:journals/corr/abs-2004-09129} that if $\{t,t'\}$ is a min 2-respecting cut, then there is an edge that covers both $t$ and $t'$.\footnote{They show something stronger, that at least half of the edges that cover $t$ cover also $t'$.} Any edge that covers a tree edge $t$ has one endpoint in the subtree $T_t$, as $T_t$ is one of the connected components in the 1-respecting cut defined by $t$. Hence, if there is an edge $e$ that covers $t$ and $t'$ and they are orthogonal, the subtrees $T_t$ and $T_{t'}$ are disjoint, and hence $e$ connects the two subtrees. 
\end{proof}

\subsection{Fragment Decomposition} \label{sec:frag}

Given a spanning tree $T$, a fragment decomposition is a decomposition of the tree into fragments with nice properties. We use the fragment decomposition from \cite{DBLP:journals/corr/abs-2004-09129}, that is based on earlier constructions from \cite{ghaffari2016near,dory2018distributed,dory2019improved}. 

In this construction, the tree is decomposed into $O(\sqrt{n})$ \emph{edge-disjoint} fragments, each of size $O(\sqrt{n})$. Each fragment $F$ has a very specific structure, it has two special vertices $r_F$ and $d_F$, that are called the root and the unique descendant of the fragment. The fragment is composed of one main path between $r_F$ and $d_F$, called the \emph{highway} of the fragment, and additional subtrees attached to the highway. The edges on these subtrees are called \emph{non-highway} edges. The only vertices in the fragment that could be connected by a tree edge to another fragment are $r_F$ and $d_F$, they can also be part of other fragments. Other vertices of the fragment are internal vertices, that are not connected by a tree edge to any other fragment. See Figure \ref{fragPic} for illustration.

The structure of the fragments is captured by a \emph{skeleton tree}. In this tree the vertices correspond to all vertices that are either $r_F$ or $d_F$ in at least one fragment, and the edges correspond to the highways of the fragments. As there are $O(\sqrt{n})$ fragments, the description of the tree has $O(\sqrt{n})$ size, and we can broadcast it to the whole graph in $O(D+\sqrt{n})$ time. The construction of the fragment decomposition takes $\tilde{O}(D+\sqrt{n})$ time. 

\setlength{\intextsep}{0pt}
\begin{figure}[h]
\centering
\setlength{\abovecaptionskip}{-4pt}
\setlength{\belowcaptionskip}{4pt}
\includegraphics[scale=0.6]{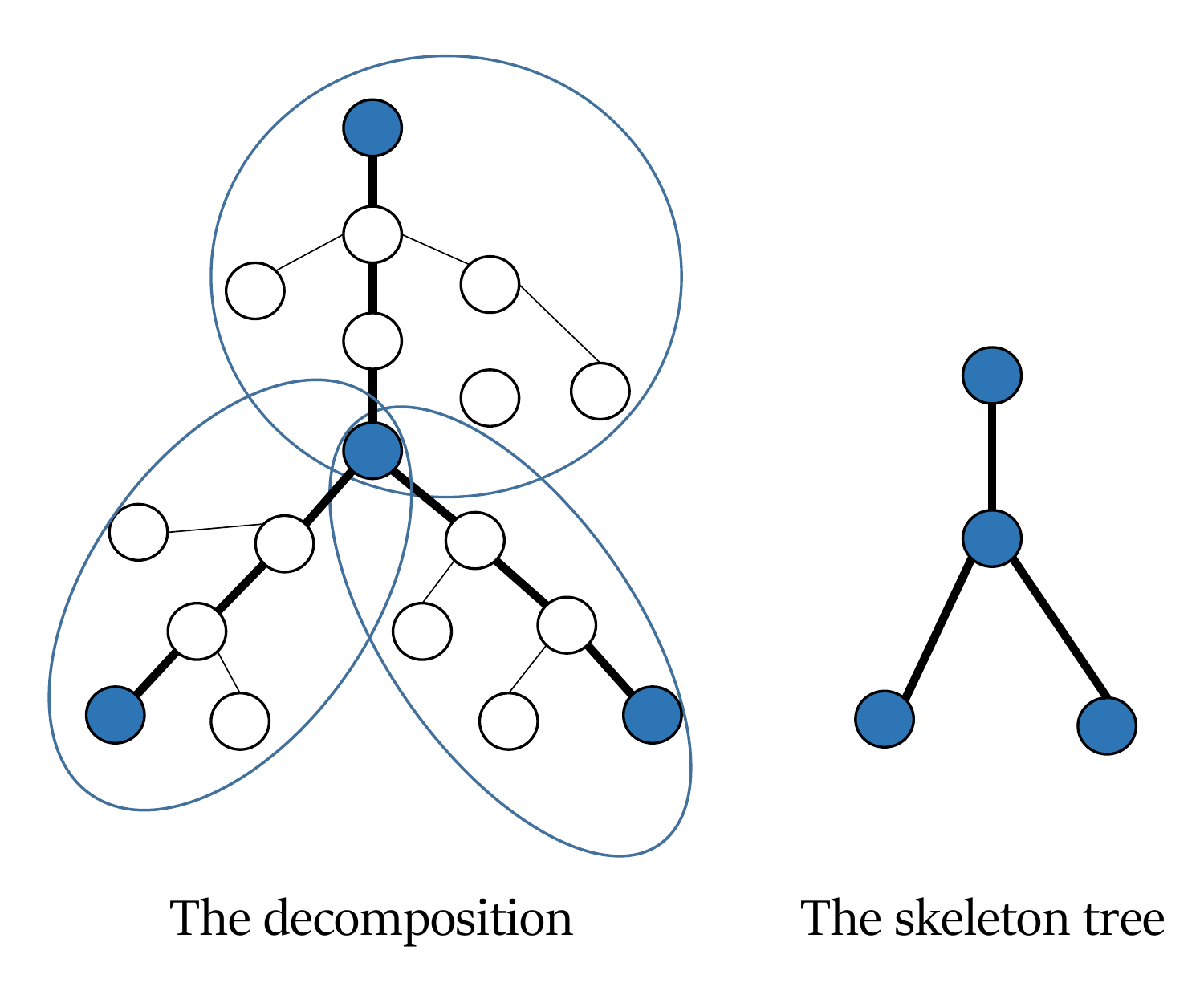}
 \caption{Illustration of the fragment decomposition.}
\label{fragPic}
\end{figure}

\begin{theorem}[\cite{DBLP:journals/corr/abs-2004-09129}]
Given a spanning tree $T$, the fragment decomposition can be computed in $\tilde{O}(D+\sqrt{n})$ time. By the end of the computation, all vertices know their fragment, and the complete structure of the skeleton tree. 
\end{theorem}   

One useful property of the fragment decomposition is that it allows non-tree edges to compute aggregate functions of tree edges that they cover efficiently. For a proof see Claim 2.5 in \cite{dory2020distributed}.\footnote{The proof focuses on the case that the information $m_t$ of tree edges is of size $O(\log{n})$, if the size is poly-logarithmic, the same proof idea works and it adds a poly-logarithmic term to the time complexity.} Similar proof for special cases appears in \cite{dory2018distributed}.

\begin{claim}[\cite{dory2020distributed}] \label{aggregate_Pe}
Assume that each tree edge $t$ has an information $m_t$ of poly-logarithmic size. In $\tilde{O}(D+\sqrt{n})$ time, all the non-tree edges $e$ can learn an aggregate function of the values $\{m_t\}_{t \in P_e}$.
\end{claim}

We next discuss a few simple observations on the fragment decomposition used in our algorithm.

\begin{claim} \label{globalCover}
Let $e$ be an edge where both its endpoints are not internal vertices in the fragment $F$ (they can be either outside the fragment, or the root or unique descendant of the fragment), then $e$ either covers exactly all tree edges in the highway $h_F$ of $F$, or does not cover any edge of $F$.
\end{claim}

\begin{proof}
Since both the endpoints of $e$ are not internal vertices in $F$, then the tree path covered by $e$ either does not go through $F$, or enters and leaves $F$ using the vertices $r_F$ and $d_F$ that are the only ones connected to vertices outside $F$, in this case $e$ covers exactly the highway of $F$.
\end{proof}

We denote by $LCA(u,v)$ the lowest common ancestor of $u$ and $v$.
The following claim follows from the structure of the fragment decomposition.

\begin{claim} \label{lcaFrag}
Let $u,v$ be two vertices that are not internal vertices in the fragment $F$ (they can either be outside the fragment, or the root or descendant of $F$), then $LCA(u,v)$ is not an internal vertex in $F$.
\end{claim}

%Let $u,v$ be two vertices in the fragments $F$ and $F'$, then $LCA(u,v)$ is either a vertex in $F$ or in $F'$ or is the unique root or descendant of another fragment.

\subsection{LCA Labels} \label{sec:lca}

Given a spanning tree $T$, the tool of lowest common ancestor (LCA) labels allows to give the vertices of the graph small labels of size $O(\log{n})$ bits such that given the labels of two vertices $u$ and $v$ we can infer the label of $LCA(u,v)$ directly from their labels. During our algorithm LCA labels will be useful for many tasks, such as checking if an edge $e$ covers a tree edge $t$, and checking if a vertex $w$ is inside the tree path between $u$ and $v$. When we work with a specific tree $T$, we use the LCA labels as the unique ids of vertices, and usually assume that if a vertex $v$ learns about a vertex $u$ it learns the LCA label of $u$.

We use the construction from \cite{censor2020fast,dory2020distributed} (see Lemma 2.9 in \cite{dory2020distributed}). This construction uses the fragment decomposition discussed above as one of its ingredients, and for each vertex the id of its fragment is part of its LCA label, we will exploit this property in our algorithm. 

\begin{lemma}[\cite{censor2020fast}]
In $O(D+\sqrt{n} \log^*{n})$ rounds, we can assign the vertices of a tree $T$ labels of size
$O(\log{n})$, such that given the labels of two vertices $u$ and $v$, we can infer the label of their LCA. 
\end{lemma}

The following useful property is shown in \cite{censor2020fast}.

\begin{claim}[\cite{censor2020fast}] \label{claimCheckCovLCA} 
Given the LCA labels of a tree edge $t$, and an edge $e$, one can deduce if $e$ covers $t$.
\end{claim}

\subsection{Basic Notation and Tools}

Given a spanning tree $T$, we denote by $p(u)$ the parent of $u$ in the tree. We say that two vertices $u$ and $v$ in the tree are \emph{orthogonal} if they are on different tree paths, i.e., if $u$ is not an ancestor of $v$ and vice verse. 

During our algorithm, we sometimes broadcast information in a spanning tree $T$. For example, in the subtrees of the fragments decomposition or in a global BFS tree. Broadcasting $\ell$ pieces of information can be easily done in $O(h+\ell)$ time, where $h$ is the height of the tree (see for example \cite{peleg2000distributed}). 

\begin{claim} \label{claim_broadcast}
Assume that there are $\ell$ pieces of information of $O(\log{n})$ bits that are initially distributed between a subset of vertices in a tree $T$ of height $h$. In $O(h+\ell)$ time all vertices in the tree can learn all $\ell$ pieces of information. 
\end{claim} 

For example, if we use a BFS tree we get a complexity of $O(D+\ell)$ time. We assume in the algorithm that we have a global BFS tree, such a tree can be constructed in $O(D)$ time \cite{peleg2000distributed}.

\section{Distributed Greedy Algorithm for Connectivity Augmentation} \label{sec:setcover}

%\mtodo{add here explanations removed from overview}

In this section, we discuss the distributed greedy algorithm for augmenting the connectivity. We prove that if we can estimate the cost-effectiveness of edges in $\tilde{O}(D+k\sqrt{n})$ time, we get an $O(\log{k}\log{n})$-approximation for $k$-ECSS in $\tilde{O}(k(D+k\sqrt{n}))$ time.
The algorithm generally follows the framework in \cite{dory2018distributed}. There are slight changes since in our case we approximate the cost-effectiveness where in \cite{dory2018distributed} it is computed exactly. We also prove that the approach from \cite{dory2018distributed} actually leads to an $O(\log{k}\log{n})$-approximation for $k$-ECSS, improving the $O(k \log{n})$-approximation guarantee shown in \cite{dory2018distributed}.

To find a low-cost $k$-edge-connected spanning subgraph of a graph $G$, we work in iterations, where in iteration $i$ we are given an $(i-1)$-edge-connected spanning subgraph $H_{i-1}$ of $G$ and we augment it to be $i$-edge-connected by adding a low-cost set of edges from $G$ to $H_{i-1}$. For example, in the first iteration $H_0 = \emptyset$ and we can compute a minimum spanning tree $H_1$ to solve the problem. In general, in iteration $i$, our goal is to solve the following augmentation problem we call $Aug_i$.

\begin{definition} [\textbf{Augmentation Problem}]
In the augmentation problem $Aug_i$ we are given an $i$-edge-connected graph $G=(V,E)$ and an $(i-1)$-edge-connected spanning subgraph $H_{i-1}$ of $G$. The goal is to find a minimum cost set of edges $A \subseteq E$ such that $H_i=H_{i-1} \cup A$ is $i$-edge-connected.
\end{definition}

An $\alpha$-approximation to $Aug_i$ is a set of edges $A$ such that $H_{i-1} \cup A$ is $i$-edge-connected, and $w(A) \leq \alpha w(A^*)$, where $A^*$ is an optimal augmentation. It is easy to see that if we can get an $\alpha$-approximation solution to $Aug_i$ for all $i$, then we get a $(k\alpha)$-approximation to $k$-ECSS, as the optimal solution $H^*$ for $k$-ECSS can be used to augment any subgraph $H$ to be $k$-edge-connected. This implies that $w(A^*_i) \leq w(H^*)$ for any $i$, where $A^*_i$ is the optimal solution for $Aug_i$. Using a tighter analysis, we can show that the achieved approximation for $k$-ECSS is in fact $(\alpha\log{k})$. Hence, our main goal is to approximate $Aug_i$.

\paragraph{Solving the Augmentation Problem via a Set Cover Algorithm.} 

The problem $Aug_i$ can be seen as a special case of the set cover problem, where our goal is to cover cuts of size $i-1$ in $H_{i-1}$. Concretely, as $H_{i-1}$ is $(i-1)$-edge-connected, the minimum cut size in $H_{i-1}$ is at least $i-1$. To augment the connectivity of $H_{i-1}$ to $i$ our goal is to add edges that cover all cuts that have size $i-1$ in $H_{i-1}$, where an edge $e \not \in H_{i-1}$ \emph{covers} a cut $C$ if $(H_{i-1} \setminus C) \cup \{e \}$ is connected. 

%\begin{definition} (\textbf{An Edge Covering A Cut})
%Let $H$ be a subgraph of $G$, and let $C$ be a cut in $H$.
%An edge $e \not \in H$ \emph{covers} the cut $C$ if $(H \setminus C) \cup \{e \}$ is connected. 
%\end{definition}

 %where an edge $e \not \in H_{i-1}$ \emph{covers} a cut $C$ if $(H_{i-1} \setminus C) \cup \{e \}$ is connected. 

\paragraph{Centralized Algorithm for this Special Set Cover Problem.} A simple centralized $O(\log{n})$-approximation for $Aug_i$ can be obtained using the simple greedy algorithm for set cover, comparing edges according to the number of cuts they cover. Let us elaborate. We denote by $C_e$ the set of cuts of size $i-1$ in $H_{i-1}$ covered by the edge $e \not \in H_{i-1}$. We define the cost-effectiveness of an edge $e \not \in {H_{i-1}}$ to be $\frac{|C_e|}{w(e)}$ where $w(e)$ is the weight of $e$. Intuitively our goal is to add edges with high values of cost-effectiveness to the augmentation, because that allows to cover many cuts while paying a low cost. This is indeed the idea of the centralized greedy algorithm. In this algorithm, at each step, we add an edge with maximum cost-effectiveness to the augmentation $A$, and we continue until $H_{i-1} \cup A$ is $i$-edge-connected. %To implement this in polynomial time in the centralized setting we need an algorithm for computing the cost-effectiveness of edges. One naive way to compute it in polynomial time could be to first find all $O(n^2)$ minimum cuts in the graph $H_{i-1}$ (this can be obtained for example using the minimum cut algorithms of Karger \cite{karger1993global, karger2000minimum}). Then, for each edge $e \not \in H_{i-1}$, we can check which minimum cuts it covers. 

\paragraph{Distributed Algorithms for this Special Set Cover Problem.} To implement this approach in a distributed setting we follow the approach in \cite{dory2018distributed}, based on a distributed set cover algorithm. At a high-level, all edges with near-maximal cost-effectiveness values are considered as \emph{candidates} to be added to the augmentation, and the algorithm uses randomization to choose which ones of them to add, such that many cuts are covered, but we do not add many edges to cover the same cuts. The whole algorithm requires only poly-logarithmic number of iterations to terminate, where in each iteration the cost-effectiveness of edges should be computed. We next discuss the algorithm in detail.

\remove{
\paragraph{Distributed Algorithms for this Special Set Cover Problem.} There are two main challenges in implementing the above approach in a \emph{distributed} setting:
\begin{itemize}
    \item[(I)] First, the greedy algorithm is inherently sequential as it adds only one edge to the augmentation at each step. To be able to get an efficient distributed algorithm we would like to add many edges in parallel. 
    \item[(II)] Second, the algorithm relies on computing the cost-effectiveness of edges, which is crucial for obtaining an $O(\log{n})$-approximation. However, in a distributed setting it is not clear how to compute the cost-effectiveness efficiently, as it requires each edge to learn about the number of cuts of size $i-1$ it covers in the graph $H_{i-1}$. 
\end{itemize}

The first challenge was addressed in \cite{dory2018distributed}, based on a distributed set cover algorithm. At a high-level, all edges with near-maximal cost-effectiveness values are considered as \emph{candidates} to be added to the augmentation, and the algorithm uses randomization to choose which ones of them to add, such that many cuts are covered, but we do not add many edges to cover the same cuts. The whole algorithm requires only poly-logarithmic number of iterations to terminate, where in each iteration the cost-effectiveness of edges should be computed. A detailed discussion of the algorithm appears in \Cref{sec:setcover}. 

}

\subsection{The Algorithm for $Aug_k$}

Recall that the input to $Aug_k$ is a $k$-edge-connected graph $G$, and $(k-1)$-edge-connected subgraph $H$ of $G$, and our goal is to find a low-cost set of edges $A$ from $G$ such that $H \cup A$ is $k$-edge-connected. During the algorithm we maintain an augmentation $A$ where initially $A = \emptyset$. In the greedy algorithm, the candidates to be added to $A$ are edges with near-maximal cost-effectiveness, where the cost-effectiveness $\rho(e)$ of an edge $e$ is equal to $\frac{|C_e|}{w(e)}$ where $C_e$ is the set of cuts of size $k-1$ in $H \cup A$ covered by $e$. 

%To understand the algorithm, it is helpful to think about the following bipartite graph with sides $X,Y$, where $X$ are all the edges $e \not \in H \cup A$, and $Y$ are all the cuts of size $(k-1)$ in $H \cup A$. An edge $e \in X$ is connected to a cut $C \in Y$ if $e$ covers $C$. The degree of an edge $e \in X$ is the number of cuts covered by $e$, which affects its cost-effectiveness. The degree of a cut $C \in Y$ is the number of possible edges that cover it. 

Our algorithm is composed of $O(\log{n})$ \emph{epochs}, where in epoch $i$ we handle edges with cost-effectiveness about $M/2^i$, where $M$ is the maximum possible value for cost-effectiveness. We assume for simplicity that the weights are positive integers polynomial in $n$.\footnote{Dealing with zero-weight edges is easy, as we can add all of them to $A$ without increasing the cost of the solution.} Since the number of minimum cuts in a graph is at most $O(n^2)$ and the minimum weight is 1, we have that $M = O(n^2)$. We denote by $w_{max}$ the maximum weight. During epoch $i$ edges with cost-effectiveness about $M/2^i$ are considered as \emph{candidates} to be added to the augmentation. When the epoch ends the maximum cost-effectiveness value decreases by a constant factor. 

Each epoch is composed of $O(\log{n})$ \emph{phases}, where in phase $j$ our goal is to handle cuts that are covered by about $m/2^j$ candidates. We denote by $deg(C)$ the number of candidates that cover a cut $C$. By the end of phase $j$ all the cuts of degree about $m/2^j$ are covered, and the maximum degree $deg(C)$ decreases by a constant factor. 

To get a good approximation, we want to guarantee two properties. First, that the candidates are only edges with near-optimal cost-effectiveness. This is important to get the $O(\log{n})$-approximation guarantee of the greedy centralized algorithm. Second, that we do not add too many candidates to cover the same cuts. Note that in the centralized algorithm, only one edge is added to the augmentation at each time, and ideally we want to make sure that also in our case at most $O(1)$ candidates are added to cover each cut. To obtain this, we work as follows. Recall that in phase $j$, $deg(C) \leq m/2^j$. Note that since there are $deg(C)$ candidates that cover $C$, if we add each one of them to the augmentation with probability $1/deg(C)$, in expectation we add one candidate to cover $C$. Hence, to cover cuts of degree about $m/2^j$, we add candidates to the augmentation with probability $2^j/m$. To guarantee that all these cuts are indeed covered with high probability we repeat the process $O(\log{n})$ times. It is important to note that for the algorithm to work we do not need to know the degrees of cuts (that seem difficult to compute). However, the algorithm guarantees that we cover all cuts of degree about $m/2^j$, even if we do not know their identity. The algorithm is summarized in Algorithm \ref{alg}. The value $c$ is a constant that depends on the analysis.
One small technical detail is that our algorithm for approximating cost-effectiveness assumes that the minimum cut in the graph $H \cup A$ is indeed $k-1$ and not larger. To make sure that this is the case, we terminate the algorithm in the first iteration that $H \cup A$ is $k$-edge-connected, we can check this using the exact min cut algorithm in \cite{DBLP:journals/corr/abs-2004-09129}.

\begin{algorithm}
\caption{Greedy Distributed Algorithm for $Aug_k$.}\label{alg}
\begin{algorithmic}[1]
\Statex Input: a $k$-edge-connected graph $G$ and a $(k-1)$-edge-connected spanning subgraph $H$ of $G$.
\Statex Output: a set of edges $A$ such that $H \cup A$ is $k$-edge-connected.
\State $A \gets \emptyset$
\For{$i=1,...,\log{(M w_{max})}$} 
\State If $H \cup A$ is $k$-edge-connected, terminate
\State For each edge $e \not \in H \cup A$, compute $\alpha$-approximation $\rho'(e)$ of $\rho(e)$
\State Define $Candidates_i$ to be the set of edges $e$ where $\rho'(e) \geq \frac{M}{\alpha 2^i}$
\For{$j=0,...,\log{m}$}
\For{$\ell=1,...,c\log{n}$}
\State If $H \cup A$ is $k$-edge-connected, terminate
\State For each edge $e \in Candidates_i \setminus A$, compute $\alpha$-approximation $\rho'(e)$ of $\rho(e)$
\State Remove $e$ from $Candidates_i$ if $\rho'(e) < \frac{M}{\alpha 2^i}$ or if $e \in A$
\State Each $e \in Candidates_i$ is added to $A$ with probability $2^j/m$
\EndFor
\EndFor
\EndFor
\end{algorithmic}
\end{algorithm} 

\subsection{Analysis}

We start with simple observations on Algorithm \ref{alg}.

\begin{claim} \label{claim:cost_ef_decrease}
By the end of epoch $i$, the maximum cost-effectiveness of an edge $e \not \in H \cup A$ is smaller than $M/2^i$.
\end{claim}

\begin{proof}
First note that the cost-effectiveness of an edge can only decrease during the algorithm, as cuts of size $k-1$ in $H \cup A$ covered by $e$ can be covered when we add edges to $A$. Let $e$ be an edge that has cost-effectiveness at least $M/2^i$ at the beginning of epoch $i$, and that is not added to $A$ in epoch $i$ or before it. As long as $\rho(e)\geq M/2^i$, then $\rho'(e) \geq \frac{M}{\alpha 2^i}$, as $\rho'(e)$ is an $\alpha$-approximation to $\rho(e)$. This means that $e \in Candidates_i$ as long as $\rho(e)\geq M/2^i$. In particular, if $\rho(e)\geq M/2^i$ in phase $j=\log{m}$, then $e$ is added to $A$ with probability 1, a contradiction. Hence, the cost-effectiveness of $e$ decreases until the end of the phase. 
\end{proof}

\begin{claim}
By the end of the algorithm, $H \cup A$ is $k$-edge-connected.
\end{claim}

\begin{proof}
From Claim \ref{claim:cost_ef_decrease}, after the last phase, the maximum cost-effectiveness of $e \not \in H \cup A$ is smaller than $1/w_{max}$. This means that $\rho(e)=0$ for all edges not in $H \cup A$. This implies that all cuts of size $k-1$ are covered. If $C$ is a cut of size $k-1$ in $H$, since $G$ is $k$-edge-connected there is at least one edge $e \not \in H$ that covers the cut. As long as $C$ is not covered, $\rho(e) \geq 1/w_{max}$. By the end of the algorithm either $e \in A$ or $\rho(e)=0$, in both cases $C$ is covered.
\end{proof}

We next analyze the approximation ratio of the algorithm. The analysis generally follows \cite{dory2018distributed} with slight changes.
For the analysis we assign costs to cuts of size $k-1$ in $H$ according to the epoch they are covered in. For a cut $C$ of size $k-1$ in $H$, that is first covered in epoch $i$, we define $cost(C)=2^i/M$. 
The proof is based on showing that $w(A) \approx O(\alpha^2) \sum_C cost(C) \leq O(\alpha^2 \log{n}) w(A^*)$, where $w(A)$ is the weight of the augmentation $A$ and $A^*$ is an optimal augmentation. The proof follows \cite{dory2018distributed} with the following changes. First, we adapt the proof to the case that we only have approximations to the cost-effectiveness (this adds $O(\alpha^2)$ term to the approximation obtained). Second, we compare the solution to an optimal \emph{fractional} solution. While this does not affect the approximation we obtain for $Aug_k$, it allows us later to show that the final approximation obtained for $k$-ECSS is $O(\log{k} \log{n})$, improving the $O(k\log{n})$ guarantee shown in \cite{dory2018distributed}. We start by discussing the proof of the right inequality.

We say that $A'$ is a fractional solution for $Aug_k$ if $A'$ assigns values $x(e) \geq 0$ to all edges, such that the following holds. If $C$ is a cut of size $k-1$ in $H$, and $E_C$ are the edges in $G \setminus H$ that cover $C$ we have $\sum_{e \in E_C} x(e) \geq 1$ (this is analogous to the requirement that all cuts are covered). An optimal fractional solution is a solution that minimizes $\sum_{e \not \in H} w(e) x(e)$. For a fractional solution $A'$, we define $w(A')=\sum_{e \not \in H} w(e) x(e)$, where $x(e)$ are the values given by the solution $A'$. We prove the following. The proof is based on the classical analysis of the greedy set cover algorithm, and is deferred to Appendix \ref{sec:app_setcover}. 

\begin{restatable}{claim}{ApproxClaim}\label{claim_approx1}
$\sum_C cost(C) \leq O(\log{n}) w(A^*)$, where $A^*$ is an optimal fractional solution to $Aug_k$.
\end{restatable}

The next claim follows from \cite{dory2018distributed}, the main difference in our case is the extra $O(\alpha^2)$ term that comes from the fact we approximate the cost-effectiveness, and not compute it exactly. See Appendix \ref{sec:app_setcover} for details. 

\begin{restatable}{claim}{ApproxExpect} \label{claim_approx2}
$E[w(A)] = O(\alpha^2) E[\sum_C cost(C)].$
\end{restatable}

%\mtodo{can probably get $O(\alpha)$ if the approximation algorithm has one-sided error. If $\alpha$ is constant it doesn't really matter.}

From Claims \ref{claim_approx1} and \ref{claim_approx2} we get that $E[w(A)] = O(\alpha^2 \log{n}) w(A^*).$ This gives approximation of $O(\alpha^2 \log{n})$ in expectation to $Aug_k$. We next show that this leads to an approximation of $O(\alpha^2 \log{k}\log{n})$ to $k$-ECSS.
Note the our algorithm for $Aug_k$ has $O(\log^3{n})$ iterations where in each iteration the only complex computation is checking if a graph is $k$-edge-connected and approximating the cost-effectiveness of edges. The first task can be solved in $\tilde{O}(D+\sqrt{n})$ time using the exact min cut algorithm in \cite{DBLP:journals/corr/abs-2004-09129}. In later sections we show that we get an $O(1)$-approximation of the cost-effectiveness of all edges in $\tilde{O}(D+k\sqrt{n})$ time. Based on this, we get the following. %\mtodo{specify $\alpha$? assume that it's constant in the theorem}

%\ThmGreedyAugi*

\begin{theorem} \label{thm_greedy_augi}
Assume that there is an $O(1)$-approximation algorithm for computing the cost-effectiveness of edges that takes $\tilde{O}(D+i\sqrt{n})$ time where $i$ is the the size of the minimum cut in the graph $H$ we augment.
Then there is an $O(\log{k} \log{n})$-approximation algorithm for $k$-ECSS that runs in $\tilde{O}(k(D+k\sqrt{n}))$ time. The algorithm works with high probability.
\end{theorem}

\begin{proof}
To solve $k$-ECSS we start with an empty subgraph $H$. To augment its connectivity to 1 we build an MST $T^*$. Next, we work in iterations $i=2,...,k$, where in iteration $i$ we augment the connectivity of $H$ from $i-1$ to $i$ using our $O(\log{n})$-approximation for $Aug_i$. Our final graph $H$ is $k$-edge-connected (assuming that the original input graph is $k$-edge-connected), and the time complexity is $\tilde{O}(k(D+k\sqrt{n}))$, as solving $Aug_i$ takes $\tilde{O}(D+i\sqrt{n})$ time, as it requires $O(\log^3{n})$ iterations where in each one we approximate the cost-effectiveness in $\tilde{O}(D+i\sqrt{n})$ time, and also compute the min cut in the graph $H \cup A$ in $\tilde{O}(D+\sqrt{n})$ time using the algorithm in \cite{DBLP:journals/corr/abs-2004-09129}. 

We next analyze the approximation ratio obtained. We denote by $H_{i-1}$ the graph $H$ at the beginning of iteration $i$.
Let $H^*$ be a minimum-cost $k$-edge-connected spanning subgraph. First, $w(T^*) \leq w(H^*)$, where $T^*$ is the minimum spanning tree of the graph, as $H^*$ is a spanning subgraph. Denote by $A_i$ the set of edges we add to $H$ in iteration $i$. Recall that $E[w(A_i)] \leq O( \log{n}) w(A^*)$, where $A^*$ is an optimal fractional solution to $Aug_i$ with the input graph $H_{i-1}$. We show that there is a fractional solution $A'_i$ to $Aug_i$ of cost at most $\frac{1}{k-(i-1)} w(H^*)$. The fractional solution is defined as follows. For any edge $e \in H^* \setminus H_{i-1}$, we set $x(e)=\frac{1}{k-(i-1)}$. For all other edges $x(e)=0$. The cost of the solution is clearly at most $\frac{1}{k-(i-1)} w(H^*)$. We next show that this is a valid fractional solution. Let $C$ be a cut of size $i-1$ in $H_{i-1}$. We need to show that $\sum_{e \in E_C} x(e) \geq 1$. Since $C$ is a minimum cut in $H_{i-1}$, removing the edges of $C$ from $H_{i-1}$ splits it to two connected components $U, V \setminus U$. Since $H^*$ is $k$-edge-connected there are at least $k$ edges in $H^*$ between $U$ and $V \setminus U$. Moreover, as $C$ is a cut of size $i-1$ in $H_{i-1}$, at most $i-1$ of these edges are in $H_{i-1}$. Hence, there are at least $k-(i-1)$ edges that cover $C$ in $H^* \setminus H_{i-1}$. Since for each of these edges we have $x(e) = \frac{1}{k-(i-1)}$, we get that $\sum_{e \in E_C} x(e) \geq 1$ as needed. This gives $E[w(A_i)]\leq O( \log n) \frac{1}{k-(i-1)} w(H^*)$. If we sum over all iterations, we get that $$E[w(H)] \leq w(T^*) + O(\log{n}) \sum_{i=2}^{k} \frac{1}{k-(i-1)} w(H^*) = O( \log{k} \log{n}) w(H^*).$$

This provides $O(\log{k}\log{n})$-approximation in expectation. To get the approximation guarantee w.h.p we can just repeat the algorithm $O(\log{n})$ times and take the lowest cost solution obtained. Note that in $O(D)$ time we can learn the cost of the solution by an aggregate computation in a BFS tree. This completes the proof.
\end{proof}

\section{Succinct Representation of Min Cuts} \label{sec:succinct_min_cuts}
A key ingredient of our algorithm is a structural lemma that shows that for each tree edge $t$, we can store $O(k \log{n})$ bits that represent all the min 2-respecting cuts containing $t$. In more detail, if we fix an edge $t$, all the edges $t'$ such that $\cut(t,t')=k$ can be divided to $\ell = O(k)$ segments $S^t_1,...,S^t_{\ell}$ with corresponding cover values $c^t_1,...,c^t_{\ell}$
such that the edges in $S^t_i$ where $\cut(t,t')=k$ are exactly all edges in $S^t_i$ where $\cov(t')=c^t_i$.

\succinctCutsLemma*

%\begin{lemma} \label{min_cuts_lemma}
%Let $G=(V,E)$ be a graph with minimum cut of size $k$, and let $T$ be a spanning tree of $G$. Fix a tree edge $t \in T$. Then there exist $\ell = %O(k)$ tuples $\{(S^t_i,c^t_i)\}_{1 \leq i \leq \ell}$ such that the edges $t'$ where $\cut(t,t')=k$ are exactly all edges in $E_t = \cup_{1 \leq i %\leq \ell} \{ t' \in S^t_i | \cov(t') = c^t_i \}$.
%\end{lemma}

To simplify the presentation, we first give a proof for the case that $T$ is a path, and later show how to extend the proof to a general tree. 

\begin{proof}[Proof for a path $T$.]
Fix a tree edge $t$. For the proof it is convenient to give an orientation to the path $T$ and focus first on min cuts $\{t,t'\}$ where $t'$ is on the right of $t$. We look at all edges that cover $t$, each such edge has one endpoint to the left of $t$ and one endpoint to the right of $t$, we sort these edges according to the right endpoint, from the closest to the furthest from $t$. Let $e_1,...,e_k$ be the first $k$ edges according to this order, and let $v_1,...,v_k$ be the corresponding right endpoints of $e_1,...,e_k$. See Figure \ref{pathLemmaPic} for illustration. Note that any edge $t'$ to the right of $v_k$ cannot be in a min 2-respecting cut with $t$, as there are more than $k$ edges crossing the cut defined by $\{t,t'\}$: the edges $e_1,...,e_k$ that cover $t$ and not $t'$, and the edge $t'$ that covers $t'$ and not $t$. Also, note that there are indeed at least $k$ edges that cover $t$ as $\cov(t) \geq k$ because the min cut has size $k$. For $1 \leq i \leq k-1$, we define $S^t_i$ to be the segment in $T$ between $v_i$ to $v_{i+1}$. As $t$ is also an edge that covers $t$, the vertex $v_1$ is the right endpoint of $t$. Also, it may be the case that some segments are empty in the case there are multiple edges that cover $t$ and have the same right endpoint. We define $c^t_i = \cov(t) + k - 2i$.

We next show that the tuples $(S^t_i,c^t_i)$ capture all the min 2-respecting cuts of the form $\{t,t'\}$ where $t'$ is on the right of $t$. As explained above, all edges $t'$ in a min 2-respecting cut with $t$ are inside the segment between $v_1$ to $v_k$, hence they are contained in one of the segments $S^t_i$. Let $t'$ be an edge in $S^t_i$, we show that $\cut(t,t')=k$ iff $\cov(t')= c^t_i$. Recall that by Claim \ref{covClaim}, $\cut(t,t')=\cov(t)+\cov(t') - 2 \cov(t,t')$. Now by the definition of the sets $S^t_i$, any edge $t'$ inside $S^t_i$ is covered by all edges that cover $t$ except $e_1,...,e_i$ that have their right endpoint before $S^t_i$. This follows as any edge that covers $t$ has one endpoint on the left of $t$ and one endpoint on the right of $t$, hence it covers $t'$ iff this right endpoint is on the right of $t'$. Hence, by definition $\cov(t,t')=\cov(t)-i$ for all edges in $S^t_i$. If $\cut(t,t')=k$, we have that $k = \cov(t)+\cov(t') - 2 \cov(t,t') = \cov(t)+\cov(t') - 2 (\cov(t) - i)$. This gives $\cov(t') = \cov(t) + k - 2i = c^t_i$, as needed.

This completes the proof for cuts of the right of $t$, to handle cuts where $t'$ is on left of $t$ we can repeat the same arguments where we replace right by left. Overall we get $2(k-1) = O(k)$ tuples $(S^t_i,c^t_i)$ that capture all the min 2-respecting cuts that contain $t$.
\end{proof}

\setlength{\intextsep}{0pt}
\begin{figure}[h]
\centering
\setlength{\abovecaptionskip}{0pt}
\setlength{\belowcaptionskip}{0pt}
\includegraphics[scale=0.6]{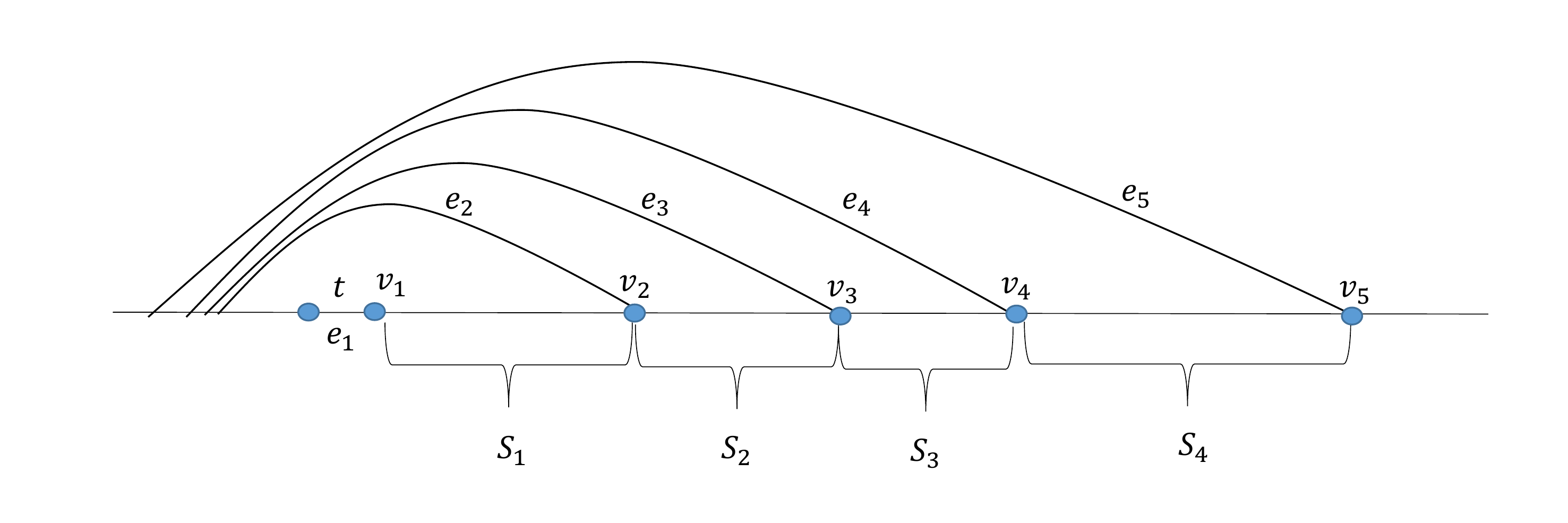}
 \caption{Illustration of the proof of Lemma \ref{min_cuts_lemma}. A tree edge $t' \in S_i$ is in a min 2-respecting cut with $t$ iff $\cov(t')=c^t_i= \cov(t) + k - 2i$.}
\label{pathLemmaPic}
\end{figure}

\subsection{A Proof for a General Tree}

We next generalize the proof to a general tree $T$. There are two main ingredients. First, we use structural properties of cuts shown in \cite{mukhopadhyay2020weighted, DBLP:journals/corr/abs-2004-09129} to focus our attention only on one tree path $P_t$ per edge $t$. Second, while now non-tree edges that cover $t$ may have endpoints outside $P_t$, we show that we can compare them and sort them according to the subpath of $P_t$ they cover. This allows us to follow the main proof idea of the path case also when dealing with a general tree. We next discuss this in detail.
First, we use the following. 

\begin{restatable}{lemma}{uniquePath}\cite{mukhopadhyay2020weighted}\label{unique_path_lemma}
Let $G$ be a graph with minimum cut of size $k$, and let $T$ be a spanning tree of $G$. For each tree edge $t$, all the edges $\{t'|\cut(t,t')=k\}$ are in one tree path $P_t$ containing $t$. 
\end{restatable}

The proof follows from \cite{mukhopadhyay2020weighted}, we add a proof in Appendix \ref{app_proofs} for completeness. 

\paragraph{Comparing edges.} In the path case we compared different edges that cover $t$ according to their right endpoint.
Ideally, we want to follow a similar argument now with respect to the path $P_t$.
However, the edges that cover $t$ can have their endpoints outside $P_t$. To deal with it, we compare edges according to the subpath they cover in $P_t$. For doing so, we need the following definitions.

Let $t=\{u,p(u)\}$ be a tree edge, and let $P$ be a tree path where $t \in P$. Let $e$ be an edge that covers $t$. We next define two virtual endpoints $v_h(e,P),v_l(e,P)$ in $P$ such that $e$ covers exactly the subpath of $P$ between $v_l(e,P)$ and $v_h(e,P)$. To do so, we first take a closer look at $P$. The path $P$ has one part below $u$ and one part that starts above $p(u)$ and ends either above $p(u)$ or orthogonal to $p(u)$. We call these parts $P^l$ and $P^h$ respectively. Let $\{p_l,p_h\}$ be the endpoints of $P$ where $p_l$ is the endpoint below $u$. Similarly, let $\{v_l,v_h\}$ be the endpoints of $e$ where $v_l$ is the endpoint below $u$ (as $e$ is an edge that covers $t$ it has one such endpoint).

We define $v_l(e,P) = LCA(p_l,v_l)$, since $u$ is an ancestor of both $p_l$ and $v_l$, $v_l(e,P)$ is a vertex in $P$.
The definition of $v_h(e,P)$ requires two steps. First, let $a_e = LCA(p(u),v_h)$ and $a_P = LCA(p(u),p_h)$. If $a_e \neq a_P$, and $a_e$ is below $a_P$, we define $v_h(e,P)=a_e$. If $a_e \neq a_P$, and $a_P$ is below $a_e$, we define $v_h(e,P)=a_P$. Otherwise, $a_e = a_P$, and we define $v_h(e,P) = LCA(p_h,v_h)$. See Figure \ref{comparePathsPic} for illustration. 

%maybe we don't want to have a defintion specific for the choice of $t$

\begin{claim} \label{lca_claim}
The tree path between $v_l(e,P)$ and $v_h(e,P)$ is exactly the subpath of $P$ covered by $e$. Moreover, the subpath between $v_l(e,P)$ to $u$ is the subpath of $P^l$ covered by $e$, and the subpath between $p(u)$ to $v_h(e,P)$ is the subpath of $P^h$ covered by $e$.
\end{claim}

\setlength{\intextsep}{0pt}
\begin{figure}[h]
\centering
\setlength{\abovecaptionskip}{0pt}
\setlength{\belowcaptionskip}{0pt}
\includegraphics[scale=0.6]{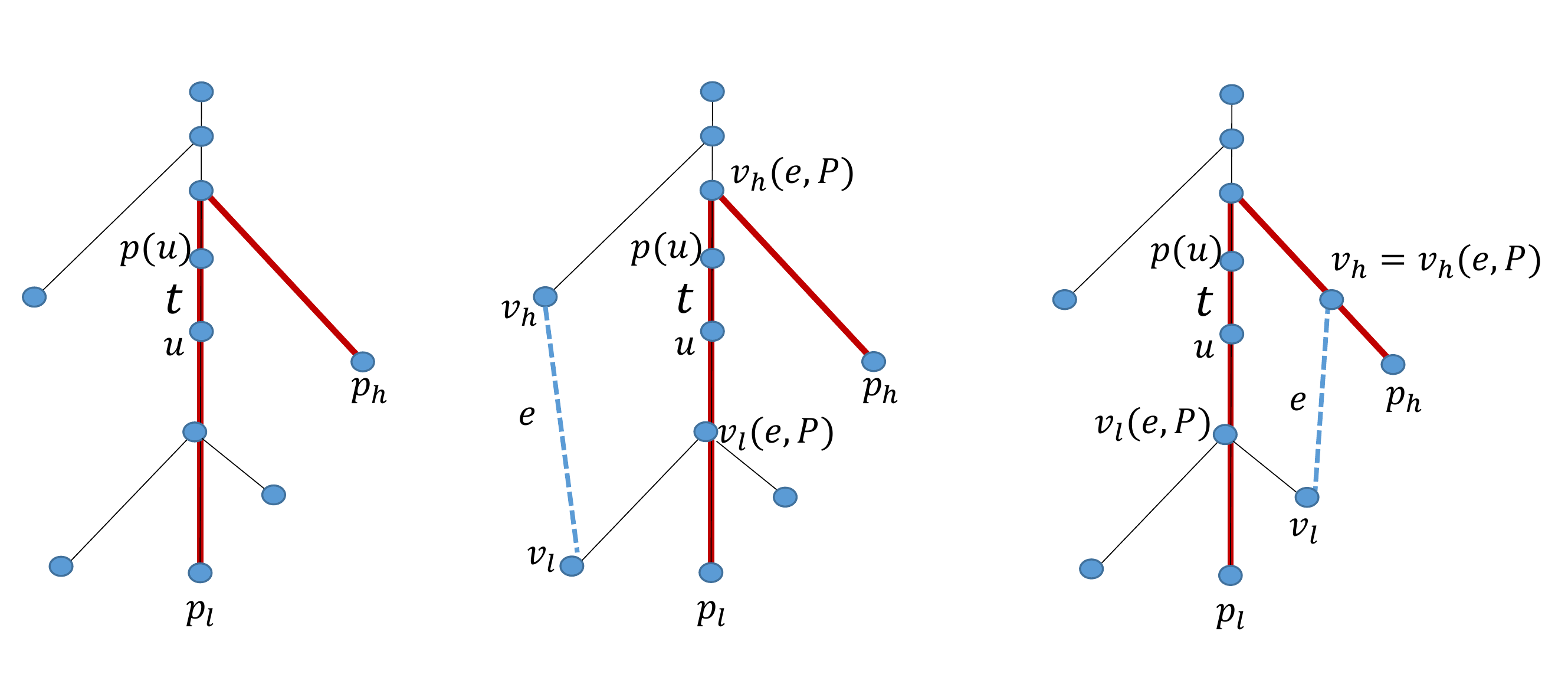}
 \caption{The path $P$ is the bold path between $p_l$ to $p_h$. On the right there are illustrations of several edges $e$ and the corresponding endpoints $v_l(e,P),v_h(e,P)$.}
\label{comparePathsPic}
\end{figure}

\begin{proof}
The edge $e=\{v_l,v_h\}$ covers $t=\{u,p(u)\}$ and the tree paths $P_1,P_2$ between $v_l$ to $u$ and between $p(u)$ to $v_h$, respectively. We next show that the intersection of $P_1,P_2$ with $P$ is exactly captured by the vertices $v_l(e,P)$ and $v_h(e,P)$. 

First, $P_1$ and $P^l$ are two subpaths below $t$ that start at $u$, they may have some joint edges before they diverge to two different subpaths that end at $v_l$ and $p_l$, the point they diverge is exactly $LCA(v_l,p_l)=v_l(e,P)$. Hence, $e$ covers exactly the subpath of $P^l$ between $v_l(e,P)$ and $u$. 

The paths $P^h$ and $P_2$ are both paths that start above $p(u)$ and have one part above $p(u)$ and one orthogonal to $p(u)$. More concretely, $P^h$ is the path between $p(u)$ and $p_h$. It has one part between $p(u)$ and $LCA(p(u),p_h)=a_P$ and the second part between $a_P$ and $p_h$. Similarly, the path $P_2$ has one part between $p(u)$ and $LCA(p(u),v_h)=a_e$ and the second part between $a_e$ and $v_h$. If $a_e \neq a_P$ then the second parts of the two paths are disjoint, and the intersection of the first parts is all edges above $p(u)$ that are below $a_e$ and $a_P$. By the definition of $v_h(e,P)$ this is exactly the subpath between $p(u)$ and $v_h(e,P)$. If $a_e = a_P$, then the first parts of the two paths are equal, and they may have some joint edges below $a_P$, until they diverge to two subpaths that end at $p_h$ and $v_h$. The point they diverge is exactly $LCA(p_h,v_h)$, which is some vertex in the subtree rooted at $a_P=a_e$, as this is an ancestor of both vertices. By definition, $v_h(e,P)=LCA(p_h,v_h)$ in the case that $a_e=a_P$, and indeed $e$ covers exactly the subpath of $P^h$ between $p(u)$ to $v_h(e,P)$. This path has two parts, one from $p(u)$ to $a_P$, and the second from $a_P$ to $v_h(e,P)$.   
\end{proof}

From the definition of $v_l(e,P),v_h(e,P)$, we can compute them using LCA computations. Note that LCA computations can also allow checking which endpoint of a path $P$ or an edge $e$ is below $t$, or check if two vertices are in the same tree path (in which case one is ancestor of the other), and if so determine which one of them is higher in the tree. 

\begin{claim} \label{claimLCAcomputeEndpoints}
Given the LCA labels of the endpoints $\{p_l,p_h\}$ of a path $P$ that contains the tree edge $t=\{u,p(u)\}$, the LCA labels of the endpoints $\{v_l,v_h\}$ of an edge $e$ that covers the tree edge $t$, and the LCA labels of $\{u,p(u)\}$ one can deduce the LCA labels of the vertices $v_l(e,P),v_h(e,P)$.
\end{claim}

We can now give a proof for Lemma \ref{min_cuts_lemma} where $T$ is a general tree.

\begin{proof}[Proof of Lemma \ref{min_cuts_lemma}]
Fix a tree edge $t=\{u,p(u)\}$. From Lemma \ref{unique_path_lemma}, all the edges where $\cut(t,t')=k$ are in one tree path $P_t$ containing $t$. This path has one subpath $P_t^l$ below $t$ starting at $u$, and another subpath $P_t^h$ starting at $p(u)$. We first focus on edges $t'$ in the first subpath $P_t^l$. We look at all edges that cover $t$ and compare them according to the subpath they cover in $P_t^l$. From Claim \ref{lca_claim}, we know that an edge $e$ that covers $t$, covers exactly the subpath of $P_t^l$ between $v_l(e,P_t)$ and $u$. Note that all vertices $v_l(e,P_t)$ for any edge $e$ that covers $t$ are vertices in the path $P_t^l$ and hence we can compare them. We sort the edges $e$ that cover $t$ according to the vertices $v_l(e,P_t)$ from the closest to the furthest from $t$. Let $\{e_1,...,e_k\}$ be the first $k$ edges in the sorted order, and let $v_i = v_l(e_i,P_t)$.

First, all edges $t' \in P_t^l$ such that $\cut(t,t')=k$ are on the subpath of $P_t^l$ between $v_1$ to $v_k$. For this, first note that $v_1=u$ as $t=\{u,p(u)\}$ is an edge that covers $t$, and $v_l(t,P_t)=u$, and this is clearly the closest vertex to $t$ in $P_t^l$. Second, for any edge $t' \in P_t^l$ that is below $v_k$, there are more than $k$ edges that cross the cut defined by $\{t,t'\}$. This follows since the edges $\{e_1,...,e_k\}$ are edges that cover $t$ and not $t'$, as the subpath of $P_t^l$ they cover is between $u$ to a vertex $v_i$ that is above $t'$. Also, the edge $t'$ covers $t'$ and not $t$, hence $\cut(t,t') > k$ for these edges.

We divide the edges $t' \in P_t^l$ between $v_1$ to $v_k$ to $k-1$ subsets, where $S^t_i$ is the subpath between $v_i$ to $v_{i+1}$ (it can be empty if $v_i = v_{i+1}$). Let $c^t_i = \cov(t) + k - 2i$. We next show that the tuples $(S^t_i,c^t_i)$ capture all min 2-respecting cuts $\{t,t'\}$ where $t' \in P_t^l$. 
Let $t' \in S^t_i$, we show that $\cut(t,t')=k$ iff $\cov(t')= c^t_i$. By Claim \ref{covClaim}, $\cut(t,t')=\cov(t)+\cov(t') - 2 \cov(t,t')$.
First, we show that for all edges $t' \in S^t_i$, the value $\cov(t,t')$ is fixed and equal to $\cov(t) - i$. Recall that we sorted the edges that cover $t$ according to the subpath they cover in $P_t^l$. Any edge that covers $t$ except $\{e_1,...,e_i\}$ covers the subpath between $u$ to $v_{i+1}$ (and possibly additional edges). As $t' \in S^t_i$ is in the subpath between $v_i$ to $v_{i+1}$ it is covered by all edges that cover $t$ except $\{e_1,...,e_i\}$ that cover subpaths that end before $S^t_i$. Hence, we have $\cov(t,t')=\cov(t)-i$. This gives $\cut(t,t')=\cov(t)+\cov(t') - 2 \cov(t,t')= \cov(t)+\cov(t') - 2 (\cov(t) - i)$. This implies that $\cut(t,t')=k$ iff $\cov(t') = \cov(t) + k - 2i = c^t_i$, as needed. 

By now we dealt with all min 2-respecting cuts where $t' \in P_t^l$. Similarly, we can define $k-1$ tuples $(S^t_{i+(k-1)},c^t_{i+(k-1)})$ that capture min 2-respecting cuts where $t' \in P_t^h$. Here, we sort the edges $e$ covering $t$ according to the vertices $v_h(e,P_t)$ from the closest to furthest from $t$, and again define $c^t_{i+(k-1)} = \cov(t) + k - 2i$. Recall that $v_h(e,P_t)$ is defined such that $e$ covers exactly the subpath of $P_t^h$ between $p(u)$ and $v_h(e,P_t)$. Following the same arguments gives that $\cut(t,t')=k$ for $t' \in P_t^h$ iff there exists an index $i$ where $t' \in S^t_{i+(k-1)}$ and $\cov(t')=c^t_{i+(k-1)}$. This completes the proof.
\end{proof} 

\section{Learning the Cut Information} \label{sec:learn_cut}

We showed that for each tree edge $t$, we can represent in $O(k \log{n})$ bits information about all min 2-respecting cuts that contain $t$, where $k$ is the value of the min cut. We next show that all tree edges $t$ can learn this information efficiently, as well as additional useful information such as the number of min 2-respecting cuts $\{t,t'\}$ where $t'$ is in a certain fragment. For the algorithm, we decompose the tree into fragments using the fragment decomposition from Section \ref{sec:frag}. We also apply on the tree the LCA labeling scheme from Section \ref{sec:lca}.

We denote by $\cutInfo(t)$ the $O(k \log{n})$ information represented by the tuples $\{(S^t_i,c^t_i)\}_{1 \leq i \leq \ell}$ from Lemma \ref{min_cuts_lemma}. Note that $S^t_i$ is a subpath in the tree represented by two vertices, $c^t_i$ is an integer bounded by $m=O(n^2)$, and $\ell = 2(k-1)$, hence the information indeed can be represented in $O(k \log{n})$ bits. When we learn about a subpath $S^t_i$, we will see that we learn the LCA labels of the endpoints of $S^t_i$, which in particular include also the fragments of the endpoints, as explained in Section \ref{sec:lca}.
In the proof of Lemma \ref{min_cuts_lemma}, we divided the tuples $\{(S^t_i,c^t_i)\}_{1 \leq i \leq \ell}$ to two parts, one is the tuples representing min cuts $\{t,t'\}$ where the tree edge $t'$ is in the path $P_t^l$ below $t$, and the second where the tree edge $t'$ is in the path $P_t^h$, above or orthogonal to $t$. We denote by $\cutInfo_l(t)$, and $\cutInfo_h(t)$ the tuples $\{(S^t_i,c^t_i)\}$ divided to these two parts, such that $\cutInfo(t)= \cutInfo_l(t) \cup \cutInfo_h(t)$. %In our algorithm, in some cases it is enough to learn only about one of these parts. 
We denote by $\fragInfo(t,F)$ the number of min cuts $\{t,t'\}$ where $t'$ is in the fragment $F$, and by $\fragInfo_h(t,F)$ the number of min cuts $\{t,t'\}$ where $t'$ is in the \emph{highway} of the fragment $F$. Recall that each fragment is composed of a main path that is called the \emph{highway} of the fragment, and additional sub-trees attached to it called \emph{non-highways}, see Section \ref{sec:frag} for the full details.
We next show how tree edges learn the information $\cutInfo(t)$, $\fragInfo(t,F)$ and $\fragInfo_h(t,F)$.
We show two algorithms, one for tree edges $t$ in non-highways, and one for tree edges in highways.
We first discuss a few useful observations.

\begin{claim} \label{claimCompare}
Let $e,e'$ be two edges that cover a tree edge $t$, and let $P$ be a path that contains $t$. We can use LCA computations to compare the vertices $v_h(e,P),v_h(e',P)$, and check which one is closer to $t$. Similarly, we can compare $v_l(e,P),v_l(e',P)$ and check which one is closer to $t$.
\end{claim}

\begin{proof}
Let  $v=v_h(e,P),v'=v_h(e',P)$. Both $v$ and $v'$ are vertices in $P^h$. The path $P^h$ starts in the vertex $p(u)$, it has a part above $p(u)$ that ends in the vertex $a_P = LCA(p(u),p_h)$, and then a part below $a_P$. We can check in which part $v,v'$ are using LCA computations, as the first part is above $p(u)$ and the second is not, so if $LCA(p(u),v)=v$ we know that $v$ is in the first part, and otherwise it is in the second part. Now if $v,v'$ are in different parts, the closer to $t$ is the one in the first part. Otherwise, if they are both in the first part, above $p(u)$, the closer to $t$ is the lower in the path, which can be checked using LCA computations: $LCA(v,v')=v$ iff $v'$ is the lower. Finally, if they are both in the second part, orthogonal to $p(u)$, the closer to $t$ is the higher, which again can be checked using LCA computations. To conclude, using the LCA labels alone we can compare the vertices $v_h(e,P)$ and see which one is closer to $t$. Using similar ideas we can compare $v_l(e,P),v_l(e',P)$ (here we just need to check which one is higher in the tree).
\end{proof}

\begin{claim} \label{claimFrag}
Given $\cutInfo(t)$, the values $\{(t',\cov(t'))\}_{t' \in F}$ and all LCA labels of vertices in $F$, as well as indication of which tree edges are part of the highway of $F$, one can deduce all the min 2-respecting cuts of the form $\{t,t'\}$ for $t' \in F$. In particular, it can deduce $\fragInfo(t,F)$ and $\fragInfo_h(t,F)$. 
\end{claim}

\begin{proof}
The information $\cutInfo(t)$ is a set of tuples $(S^t_j,c^t_j)$, for each one of the tuples we would like to check which edges in $S^t_j$ with cover value $c^t_j$ are in the fragment $F$. Summing the number of such edges for all $j$ gives exactly the number of min 2-respecting cuts $\{t,t'\}$ where $t' \in F$ from the definition of $\cutInfo(t)$, this gives $\fragInfo(t,F)$. To get $\fragInfo_h(t,F)$ we sum only cuts where $t'$ is on the highway. We next explain the computation for one such tuple $(S^t_j,c^t_j)$. The segment $S^t_j$ is the tree path between $v_j$ to $v_{j+1}$. The information $\cutInfo(t)$ has the LCA labels and fragments of $v_j$ and $v_{j+1}$. We use this information to identify the intersection of $S^t_j$ and $F$ (that may be empty), as follows. 
\begin{itemize}
\item {\textbf{Case 1: $v_j$ and $v_{j+1}$ are vertices in $F$}. Here the segment $S^t_j$ is contained entirely inside $F$. Since we know the complete structure of the fragment $F$ and the cover values of all edges, we can just check which edges in the segment $S^t_j$ have cover value $c^t_j$ in this case.}
\item {\textbf{Case 2: $v_j$ and $v_{j+1}$ are vertices outside $F$}. Here, we use the structure of the skeleton tree and the fragments of $v_j$ and $v_{j+1}$ to identify the tree path between them in the skeleton tree, denote this path by $P_S$. Let $r_F,d_F$ be the root and unique descendant of $F$. Recall that the skeleton tree has one edge $e_F= \{r_F,d_F\}$ that represents the highway of $F$. If $P_S$ does not contain $e_F$, it follows that the intersection $S^t_j \cap F$ is empty. Otherwise, $S^t_j \cap F$ is exactly the highway of $F$.
As we know the cover values of edges in the highway, we can check which of them have cover value $c^t_j$, which completes this case.}
\item {\textbf{Case 3: exactly one of $v_j$ and $v_{j+1}$ is in $F$}. Assume w.l.o.g that $v_j \in F$. Here we use the structure of the skeleton tree, and fragment $F'$ of $v_{j+1}$ to identify the closest vertex to $v_{j+1}$ in $F$ (for this, we look at the tree path between $F'$ and $F$). This vertex can only be the root $r_F$ or descendant $d_F$ of $F$ as they are the only vertices that connect directly to vertices outside the fragment. Assume w.l.o.g that $r_F$ is the closest vertex to $v_{j+1}$. It follows that the intersection $S^t_j \cap F$ is the tree path between $r_F$ and $v_j$. As we know the complete structure of $F$ and the cover values of all edges in the fragment, we can check which edges in this path have cover value $c^t_j$, as needed.}
\end{itemize}
%Doing this computation for all tuples $(S^t_j,c^t_j)$ results in the computation of $\fragInfo(t,F)$.
\end{proof}

\subsection{The Algorithm for Non-Highways}

We next describe the algorithm for learning $\cutInfo(t)$, $\fragInfo(t,F)$ and $\fragInfo_h(t,F)$ for non-highways. For simplicity, we first present the algorithm for the case that each tree edge $t$ already knows the tree path $P_t$, and we later explain how to overcome this assumption.

\begin{lemma} \label{nh_cut_info_lem}
Assume that each non-highway edge $t$ knows the path $P_t$. Then all the non-highway tree edges $t$ can learn $\cutInfo(t)$ in $O(k\sqrt{n})$ time. 
\end{lemma}

\begin{proof}
Fix a non-highway tree edge $t=\{u,p(u)\}$. We start by computing $\cutInfo_h(t)$. 
Following the proof of Lemma \ref{min_cuts_lemma}, to compute the tuples $\{(S^t_i,c^t_i)\}$ in $\cutInfo_h(t)$ we should sort the edges $e$ covering $t$ according to the vertices $v_h(e,P_t)$, from the closest to furthest from $t$, and learn about the $k$ first edges $e_1,...,e_k$ in the sorted order and the corresponding vertices $v_1=v_h(e_1,P_t),...,v_k=v_h(e_k,P_t)$. Then, we define $S^t_i$ to be the subpath between $v_i$ to $v_{i+1}$, and $c^t_i = \cov(t) + k - 2i$. The value $\cov(t)$ is known to $t$ by Claim \ref{claimLearnCov}, hence to compute $\{(S^t_i,c^t_i)\}$ it just needs to learn about the edges $e_1,...,e_k$ and vertices $v_1,...,v_k$. 

To do so, first note that given the LCA labels of the edges $e,t$, and the endpoints of the path $P_t$ we can check if $e$ covers $t$ and if so, compute the LCA label of the vertex $v_h(e,P_t)$ using Claims \ref{claimCheckCovLCA} and \ref{claimLCAcomputeEndpoints}. 
We can also use the LCA labels to compare two vertices $v=v_h(e,P_t),v'=v_h(e',P_t)$ and see which one is closer to $t$, using Claim \ref{claimCompare}.
Moreover, as $t$ is a non-highway edge, any edge that covers it has one endpoint in the subtree $T_u$ below $t$ (rooted at $u$) that is in the same fragment of $t$. Hence, for each such edge, there is one vertex in $T_u$ that knows about it. This suggests the following algorithm. Each vertex in $T_u$ first looks at the edges adjacent to it, checks which ones cover $t$, and compare them according to the vertices $v_h(e,P_t)$. Then, we start sending suggestions for $t$, a vertex of height $j$ in the tree, sends to its parent in round $j+i$ the $i$'th best edge it knows about according to the sorted order, including edges adjacent to it and edges it received from its children. The $k$ first edges in the sorted order (and their corresponding vertices $v_h(e,P_t)$) get to $u$ in at most $O(\sqrt{n}+k)$ rounds. For this, note that a vertex of height $j$, learns about the $i$'th best edge in its subtree until round $j+i$, as it can only be delayed by $i-1$ better edges. Hence, in $O(\sqrt{n}+k)$ time $t$ learns $\cutInfo_h(t)$. The whole computation is inside the fragment of $t$.

To let all non-highway edges $t$ learn $\cutInfo_h(t)$, we work in all fragments simultaneously. In each fragment $F$, first all non-highway edges $t$ send to the whole fragment the LCA labels of $t$ and the endpoints $\{p_l,p_h\}$ of $P_t$, this is $O(\sqrt{n})$ pieces of information hence takes $O(\sqrt{n})$ time. Then, using pipelining, we let each edge $t$ learn about the $k$ best edges that cover it as described above. This takes $O(\sqrt{n} + k\sqrt{n})=O(k \sqrt{n})$ time overall, as each one of the $O(\sqrt{n})$ edges should learn $O(k)$ pieces of information. 

Computing $\cutInfo_l(t)$ can be done exactly in the same manner, but now comparing edges according to $v_l(e,P_t)$ instead of $v_h(e,P_t)$.  
\end{proof}

For a tree edge $t$, we denote by $F_t$ the fragment of $t$. 

\begin{lemma} \label{lem_nh_frag_info}
Assume that each non-highway edge $t$ knows the path $P_t$. Then all the non-highway tree edges $t$ can learn $\fragInfo(t,F)$ and $\fragInfo_h(t,F)$ for all fragments $F$ in $O(D+k\sqrt{n})$ time.
\end{lemma}

\begin{proof}
The general proof idea is as follows. For each non-highway tree edge $t$, we define $O(k)$ special fragments where the segments $S^t_i$ start or end. With these special fragments $F$ we would like to talk directly via an edge between $F_t$ and $F$ (if exists) and use it to compute $\fragInfo(t,F)$ and $\fragInfo_h(t,F)$. 
For non-special fragments or special fragments where there is no direct edge between $F_t$ and $F$, we show that it is enough to broadcast $O(\sqrt{n})$ information to the whole graph to compute the values $\fragInfo(t,F)$ and $\fragInfo_h(t,F)$. We next describe the algorithm in detail.  

\paragraph{Step 1: Broadcasting information.}
We broadcast the following information. First, in each fragment $F$, we let all vertices in $F$ learn $\{t,\cov(t),\cutInfo(t)\}$ for all tree edges $t \in F$, we also learn which ones of the tree edges are highways. As this is at most $O(k \log{n})$ bits of information per edge, we can broadcast the whole information inside each fragment in $O(k\sqrt{n})$ time.
Second, we let all the vertices in the graph learn the following information about each fragment. Let $c_F$ be the minimum cover value of an edge in the highway $h_F$ of $F$, and let $n_F$ be the number of edges $t$ in $h_F$ such that $\cov(t)=c_F$. The values $(c_F,n_F)$ can be computed locally in $F$ in $O(\sqrt{n})$ time by scanning the highway.
We let all vertices in the graph learn $(c_F,n_F)$ for all fragments $F$. As there are $O(\sqrt{n})$ fragments, broadcasting this information over a BFS tree takes $O(D+\sqrt{n})$ time from Claim \ref{claim_broadcast}. 

\paragraph{Step 2: Special fragments.}  
For a non-highway tree edge $t$, we define $O(k)$ special fragments as follows. First, $F_t$ is a special fragment, as well as the first and last fragments in the path $P_t$. Second, recall that each segment $S^t_i$ is a segment between two vertices $v_i$ and $v_{i+1}$. Let $F_i$ be the fragment of $v_i$. If $v_i$ is an internal vertex in the fragment (i.e., not the root or the descendant of the fragment), the fragment $F_i$ is special. Overall we have $O(k)$ special fragments, as we have $O(k)$ segments $S^t_i$. Note that $t$ knows the fragments $F_i$ as it knows the LCA labels of the vertices $v_i$, and the fragment $F_i$ is part of the label of $v_i$ (see Section \ref{sec:lca}). 

%Recall that when we defined the segments $S^t_i$ we found $k$ edges $e_1,...,e_k$ and corresponding vertices $v_1=v_h(e_1,P_t),...,v_k=v_h(e_k,P_t)$ such that $S^t_i$ is the segment between $v_i$ to $v_{i+1}$.
%For a non-highway tree edge $t$, we define $O(k)$ special fragments as follows. Recall that when we defined the segments $S^t_i$ we found $k$ edges $e_1,...,e_k$ and corresponding vertices $v_1=v_h(e_1,P_t),...,v_k=v_h(e_k,P_t)$ such that $S^t_i$ is the segment between $v_i$ to $v_{i+1}$. 

For each special fragment $F$, we would like to use an edge between $F_t$ to $F$ to compute the values $\fragInfo(t,F)$ and $\fragInfo_h(t,F)$. We first focus on special fragments $F$ where there is indeed an edge between $F_t$ and $F$, and later discuss the case there is no edge. For technical reasons, we focus on the existence of \emph{internal edges} between the fragments, where we say that an edge between two fragments is \emph{internal} if it connects two internal vertices in the fragments. A useful property of internal edges is that such edges that connect different pairs of fragments are disjoint, as any internal vertex only belongs to one fragment. Note that we can check in $O(\sqrt{n})$ time if there is an internal edge from the fragment $F_t$ to each one of the $O(\sqrt{n})$ fragments in the graph, and if so find such an edge, by aggregate computations in the fragment. Let $e=\{v,w\}$ be an internal edge connecting $v \in F_t$ and $w \in F$. Since $v \in F_t$ it knows $\cutInfo(t)$ as we broadcast this information to the whole fragment. Also, as $w \in F$ it knows the complete structure of $F$ and all the cover values of edges in the fragment, as we broadcast this information to the whole fragment. By Claim \ref{claimFrag}, using this information $v$ and $w$ can deduce $\fragInfo(t,F)$ and $\fragInfo_h(t,F)$. This only requires $v$ to send to $w$ the value $\cutInfo(t)$, and then $w$ can compute $\fragInfo(t,F)$ and $\fragInfo_h(t,F)$ and send it to $v$ that can later send it to $t$. 

For each special fragment $F$, we use an internal edge to $F_t$ if exists to compute $\fragInfo(t,F)$ and $\fragInfo_h(t,F)$. Then we send the information computed to $t$ by broadcast in the fragment. Note that each tree edge $t$ should learn at most $O(k)$ values, hence overall we can broadcast the whole information to all $O(\sqrt{n})$ tree edges $t$ in the fragment in $O(k \sqrt{n})$ time. We can work in parallel in different fragments, as internal edges connecting different fragments are disjoint, resulting in $O(k \sqrt{n})$ complexity for this part.

\paragraph{Step 3: Special fragments with no internal edge.} Here we show that if there are no internal edges between $F_t$ and $F$, then it is easy to compute $\fragInfo(t,F)$ and $\fragInfo_h(t,F)$. This is the only step in the proof where we use the fact that $t$ is a non-highway edge. First, if $t' \in F$ such that $\cut(t,t')=k$, then $t'$ is a highway edge. To see this, assume to the contrary that $t'$ is a non-highway edge, note that non-highway edges in different fragments are orthogonal to each other from the structure of the fragment decomposition. If $\{t,t'\}$ is a min 2-respecting cut, then from Claim \ref{subtreeEdge} there is an edge between the subtree below $t$ and the subtree below $t'$. However, these trees are composed of internal vertices of the fragment (as $t$ and $t'$ are non-highways) implying there is an internal edge between $F_t$ and $F$, a contradiction. 

Also, for any non-tree edge $e$ that covers $t$, $e$ either covers all the highway of $F$ or covers none of the edges in $F$. This holds as one of the endpoints of $e$ is an internal vertex in $F_t \neq F$ (as it is in the subtree below $t$) and the other is not an internal vertex in $F$, as there are no internal edges between the fragments. Hence, both the endpoints of $e$ are not internal vertices in $F$. Then, from Claim \ref{globalCover}, $e$ either covers exactly the highway of $F$ or does not cover any edge of $F$.
This implies that the value $\cov(t,t')$ for $t'$ in the highway of $F$ does not depend on the specific choice of $t'$, denote this value by by $\cov(t,F)$. This value can be computed by $t$ as follows. This is the cost of all edges that cover $t$ and $h_F$. We can run an aggregate computation in $F_t$ to sum the costs of these edges: at the end each vertex $u \in F_t$ learns the cost of edges that have an endpoint below $u$ in $F_t$ and cover $h_F$. For the vertex $u$ such that $t=\{u,p(u)\}$ this is exactly $\cov(t,F)$. Hence all non-highway edges $t \in F$ learn $\cov(t,F)$. Moreover, we can compute these values for all fragments $F$ where there are no internal edges between $F_t$ to $F$, as there are $O(\sqrt{n})$ fragments, the whole computation takes $O(\sqrt{n})$ time using pipelining. Now $\cut(t,t')=k$ for $t' \in h_F$ iff $\cov(t)+\cov(t')-2\cov(t,t')=k$. As $\cov(t,t')=\cov(t,F)$, this holds iff $\cov(t')=k-\cov(t)+2\cov(t,F)$. Moreover, if $\{t,t'\}$ is min 2-respecting cut, then $\cov(t')=c_F$ must be the minimum cover value of an edge in $h_F$, as otherwise, for an edge with $\cov(t'')=c_F$, we get that $\cut(t,t'')=\cov(t)+c_F-2\cov(t,F) < \cov(t) + \cov(t') - 2\cov(t,F) = k$, contradiction to the fact $k$ is the minimum cut value. To conclude, to compute $\fragInfo(t,F)$ that is equal to $\fragInfo_h(t,F)$ in this case, $t$ should learn how many edges in $h_F$ have $\cov(t')=k-\cov(t)+2\cov(t,F)=c_F$. If $k-\cov(t)+2\cov(t,F) \neq c_F$, there are no such edges. Otherwise, the number of such edges is exactly $n_F$, the number of edges in $h_F$ with $\cov(t')=c_F$. As all vertices know $c_F$ and $n_F$ for all fragments, then $t$ can compute this locally using $\cov(t)$ and $\cov(t,F)$.

\paragraph{Step 4: Non-special fragments.} 
We first show that if $F$ is a non-special fragment for $t$, then the only edges $t' \in F$ that may be part of a min 2-respecting cut $\{t,t'\}$ are highway edges. Additionally, we show that the highway of $F$ is either completely contained in one of the segments $S^t_i$ or its intersection with all segments $S^t_i$ is empty (in which case $\fragInfo(t,F)$=0). For this, first look at the tree path $P_t$, this path contains all the edges $t'$ such that $\cut(t,t')=k$ by definition. From the structure of the skeleton tree, any tree path in the graph between two fragments $F_1$ and $F_2$ can only contain complete highways of other fragments. Hence, except of the first and last fragments of $P_t$ that are special, the intersection of $P_t$ with any other fragment $F$ is either $h_F$ or empty. Moreover, from the structure of the skeleton tree and from the first and last fragments of $P_t$, $t$ can deduce all the fragments that intersect the path $P_t$. Let $F$ be a non-special fragment such that $F \cap P = h_F$. As $F$ is non-special, we know that all vertices $v_i$ are not internal vertices of $F$. If we look at the path $P_t$, as it contains $h_F$, we can find an index $i$ such that $v_i$ is in $P_t$ before $h_F$, and $v_{i+1}$ is in $P_t$ after $h_F$. Or in other words, there is a unique segment $S^t_i$ such that $h_F \subseteq S^t_i$. Moreover, $t$ can compute this segment as it knows the fragments of all vertices $v_i$. Now, from Lemma \ref{min_cuts_lemma}, we have that for $t' \in S^t_i$, $\cut(t,t')=k$ iff $\cov(t')=c^t_i$. Moreover, for any edge in $S^t_i$, we have $\cov(t') \geq c^t_i$. For this, first note that $\cov(t,t')$ is fixed for all edges in $S^t_i$ as explained in the proof of Lemma \ref{min_cuts_lemma}. Hence if we have an edge $t'' \in S^t_i$ where $\cov(t'') < c^t_i$, we get that $\cut(t,t'')=\cov(t) + \cov(t'') - 2\cov(t,t'') < \cov(t) + c^t_i - 2\cov(t,t') = \cut(t,t')= k$, where $t'$ is an edge in $S^t_i$ with $\cut(t,t')=k$, contradicting the fact that $k$ is the min cut value. To conclude, we need to check how many edges $t' \in h_F$ have $\cov(t')=c^t_i$. If the minimum cover value $c_F$ of an edge in $h_F$ is greater than $c^t_i$, there are no such edges. Otherwise $c_F=c^t_i$ (it cannot be smaller from the discussion above), and we have that $\fragInfo(t,F)=n_F$. As $t$ knows the values $c_F,n_F$ for all fragments and the values $c^t_i$, it can deduce $\fragInfo(t,F)$ for all non-special fragments. As explained, all edges $t'$ in a min 2-respecting cut with $t$ in this case are highway edges, hence $\fragInfo_h(t,F)=\fragInfo(t,F)$ in this case.
\end{proof}

\paragraph{Dealing with the case that $P_t$ is not initially known.}

During the proof we assumed for simplicity that the path $P_t$ is known to the edge $t$. We now explain how to extend the algorithm to the case it is not known. In \cite{DBLP:journals/corr/abs-2004-09129}, the authors show how each tree edge learns a set of $\poly(\log{n})$ paths that includes the path $P_t$ (the path $P_t$ is called the \emph{interesting path} for $t$, where the paths found are called the \emph{potentially interesting} paths). The time complexity of the algorithm in \cite{DBLP:journals/corr/abs-2004-09129} is $\tilde{O}(D+\sqrt{n})$. We can extend our algorithm to work with this information, with a poly-logarithmic overhead in the complexity. Here, we just repeat the algorithm poly-logarithmic number of times with all the possible choices for the path $P_t$. Note that if we run the algorithm for $t$ with a path $P \neq P_t$, it follows from the proof that if we learn about a tuple $(S^t_i,c^t_i)$, all the edges $t' \in S^t_i$ where $\cov(t')=c^t_i$ are such that $\cut(t,t')=k$. The only difference is that if we work with a path $P \neq P_t$ we are not guaranteed that the tuples computed contain all the tree edges $t'$ where $\cut(t,t')=k$, where for the right choice for $P_t$ we are guaranteed to find all cuts. 
Since in the algorithm we compute $\fragInfo(t,F)$ for all fragments $F$, each tree edge $t$ can compute the number of min 2-respecting cuts $\{t,t'\}$ found for any choice of path $P$. Then it chooses $P_t$ as the path where this number is maximal. This path $P_t$ necessarily contains all edges $t'$ where $\cut(t,t')=k$, because we know that all the cuts found are indeed min cuts, and we know that one of the paths $t$ tries is the right path $P_t$. Trying poly-logarithmic number of options adds poly-logarithmic overhead to the complexity. Note that from the discussion above, after trying all options the tree $t$ learns about the path $P_t$, giving the following. 

\begin{lemma} \label{cut_info_pt}
All the non-highway tree edges $t$ can learn $\cutInfo(t)$, $\fragInfo(t,F)$ and $\fragInfo_h(t,F)$ for all $F$, and the identity of the path $P_t$ in $\tilde{O}(D+k\sqrt{n})$ time.
\end{lemma}

\subsection{The Algorithm for Highways}

We next explain how highway edges learn the cut information. 
As explained above, in \cite{DBLP:journals/corr/abs-2004-09129} it is shown that each edge $t$ learns a set of $O(\log^2{n})$ paths that contains the path $P_t$. 

\begin{lemma}[\cite{DBLP:journals/corr/abs-2004-09129}] \label{lem:highway_internal}
In $\tilde{O}(D+\sqrt{n})$ time, each tree edge $t$ learns about a set of paths of size $O(\log^2{n})$ that contains the path $P_t$.
\end{lemma}

For the highway case, we need a stronger guarantee that also follows from \cite{DBLP:journals/corr/abs-2004-09129}. 

\begin{lemma}[\cite{DBLP:journals/corr/abs-2004-09129}] \label{lem:highway_bound}
Let $F$ be a fragment with highway $h_F$. The edges of each highway can learn a set of paths $Q_h$ of size $O(\log^2 {n})$ that contains all the paths $\{P_t^h | t \in h_F, v_h \not \in F \}$ and a set of paths $Q_l$ of size $O(\log^2 {n})$ that contains all the paths  $\{P_t^l | t \in h_F, v_l \not \in F \}$. The time complexity is $\tilde{O}(D+\sqrt{n})$.
\end{lemma} 

This lemma follows as a corollary of Lemma 5.13 in \cite{DBLP:journals/corr/abs-2004-09129} as is also explained in \cite{DBLP:journals/corr/abs-2004-09129}.
Lemma \ref{lem:highway_bound} would allow us to deal efficiently with paths $P_t$ that have endpoints outside $F_t$, dealing with paths that are contained in $F_t$ will be similar to the non-highway case.

In our algorithm, initially tree edges do not know $P_t$, but just a poly-logarithmic number of options for the path $P_t$, that include the correct path $P_t$. We denote by $\cutInfo(t,P)$ the value $\cutInfo(t)$ when $t$ assumes that the path $P_t=P$. For the correct option $\cutInfo(t,P_t)=\cutInfo(t)$. Similarly we define $\fragInfo(t,F,P)$ to be $\fragInfo(t,F)$ where $t$ assumes that $P_t=P$. During our algorithm each tree edge computes the values $\cutInfo(t,P)$, $\fragInfo(t,F,P)$ and $\fragInfo_h(t,F,P)$ for $O(\log^2{n})$ options for the path $P$ that contain the correct choice $P_t$. All the values computed indeed describe min cuts containing $t$, but if $P \neq P_t$, we may not find all cuts. Eventually, $t$ can identify the correct path $P_t$ as it chooses the path where $\sum_{F} \fragInfo(t,F,P)$ is maximal.

For the tree edge $t$ we denote by $Q_t$ the set of $O(\log^2{n})$ potential candidates for $P_t$ that $t$ learns about in Lemmas \ref{lem:highway_internal} and \ref{lem:highway_bound}. 

\begin{lemma} \label{lem:cut_info_highway}
In $\tilde{O}(D+k\sqrt{n})$ time, each highway tree edge $t$ can learn $\cutInfo(t,P)$ for each $P \in Q_t$. 
\end{lemma}

\begin{proof}
Fix a highway tree edge $t=\{u,p(u)\}$ in a fragment $F_t$. 
We focus on learning $\cutInfo_h(t,P)$, learning $\cutInfo_l(t,P)$ can be done in a similar manner. We divide into cases depending if the endpoint $v_h$ of $P$ is inside $F_t$ or not. 
 
\paragraph{Case 1: $v_h \in F_t$.}
Here for each tree edge $t$ we fix one path $P \in Q_t$ with $v_h \in F_t$ (if exists). Repeating the algorithm $O(\log^2{n})$ times would cover all paths $P$ in this case, and would add $O(\log^2{n})$ term to the complexity.
%In case $v_h \in F$, we work as follows. 
As explained before, our goal is to learn about $k$ edges $e_1,e_2,...,e_k$ and corresponding endpoints $v_1,...,v_k$ that define the segments $S^t_i$. These edges may have an internal endpoint inside $F_t$ or not. We first learn about $k$ relevant edges $e'_1,...,e'_k$ that have at least one internal endpoint in $F_t$. Here our goal is to learn about $k$ edges that cover $t$ and their endpoints $v_h(e,P)$ are the closest to $t$. This can be done exactly as in the proof of the non-highway case (Lemma \ref{nh_cut_info_lem}), and takes $O(k \sqrt{n})$ time for all tree edges $t$. We next consider edges that both their endpoints are not internal vertices in $F_t$, we call them \emph{global} edges. Here we use Claim \ref{globalCover} that states that every edge that both its endpoints are not internal vertices in $F_t$ either covers the complete highway of $F_t$ or does not cover any edge of $F_t$. This means that for any global edge that covers $t$, the part of $P^h$ it covers is exactly $P^h \cap h_F$, where $h_F$ is the highway of $F_t$, this path is exactly the segment from $p(u)$ to $LCA(p(u),v_h)$ (this contains exactly all highway edges above $t=\{u,p(u)\}$ in $P_t^h$). Let $v'=LCA(p(u),v_h)$. To determine the segments $S^t_i$, we want $t$ to learn how many global edges cover $t$, denote this number by $global_t$. This information together with the edges $e'_1,...,e'_k$ computed before allows to determine the vertices $v_1,...,v_k$ that define the segments $S^t_i$. To do so, we sort the vertices $v_h(e'_1,P),...,v_h(e'_k,P),v'$ according to their closeness to $t$, where the vertex $v'$ appears $global_t$ times (once per each relevant edge that covers the segment between $t$ and $v'$). Finally we take the $k$ closest endpoints according to this order, and they define the segments $S^t_i$. Note that the only information we needed to learn about global edges is the number of edges that cover the whole highway $h_F$, this information is not specific for $t$. To learn this information we can use global communication over a BFS tree to sum the number of such edges. Doing this for all fragments $F$ takes $O(D+\sqrt{n})$ time. Hence, overall the computation in this case takes $\tilde{O}(D+k\sqrt{n})$ time.

\paragraph{Case 2: $v_h \not \in F_t$.}
We next consider the case that $v_h$ is outside $F_t$. We first explain how to handle one highway edge $t$ and fix one path $P \in Q_t$ from this case, and then show how to do the computation efficiently for all highway edges.
To determine the segments $S^t_i$, exactly as the previous case we start by learning the $k$ closest edges to $t$ that cover $P^h$ and have at least one internal endpoint in $F_t$, denote them by $e'_1,...,e'_k$. We next turn our attention to edges that cover $P^h$ and both their endpoints are not internal vertices in $F_t$. Each such edge covers the whole highway $h_F$. For each such edge we can compute $v_h(e,P^h)$.
We can use a global communication over a BFS tree to learn the $k$ global edges where the vertices $v_h(e,P^h)$ are the closest to $F_t$. This takes $O(D+k)$ time for one path $P^h$. 
At the end of the computation, the tree edge $t$ can determine the segments $S^t_i$ by sorting the endpoints $\{v_h(e'_i,P)\}_{1 \leq i \leq k}$ it learned about locally, and the endpoints $v_h(e,P)$ of the $k$ global edges closest to $F_t$ it learns in the global communication. Taking the $k$ closest vertices defines the segments $S^t_i$ as needed. 

We next explain how to do the computation efficiently for all highway edges. As in previous cases, finding the local edges can be done for all edges in $O(k\sqrt{n})$ time, the main challenge is finding the global edges. If for each edge $t$, there is a different path $P^h$ to consider, this part can take linear time, which is too expensive. Here we exploit Lemma \ref{lem:highway_bound} that implies that there is a set $Q_h$ of $O(\log^2{n})$ size that contains all the paths $\{P_t^h | t \in h_F, v_h \not \in F \}$. Hence, we just need to consider $O(\log^2{n})$ paths for all different highway edges in $h_F$.
%The main ingredient to overcome this, is a structural lemma from \cite{DBLP:journals/corr/abs-2004-09129}, that implies that the set $\{P_t^h | t \in h_F,  v_h \not \in F_t \}$ has a poly-logarithmic size. In other words, there are only poly-logarithmic different paths $P_t^h$ to consider for all highway edges of the fragment $F_t$. 
Hence, and since there are $O(\sqrt{n})$ fragments, we can compute $S^t_i$ for all highway edges in $\tilde{O}(D+k\sqrt{n})$ time via pipelining. 
%Here we use \mtodo{add} that imply that there are only $polylog(n)$ paths $P_t^h$ such that $v_h$ is outside $F_t$ that include all the paths $P_t^h$ for all $t \in h_F$ from this case. The fact we only need to consider small number of paths for all the tree edges in the highway, would be crucial to obtain an efficient algorithm for this case. 
%To determine the segments $S^t_i$, exactly as the previous case we start by learning the $k$ closest edges to $t$ that cover $P_t^h$ and have at least one endpoint in $F$, denote them by $e'_1,...,e'_k$. We next turn our attention to edges that cover $P_t^h$ and have both endpoints outside $F$. Each such edge covers the whole highway $h_F$. For each such edge we can compute $v_h(e,P_t^h)$.
%We can use a global communication over a BFS tree to learn the $k$ global edges where the vertices $v_h(e,P_t^h)$ are the closest to $F_t$. This takes $O(D+k)$ time for one path $P_t^h$. 
%Since there are only $polylog(n)$ paths $P_t^h$ we need to consider for each fragment, repeating the computation for all fragments and all relevant paths $P_t^h$ takes $\tilde{O}(D+k\sqrt{n})$ time via pipelining. 
%At the end of the computation, each tree edge $t$ can determine the segments $S^t_i$ by sorting the endpoints $\{v_h(e'_i)\}_{1 \leq i \leq k}$ it learned about locally, and the endpoints $v_h(e,P_t^h)$ of the $k$ global edges closest to $F_t$ it learns in the global communication. Taking the $k$ closest vertices defines the segments $S^t_i$ as needed.  
\end{proof}

\begin{lemma}
All the highway tree edges $t$ can learn $\fragInfo(t,F,P)$ and $\fragInfo_h(t,F,P)$ for all fragments $F$ and $P \in Q_t$ in $\tilde{O}(D+k\sqrt{n})$ time.
\end{lemma}

\begin{proof}
For the proof, we again first fix one path $P \in Q_t$ for each edge $t$.
The proof generally follows the proof of Lemma \ref{lem_nh_frag_info}, most of the steps in the proof apply for highway tree edges as well. In particular, dealing with special fragments where there is a direct internal edge to them, or dealing with non-special fragments can be done exactly the same. The only difference is in Step 3 which we describe now. 
In Step 3, the tree edge $t$ learns $\fragInfo(t,F,P)$ for special fragments $F$ where there is no internal edge between $F_t$ and $F$. We focus on learning about min 2-respecting cuts $\{t,t'\}$ where $t' \in P^h$. Dealing with the case that $t' \in P^l$ can be done in the same manner. Note that we only need to consider edges $t$ where $P^h$ ends outside $F_t$, as otherwise the claim trivially holds. We first show that there is only a limited number of options for the identity of the fragment $F$. If $F \neq F_t$ is a special fragment for $t$ where $F \cap P^h \neq \emptyset$ it can be either the last fragment $F_{last}$ of $P^h$, or there is at least one internal vertex $v_i \in F$ that is an endpoint of one of the segments $S^t_i$. We first deal with the latter case.

\paragraph{Case 1: $F \neq F_{last}$.} Here, by the definition of the segments $S^t_i$, there is an edge $e$ that covers $t=\{u,p(u)\}$ where $v_h(e,P^h)=v_i \in F$. We show that if $F \neq F_{last}$, then $e$ must be a global edge found in the global communication phase described in the proof of Lemma \ref{lem:cut_info_highway}. 
Assume to the contrary that $e=\{v_l,v_h\}$ is an edge that has an internal endpoint in $F_t$. Note that the second endpoint of $e$ cannot be an internal vertex in $F$, as there are no internal edges between $F_t$ and $F$. Let $p_h$ be the endpoint of $P^h$ in $F_{last}$, and $v_h$ be the endpoint of $e$ that is above or orthogonal to $t$. From the discussion above $v_h$ is not an internal vertex in $F$.
Recall that by the definition of $v_i=v_h(e,P^h)$ it can be equal either to $a_e = LCA(p(u),v_h), a_P = LCA(p(u),p_h)$ or $LCA(p_h,v_h)$, where $p(u) \in F_t \neq F, p_h \in F_{last} \neq F$ and $v_h$ is not an internal vertex in $F$, in all cases we get that $v_h(e,P^h)$ cannot be an internal vertex in $F$, as the LCA of two vertices that are not internal vertices in $F$ cannot be an internal vertex in $F$ by Claim \ref{lcaFrag}, a contradiction.

To conclude, we have that if $v_h(e,P^h) \not \in F_{last}$ then it is an endpoint of one of the $k$ global edges covering $t$ that were found in the global communication phase. Note that the identity of $F_{last}$ and of the $k$ global edges depends only on $P^h$, and not on the specific choice of $t$. Hence, for any $P^h$ there are only at most $O(k)$ options for the fragment $F$. We call these fragments the special fragments of $P^h$. For any special fragment of $P^h$, we would like to learn globally the following information. Let $e^g_1,...,e^g_k$ be the global edges found in the global communication phase, and let $u_i=v_h(e^g_i,P^h)$. We first broadcast the information $(P^h,F_t,u_1,...,u_k)$ to the whole graph, which allows the whole graph learn the identity of the special fragments of $P^h$. For a special fragment $F \neq F_{last}$ of $P^h$ where there are no internal edges between $F$ and $F_t$, let $u_{i_1},...,u_{i_t}$ be the vertices of $u_1,...,u_k$ inside $F$ (there is at least one by the definition of special fragment), ordered by closeness to $t$. Let $a_F$ and $ b_F$ be the closest and furthest vertices from $t$ on the highway of $F$ (one is the root and one is the unique descendant of the fragment). We look at the subpaths $a_F - u_{i_1}, u_{i_1}-u_{i_2},...,u_{i_t}-b_F$
and for each one of those subpaths $S$, we find the minimum cover value of the subpath, $c_S$, and the number of times it appears in the subpath $S$. Each vertex in $F$ can learn this information locally as it knows the cover values of all edges in $F$. We then broadcast the information $(S,c_S,n_S)$ for all relevant subpaths to the whole graph. Note that for any path $P^h$ we overall broadcast $O(k)$ messages of $O(\log{n})$ bits, and since there are $O(\log^2{n})$ paths $P^h$ relevant for each fragment from Lemma \ref{lem:highway_bound}, the whole communication takes $\tilde{O}(D+k\sqrt{n})$ time. 

We now show that this information suffices for a tree edge $t$ to learn $\fragInfo(t,F,P)$ and $\fragInfo_h(t,F,P)$, for a special fragment $F \neq F_{last}$ (with respect to $P$) with no internal edges between $F_t$ and $F$. First note that the paths $P^h$ that could be relevant for $t$ are among the $O(\log^2 {n})$ paths $P^h$ relevant for the fragment (described in Lemma \ref{lem:highway_bound}), hence $t$ learned information about them as described above. We next fix one such path $P^h$. As explained above, $F$ must be one of the special fragments of $P^h$. Let $S^t_i$ be a segment that intersects $F$. As explained above all the edges $e$ that cover $t$ where $v_h(e,P^h) \in F$ must be global edges. It follows that the segment $S^t_i \cap F$ is one of the subpaths $S$ of $F$ computed before, and that $t$ knows the values $(S,c_S,n_S)$. Now edges $t' \in S$ define a cut of size $k$ with $t$ iff $\cov(t')=c^t_i$. Also, for all edges in $S^t_i$ the value $\cov(t,t')$ is fixed, which implies that $\cov(t') \geq c^t_i$ (as $k$ is the min cut value). It means that there are edges in $S$ where $\cov(t')=c^t_i$ iff $c^t_i=c_S$, in which case the number of these edges is $n_S$. Hence $t$ can determine the number of min 2-respecting cuts in any segment $S^t_i \cap F$. Note that it can also check for any segment $S$ if it is contained in $S^t_i$ using the LCA labels of the endpoints of a segment. Note that all the cuts we learn about are cuts where $t'$ in on a highway in this case, hence $\fragInfo_h(t,F,P) = \fragInfo(t,F,P)$. This completes the case that $F \neq F_{last}.$

\paragraph{Case 2: $F = F_{last}$.}
We are left with the case that $F=F_{last}$ and there are no internal edges between $F$ and $F_t$ (as otherwise this case is already covered in Step 2). 
We would like to follow a similar algorithm with slight changes. Note that the edges $e_1,e_2,...,e_k$ that define the segments $S^t_i$ are either global edges $e^g_i$ or local edges with at least one internal endpoint in $F_t$. Let $e$ be a local edge that covers $t$ where $v_h(e,P^h)$ is an internal vertex in $F_{last}$. Note that the second endpoint of $e$ cannot be internal vertex in $F_{last}$, as there are no internal edges between $F_t$ and $F_{last}$. We claim that the vertex $v_h(e,P^h)$ is equal to $LCA(p_h,d_F)$, where $p_h$ is the endpoint of $P^h$ in $F_{last}$ and $d_F$ is the unique descendant of $F_{last}$. In other words, $v_h(e,P^h)$ does not depend on the specific choice of $e$ and $t$, just on the path $P^h$. To show this, note that both $e$'s endpoints are not internal vertices in $F_{last}$, but $e$ also covers some part of $F_{last}$ as $v_h(e,P^h)$ is an internal vertex in $F_{last}$. In this case, it follows from Claim \ref{globalCover} that $e$ covers exactly the highway $h_F$ of $F_{last}$. The intersection $h_F \cap P^h$ is the subpath of $P^h$ in the highway. Let $a_F$ be the closest vertex to $t$ on the highway (could be either the root $r_F$ or unique descendant $d_F$ of $F_{last}$), then $P^h \cap F_{last}$ is the subpath between $a_F$ to $p_h$. Its intersection with the highway is exactly the path between $a_F$ to $LCA(p_h,d_F)$, showing that $e$ covers exactly the subpath of $F_{last}$ between $a_F$ to $LCA(p_h,d_F)$, as needed. Denote by $u_L = LCA(p_h,d_F)$. Recall that $u_1,u_2,...,u_k$ were the endpoints $v_h(e^g_i,P^h)$ of global edges. Let $u_{i_1},...,u_{i_t}$ be the vertices of $u_1,...,u_k,u_L$ that have endpoint in $F_{last}$, sorted according to closeness to $t$. We again want $t$ to learn the information $(S,c_S,n_S)$ on each one of the segments $a_F-u_{i_1}$,$u_{i_2}-u_{i_3}$,...,$u_{i_t}-p_h$ where $c_S,n_S$ are defined as before. The computation can be done exactly as in the previous case, the only difference is that now $u_L$ is also one of the endpoints, and $p_h$ (the end of $P^h$) is the end of the last segment. As $u_L=v_h(e,P^h)$ is the endpoint of all local edges $e$ that cover $t$ where $v_h(e,P^h)$ is an internal vertex in $P^h$, and the vertices $u_i=v_h(e^g_i,P^h)$ are the endpoints of all relevant global edges, we have that all endpoints of the segments $S^t_i$ in $F_{last}$ are in the segments $S$ computed. Also, as the vertices $u_{i_j}$ do not depend on the specific choice of $t$, we can do the computation in $\tilde{O}(D+k\sqrt{n})$ time for all tree edges $t$, as described in the previous case. Note that $t$ can also distinguish between cuts where $t'$ is on a highway and cuts where it is not on a highway, as all the segments between $a_F$ to $u_L$ are on the highway, and the segments between $u_L$ and $p_h$ are non-highways, this allows to deduce both $\fragInfo_h(t,F,P)$ and $\fragInfo(t,F,P)$. This completes the proof.  
\end{proof}

To conclude, each highway edge $t$ learns the values $\cutInfo(t,P),\fragInfo(t,F,P)$ and $\fragInfo_h(t,F,P)$ for at most $O(\log^2{n})$ paths $P$ that contain the path $P_t$. To identify the correct path $P+t$, it chooses the path where $\sum_{F} \fragInfo(t,F,P)$ is maximal. This identifies the correct path $P_t$ as all cuts $t$ learns about are min 2-respecting cuts, and with the correct choice of $P_t$ it learns about all possible cuts, so the sum is maximal. By definition, $\cutInfo(t)=\cutInfo(t,P_t), \fragInfo(t,F) = \fragInfo(t,F,P_t)$ and $\fragInfo_h(t,F)=\fragInfo(t,F,P_t)$. This gives the following.

\begin{lemma} \label{cut_info_pt_highway}
All the highway tree edges $t$ can learn $\cutInfo(t)$, $\fragInfo(t,F)$ and $\fragInfo_h(t,F)$ for all $F$, and the identity of the path $P_t$ in $\tilde{O}(D+k\sqrt{n})$ time.
\end{lemma}

\section{Computing the Cost-Effectiveness} \label{sec:cost_ef}

We next explain how to exploit the cut information to compute cost-effectiveness of edges.
Let $H$ be a $k$-edge-connected subgraph we want to augment to be $(k+1)$-edge-connected, and let $e = \{u,v\} \not \in H$.
Recall that in order to estimate the cost-effectiveness of $e$ it is enough to estimate the number of min 2-respecting cuts that $e$ covers with respect to a spanning tree $T$ of $H$, as was shown in Lemma \ref{lem:reduction}. A min 2-respecting cut has at most 2 tree edges. While our main focus is on cuts that have exactly 2 edges, we first note that it is easy to compute the number of min 1-respecting cuts covered by $e$.

%Recall that our goal is to estimate the number of cuts of size $k$ in $H$ covered by $e$.
%To do so, we estimate the number of minimum 2-respecting cuts that $e$ covers in the trees of the tree packing. As $e \not \in H$, then $e$ is a non-tree edge in all these trees. We start by focusing on one tree $T$. 

\paragraph{1-respecting cuts.} It is easy to learn about all min 1-respecting cuts covered by $e$. 
As explained in Section \ref{pre_cov_cut}, the value of the 1-respecting cut defined by a tree edge $t$ is equal to $\cov(t)$. Hence, $t$ defines a min 1-respecting cut iff $\cov(t)=k$. Also, all tree edges know their cover value from Claim \ref{claimLearnCov}. From Claim \ref{claim_one_cut_cover}, the min 1-respecting cuts covered by $e$ are exactly those where $t \in P_e$ and $\cov(t)=k$. Let $m_t$ be an indicator variable that equals 1 iff $\cov(t)=k$. To learn about the number of min 1-respecting cuts in $T$ covered by $e$, it should compute $\sum_{t \in P_e} m_t$. From Claim \ref{aggregate_Pe}, all non-tree edges can learn these values in $\tilde{O}(D+\sqrt{n})$ time.

\paragraph{2-respecting cuts.} The rest of the section focuses on min 2-respecting cuts. Recall that from Claim \ref{claim_cut_cover}, $e$ covers a 2-respecting cut $\{t,t'\}$ iff exactly one of $t,t'$ is in $P_e$, where $P_e$ is the tree path between the endpoints of $e$. To compute cost-effectiveness we prove the following useful lemma.

\begin{restatable}{lemma}{LemSketch} \label{lemma_sketches}
Assume that for any tree edge $t$, and for any min 2-respecting cut of the form $\{t,t'\}$ there is a specific vertex in the fragment of $t$ that knows a unique $O(\log n)$-bit name of the cut, such that the same name is known also to one specific vertex in the fragment of $t'$. Then, in $\tilde{O}(D+\sqrt{n})$ time, all tree edges $t$ can learn a value $\sketch(t)$ with the following properties.
\begin{enumerate}
\item $\sketch(t)$ has size $\Theta(\log^3 n)$ bits.
\item For any tree path $P$, given $\oplus_{t \in P} \sketch(t)$, we can compute a value $\tilde{k}\in [k, 8k]$, where $k$ denotes the number of min 2-respecting cuts with exactly one edge in $P$, with high probability.
\end{enumerate} 
\end{restatable}

Note that once we compute sketches with the above properties, we can estimate the number of min 2-respecting cuts covered by all non-tree edges $e$ as follows. First, each edge $e$ learns the value $\oplus_{t \in P_e} \sketch(t)$, this can be done in $\tilde{O}(D+\sqrt{n})$ time for all edges using Claim \ref{aggregate_Pe} as XOR is an aggregate function and the sketches have poly-logarithmic size. From the second property of the sketches this gives a constant approximation to the number of min 2-respecting cuts with exactly one edge in $P_e$, which are exactly the min 2-respecting cuts covered by $e$.

The proof of the lemma is based on \emph{sketching} ideas and appears in Section \ref{sec:sketches}. To use the lemma, we need to show that for any edge $t$ and any min 2-respecting cut $\{t,t'\}$ there is a specific vertex in the fragment of $t$ that knows a unique name of the cut $\{t,t'\}$. This is proved in Section \ref{sec:cut_names}. At a high-level, we exploit the information $\cutInfo(t)$ to obtain this. To illustrate the idea, we first focus on cuts $\{t,t'\}$ where there is an edge $\{v,v'\}$ between the fragments $F,F'$ of $t,t'$. The vertices $v,v'$ can learn the information $\cutInfo(t),\cov(t)$ for all edges in their fragment in $O(k\sqrt{n})$ time. Then, by the definition of $\cutInfo(t)$, they can deduce the names of all cuts $\{t,t'\}$ where $t \in F, t' \in F'$ as needed. The more challenging case is that there is no edge between the fragments. Here we use the fact there is no edge to show that the min 2-respecting cuts have a very specific structure that allows tree edges to learn enough information about these cuts. In this case, the vertices do not know the original name of cuts $\{t,t'\}$ but we can generate unique name for each cut known both to $t$ and $t'$. See Section \ref{sec:cut_names} for the full details.

\subsection{Computing Sketches} \label{sec:sketches}

To prove Lemma \ref{lemma_sketches}, we first prove the following lemma. Later when we use it, for each min 2-respecting cut $\{t,t'\}$ we define a relation $r$.
%To compute cost-effectiveness, we rely on the following lemma for computing sketches. 
We note that the ideas used for building this lemma are standard in the streaming and sketching literature, see e.g., \cite{ahn2012analyzing}. However, we did not find a concrete result statement in the literature that provides our desired properties (i.e., approximation of count, linearity of sketches, and with explicit mention of $\poly(\log n)$ bits of shared randomness), without requiring some extra writing. Thus, instead, we here provide a simple and nearly self-contained proof of the lemma.

\begin{lemma}\label{lem:Sketches}
Suppose that we have a set of virtual pairwise relations, where each relation $r$ has a unique $O(\log n)$-bit name that is known to two vertices $v$ and $u$. There is a scheme that uses $O(\log^3 n)$ bits of shared randomness and then, for each relation $r$ generates a $\Theta(\log^3 n)$-bit  $\sketch(r)$, such that, we have the following: for any set $R$ of relations, from  $\oplus_{r \in R} \sketch(r)$, we can deduce an approximation $\tilde{k} \in [k, 8k]$ for $k=|R|$.
\end{lemma}
\begin{proof}
Before describing the sketches, we first generate a random identifier for each  relation. More concretely, for each  relation $r$, we generate a $C\log n$-bit random identifier $RID(r)$, for a large enough constant $C$, where each bit is set to be $0$ or $1$ using an \emph{$\epsilon$-bias space} with $\epsilon=1/4$~\cite{naor1993small}. This guarantees the following two properties for these random identifiers: (1) Sharing $O(\log^2 n)$ bits of randomness among all vertices suffices to generate these random identifiers, as we use $O(\log n)$ bits of shared randomness for each bit of the random identifiers. Since we shared the randomness, the two vertices $v$ and $u$ that know the unique identifier of the relation $r$ generate the same random identifier for $r$. (2) For any particular non-empty set of relations $R'$, the probability that the bit-wise XOR of the random identifiers $\oplus_{r \in R'} RID(r)$  equals a vector of all zeros is at most $(3/4)^{C\log n} \leq 1/n^{10}$.

\paragraph{Sketches.} The sketches are built as follows: The sketch has $\log n$ density levels, and $O(\log n)$ identical repetitions per level. At level $i$, in each repetition, we sample each  relation to be \emph{active} with probability $1/2^{i}$ and inactive otherwise. It suffices if the active/inactive choices of different relations are merely pairwise independent (instead of fully independent) and this allows us to use only $O(\log n)$ bits of shared randomness for each repetition\footnote{Here is a sketch of one standard method for building such pairwise independent variables: Suppose that the relation identifiers are in $\{1, 2, \dots, N\}$, where $N\in n^{O(1)}.$ Let $K = \max\{N, n^{10}\}$ and let $K'$ be a prime in $[K, 2K]$. We pick two independent random numbers $a, b$ in $\{0, \dots, K'-1\}$, using $O(\log n)$ bits. Then, for each relation with identifier $x\in \{1, 2, \dots, N\}$, define $h(x)=ax+b \textit{ mod } K'$. These $h(x)$ values provide $N$ many random variables which are pairwise independent, and each is uniformly distributed in $\{0, \dots, K' \}$. Then, we can say relation number $x$ is active iff $h(x)\leq K' \cdot 2^{-i}$. This gives our desired active/inactive random variable modulo a negligible additive $1/n^{10}$ error in the probability.} so that the two endpoints of any relation know whether the relation is active or not. The sketch is then defined as follows: for each relation $r$, $\sketch(r)$ is made of $O(\log^2 n)$ parts, each of length $C\log n$: each part is set either equal to $RID(r)$, if relation $r$ is active in the corresponding repetition of that density level, and is otherwise set equal to a $C\log n$-bit string of all zeros. Notice that overall, we have used $O(\log^2 n)$ bits of shared randomness to generate the random identifiers, and $O(\log^3 n)$ bits of shared randomness to determine the active/inactive relations, where we used $O(\log n)$ bits for each of the $O(\log n)$ repetitions of each of the $O(\log n)$ density levels.

\paragraph{Count Approximation Given a Linear Mixture of Sketches.} Given $\oplus_{r \in R} RID(r)$, to output an approximation of $k=|R|$, we compute $i^*$ which is the largest density level for which at least a $1/4$  fraction of the repetitions in the density level $i$ turned out to be non-zero. We then output $2^{i^*}$ as our approximation.

\paragraph{Approximation Analysis.} Given $\oplus_{r \in R} \sketch(r)$, we next argue that the output $\tilde{k}= 2^{i^*}$ will be in $[k, 8k]$, with high probability. 

Let $i$ be $\lceil\log k\rceil$ and $z=\frac{k}{2^{\lceil\log k\rceil}}$. Notice that $z\in (1/2,1]$. In the level $i$ of sketches, in each repetition, the probability that at least one relation is active is at least $k 2^{-i} - \binom{k}{2} 2^{-2i} \geq \frac{k}{2^{\lceil\log k\rceil}} - \frac{k^2}{2 \cdot 2^{2\lceil\log k\rceil} } = z - z^2/2 > 3/8,$ by Bonferroni's inequality (i.e., truncated inclusion-exlusion)\cite{bonferroni1936teoria}. Notice that in this analysis, we used only pairwise independence among the random choices of different  relations being active or inactive. If we have at least one relation that is active, we have the bit-wise XOR of a non-empty set of relations. This XOR is non-zero with probability at least $1/n^{10}$. We have $O(\log n)$ identical repetitions in this density level and and thus, by Chernoff bound, with high probability, in at least $1/4$ of them, we see a non-zero XOR. This means, with high probability, we will have $i^*\geq \frac{k}{2^{\lceil\log k\rceil}}$, and thus the output value $2^{i^*}$ is at least $2^{i^*}\geq k$.

Now, let $j$ be a density level such that $j \geq \log k +3$. By union bound, the probability of any relation being active is at most $k2^{-j} \leq 1/8$. When no relation is active, the resulting XOR is all zeros. Hence, by Chernoff bound, in the course of the $O(\log n)$ identical repetitions, with high probability, we have strictly less than $1/4$ non-zero XORs. This means, with high probability, we will have $i^* \leq \log k +3$ and thus the output value $2^{i^*}$ is at most $2^{\log{k}+3} =  8k$.
\end{proof}

\remove{
\begin{lemma} \label{lemma_sketches}
Assume that for any tree edge $t$, and for any min 2-respecting cut of the form $\{t,t'\}$ there is a specific vertex in the fragment of $t$ that knows a unique $O(\log n)$-bit name of the cut, such that the same name is known also to one specific vertex in the fragment of $t'$. Then, in $\tilde{O}(D+\sqrt{n})$ time, all tree edges $t$ can learn a value $\sketch(t)$ with the following properties.
\begin{enumerate}
\item $\sketch(t)$ has size $\Theta(\log^3 n)$ bits.
\item For any tree path $P$, given $\oplus_{t \in P} \sketch(t)$, we can compute a value $\tilde{k}\in [k, 8k]$, where $k$ denotes the number of min 2-respecting cuts with exactly one edge in $P$, with high probability.
\end{enumerate} 
\end{lemma}
}

We next use Lemma \ref{lem:Sketches} to prove Lemma \ref{lemma_sketches}.

\LemSketch*

\begin{proof} For each two tree edges $t$ and $t'$ such that $\{t,t'\}$ is a min 2-respecting cut, we think of this as a virtual relation $r_{t, t'}$ between $t$ and $t'$, in the context of  \Cref{lem:Sketches}. We assume that this relation, and in particular its unique  $O(\log n)$-bit name, is known to one specific vertex $v$ in the fragment of edge $t$ and one specific  vertex $v'$ in the fragment of edge $t'$.  We share $O(\log^3 n)$ bits of randomness among all vertices of the graph, in $O(D+\log^2 n)$ rounds. Then, we apply \Cref{lem:Sketches} so that each vertex that is an endpoint of a relation generates a sketch of it. In particular, vertex $v$ will know $\sketch(r_{t, t'})$, which is a $\Theta(\log^3 n)$-bit message. 

For each tree edge $t$, let us define $\sketch(t)= \oplus_{r \in R_t} \sketch(r)$, where $R_t$ denotes the set of all relations $r_{t,t'}$ between edge $t$ and edges $t'$ in a min 2-respecting cut with $t$. We use $\tilde{O}(D+\sqrt{n})$ rounds of communication to compute $\sketch(t)$ for all tree edges $t$. In particular, $\sketch(t)$ will be known to all vertices of the fragment that contains $t$. For that, we use a separate aggregate computation for each tree edge $t$, where each vertex $v$ of the fragment starts with the XOR of all the relations $r_{t,t'}$ that it knows between tree edge $t$ and other tree edges. Since the fragment has $O(\sqrt{n})$ edges, these are $O(\sqrt{n})$ separate aggregate computations in each fragment. Using standard pipelining in each fragment---see, e.g., \cite{peleg2000distributed}, we can perform all these aggregate computations in parallel in $\tilde{O}(\sqrt{n})$ rounds. At the end, for each tree edge $t$, the two endpoints of $t$ (and in fact all vertices of the fragment of $t$) know $\sketch(t)= \oplus_{r \in R_t} \sketch(r)$.

Finally, consider an arbitrary tree path $P$, and suppose that we are given $\oplus_{t \in P} \sketch(t)$. Notice that relations where zero or two of the endpoints of the relation are in path $P$ do not contribute to  $\oplus_{t \in P} \sketch(t)$; in particular, relations where both edges are in $P$ are added twice to the XOR and are thus canceled out. Hence, $\oplus_{t \in P} \sketch(t)$ is simply equal to $\oplus_{r \in R} \sketch(r)$ where $R$ is the set of virtual relations (i.e., min 2-respecting cuts) where exactly one endpoint of the relation (correspondingly, exactly one of the two edges defining the min 2-respecting cut) is in the path $P$. By \Cref{lem:Sketches}, since we know $\oplus_{r \in R} \sketch(r)$, we get an approximation $\tilde{k}\in [k, 8k]$, where $k=|R|$. That is, $k$ is equal to the number of min 2-respecting cuts with exactly one edge in $P$.
\end{proof}

\subsection{Learning the Cut Names} \label{sec:cut_names}

%As explained above, if we manage to construct such sketches, all non-tree edges get a constant approximation of their cost-effectiveness by computing  $\oplus_{t \in P_e} \sketch(t)$. 
To use Lemma \ref{lemma_sketches}, we next show that for any min 2-respecting cut $\{t,t'\}$, there are indeed specific vertices in the fragments of $t$ and $t'$ that know a unique name of the cut.

\begin{lemma} \label{lem:names}
There is an $\tilde{O}(D+k\sqrt{n})$-round algorithm for learning the names of the cuts.
Concretely, after the algorithm, for any tree edge $t$, and for any min 2-respecting cut of the form $\{t,t'\}$ there is a specific vertex in the fragment of $t$ that knows a unique $O(\log n)$-bit name of the cut, such that the same name is known also to one specific vertex in the fragment of $t'$.
\end{lemma}

\begin{proof}
Let $\{t,t'\}$ be a min 2-respecting cut, where $F,F'$ are the fragments of $t,t'$, respectively. The proof has several cases, depending if there is an internal edge between the fragments of $t,t'$. At a high-level, if there is an internal edge between $F$ and $F'$, the endpoints of the edge can use the cut information and cover values of edges in the fragments to learn about all the relevant cuts. The more challenging case is that there are no internal edges between the fragments. Here we can exploit the special structure of min 2-respecting cuts in this case. At a high-level, we show that there are two sets of edges, $A,B$ in $F,F'$, respectively, such that the min 2-respecting cuts $\{t,t'\}$ are exactly those where $t \in A, t' \in B$, and we show that the edges $t,t'$ can learn the size of the sets $A,B$. Here although there is no vertex in the fragments of $t,t'$ that knows the original name of the cut $\{t,t'\}$, we can use the information known to create a unique name for each such cut, known both to $t$ and $t'$.

\paragraph{Case 1: There is an internal edge between $F$ and $F'$.} Here we can use the cut information and Claim \ref{claimFrag} to prove the lemma. Note that using aggregate computation in each one of the fragments $F$ and $F'$, they can check if there is an internal edge between the fragments, and can also find the first such edge, sorting the edges according to the ids of vertices. Doing so for all pairs of fragments takes $O(\sqrt{n})$ time using pipelining. Let $e=\{v,v'\}$ be the edge found between $F$ and $F'$. Then, $v$ and $v'$ are going to be the specific vertices in $F$ and $F'$ that know about all cuts of the form $\{t,t'\}$ where $t \in F, t' \in F'$. Since $e$ has endpoints in both fragments, the vertex $v$ can learn $\{(t,\cutInfo(t),\cov(t))\}_{t \in F}$ and $v'$ can learn $\{(t',\cutInfo(t'),\cov(t'))\}_{t' \in F'}$. This takes $O(k\sqrt{n})$ time by just broadcasting the cut information and cover values of all edges in the fragment to the whole fragment. Sending the whole information over the edge $e$ takes $O(k\sqrt{n})$ time (note that since $e$ is an internal edge it connects exactly two fragments, hence we can compute this in parallel for edges that connect different pairs of fragments). Now from this information and from Claim \ref{claimFrag}, $v$ and $v'$ can learn about all min 2-respecting cuts $\{t,t'\}$ where $t \in F, t' \in F'.$ In this case, $\{t,t'\}$ is the unique name of the cut.

\paragraph{Case 2: There is no internal edge between $F$ and $F'$, and both $t,t'$ are highway edges.} Here our goal would be to exploit the very specific structure of min 2-respecting cuts in this case. We follow the approach in \cite{DBLP:journals/corr/abs-2004-09129} (see Lemma 6.18). 
Let $h_{F},h_{F'}$ be the highways of the fragments $F,F'$, respectively.  
From Claim \ref{covClaim}, $\cut(t,t')=\cov(t)+\cov(t')-2\cov(t,t')$. We observe that in the special case that there are no internal edges between the fragments, the value $\cov(t,t')$ can be broken to a part that only depends on $t$, and a part that only depends on $t'$. Specifically, the edges that cover both $t$ and $t'$ can be divided to 3 parts:
\begin{enumerate}
\item Edges $e=\{u,v\}$ where both $u$ and $v$ are not internal vertices in $F \cup F'$, that cover the whole highways $h_F$ and $h_{F'}$.
\item Edges that have an internal endpoint in $F$ and an endpoint that is not an internal vertex in $F \cup F'$, that cover the whole highway $h_{F'}$.
\item Edges that have an internal endpoint in $F'$ and an endpoint that is not an internal vertex in $F \cup F'$, that cover the whole highway $h_{F}$.
\end{enumerate}
From Claim \ref{globalCover} these are indeed all the options. Concretely, any edge that covers $t,t$ and that both its endpoints are not internal vertices in $F \cup F'$ covers the complete highways $h_F,h_{F'}$. Similarly, any edge that covers $t'$ and both its endpoints are not internal vertices in $F'$ covers the whole highway $h_{F'}$.
Note that the edges from the first type cover any pair of edges $t \in h_F,t' \in h_{F'}$, and hence they do not have any affect on minimizing the expression $\cut(t,t')$. To deal with edges of types 2 and 3, we denote by $\cov(t,F')$ the sum of costs of edges from the second type that cover $t$ (they do not depend on the specific choice of $t'$ as they cover the whole highway $h_{F'}$). Similarly we denote by $\cov(t',F)$ the sum of costs of edges from the third type that cover $t'$. Note that $\cov(t,F')$ and $\cov(t',F)$ can be computed in $O(\sqrt{n})$ time. For example, for computing $\cov(t,F')$ since all edges of type 2 have an endpoint in the fragment $F$, we can run an aggregate computation in the fragment to compute those for all $t$ and for all fragment $F'$, the whole computation takes $O(\sqrt{n})$ time using pipelining. For full details see Claim 6.17 in \cite{DBLP:journals/corr/abs-2004-09129}. 
Now minimizing the expression $\cut(t,t')$ boils down to minimizing both expressions: $(\cov(t)-2\cov(t,F'))$ and $(\cov(t')-2\cov(t',F))$. The first one only depends on $t$, and the second only depends on $t'$. We can use a simple scan of each one of the fragments to minimize the two expressions locally in each one of the fragments. Note that any edge $t \in h_F$ already knows the value $\cov(t)-2\cov(t,F')$, so we just need to find the minimum of all values computed. We denote by $A$ the set of edges in $h_{F}$ that minimize $(\cov(t)-2\cov(t,F'))$, and denote by $B$ the set of edges in $h_{F'}$ that minimize $(\cov(t')-2\cov(t',F))$. From the discussion above, for any pair of edges $t \in A, t' \in B$ the value $\cut(t,t')$ is identical and minimal.  Hence, all the edges in $A$ are candidates to be in a min 2-respecting cut with the edges in $B$. However, knowing that an edge is in $A$ is not sufficient to decide if it is indeed part of a min 2-respecting cut, as the value $\cut(t,t')$ depends also on the information $(\cov(t')-2\cov(t',F))$ known only to edges in $F'$, and on the sum of costs of edges from type 1, that is not computed (computing it for all pairs of fragments would take a linear time). 

To decide if there are indeed min 2-respecting cuts with edges in both fragments, we use $\fragInfo_h(t,F')$. Recall that each tree edge $t$ already computed the number of min 2-respecting cuts with the second edge in $h_{F'}$ (see Lemma \ref{cut_info_pt_highway}). 
If there is at least one edge $t \in h_F$ with $\fragInfo_h(t,F')>0$, we know that there are such min 2-respecting cuts, and moreover we know that any pair of edges $t \in A, t' \in B$ defines a min 2-respecting cut because the value $\cut(t,t')$ is identical for all these pairs. Moreover, an edge $t \in A$ knows $|B|$, as this is just equal to $\fragInfo_h(t,F')$. Similarly, an edge $t' \in B$ knows $|A|$. While this does not provide the original names $\{t,t'\}$ of the min 2-respecting cuts, this is enough to provide unique names for them, which is enough for our purposes. We denote by $a_1,a_2,...,a_{|A|}$ the edges in $A$ ordered according to their position in $h_F$ from the lowest to the highest, and by $b_1,b_2,...,b_{|B|}$ the edges in $B$, from the lowest to the highest. We denote by $(h,F,F',i,j)$ the name of the min 2-respecting cut $\{a_i,b_j\}$, where the order between $F$ and $F'$ is determined by the ids of the fragments (having the fragment with lower id first). This name is different for different cuts. To complete the proof, we need to show that there is a vertex in each one of the fragments that knows the names of the cuts. For each edge $a_i$ we show that its lower endpoint can be this special vertex, the same can be shown for the edges $b_j$. To do so, we run an aggregate computation in which each vertex $u \in h_F$ learns how many edges in $A$ are below it in $h_F$. For the lower endpoint of $a_i \in A$ this number is just $i-1$, and since it knows that $a_i \in A$ and the size of $B$, it knows that it is the $i$'th edge in $A$ and can deduce all the names $(h,F,F',i,j)$ of cuts $\{a_i,b_j\}$. As we only run one aggregate computation inside $F$ per fragment $F'$, we can run the computation for all pairs of fragments in $O(\sqrt{n})$ time. This completes the proof of this case.

\paragraph{Case 3: There is no internal edge between $F$ and $F'$, and at least one of $t,t'$ is a non-highway edge.}
First, if both $t,t'$ are non-highway edges, then there is an internal edge between $F$ and $F'$ from Claim \ref{subtreeEdge}. Concretely, there is an edge connecting the subtrees $T_t$ and $T_{t'}$ below $t$ and $t'$, that include only internal vertices in $F$ and $F'$ in the case $t$ and $t'$ are non-highway edges.
Hence, we need to take care only of the case that exactly one of $t,t'$ is a non-highway edge. Assume without loss of generality that $t \in F$ is a non-highway edge, where $t' \in F'$ is a highway edge in $h_{F'}$.

Let $(c_{F'},n_{F'})$ be the minimum cover value in $h_{F'}$ and the number of edges in $h_{F'}$ with cover value $c_{F'}$. This information can be globally known for all $F'$ in $O(D+\sqrt{n})$ time. We show that either these edges are exactly all the edges $t' \in h_{F'}$ where $\cut(t,t')=k$ or there are no edges in $h_{F'}$ where $\cut(t,t')=k$, and that $t$ can distinguish between the two cases. 
First, note that since $t$ is a non-highway edge, any edge that covers it has an endpoint in the subtree below $t$ that has only internal vertices in $t$'s fragment. Since there are no internal edges between $F$ and $F'$, all the edges that cover $t$ and $t' \in h_{F'}$ are global edges that cover the whole highway $h_{F'}$ from Claim \ref{globalCover}. It follows that $\cov(t,t')$ is fixed for $t' \in h_{F'}$, and that $t$ can compute this value by aggregate computation in its fragment, summing the cost of all the edges with endpoint in the subtree of $t$ that cover the whole highway $h_{F'}$. Doing so for all fragments takes $O(\sqrt{n})$ time via pipelining. Now, $\cut(t,t')=\cov(t)+\cov(t')-2\cov(t,t')$. As $\cov(t,t')$ is fixed for edges in $h_{F'}$, the value $\cut(t,t')$ is minimized for $t' \in h_{F'}$ with minimum cover value. As $t$ knows $\cov(t),\cov(t,t')$ and the minimum cover value $\cov(t')$ for $t' \in h_{F'}$, it can compute the minimum cut value $\cut(t,t')$ for $t' \in h_{F'}$. If it equals $k$, it follows that any edge with cover value $c_{F'}$ is in a min 2-respecting cut with $t$, and otherwise $\cut(t,t')>k$ for any edge $t' \in h_{F'}$. This gives $t$ sufficient information about min 2-respecting cuts $\{t,t'\}$ where $t' \in h_{F}$. The main challenge is to show that $t'$ can also learn about these cuts.

We suggest the following algorithm. For any pair of fragments $F$ and $F'$ where there are no internal edges between the fragments, we count the number $n_{F,F'}$ of non-highway edges in $F$ that participate in a min 2-respecting cut with an edge in $h_{F'}$. If this number is greater than 0 we broadcast the tuple $(F,F',n_{F,F'})$ to the whole graph. Note that since non-highway edges already know if they participate in a min 2-respecting cut with an edge in $h_{F'}$, we can compute the values $n_{F,F'}$ for any $F'$ locally in $F$ in $O(\sqrt{n})$ time via pipelining. At a first glance, it is not clear why broadcasting all relevant values $n_{F,F'}$ is efficient, as there are $\Omega(n)$ pairs of fragments $F,F'$. We next show that only in $O(\sqrt{n})$ cases the values computed are non-zero, so we just need to broadcast $O(\sqrt{n})$ values. For this, we use the following observations:
\begin{enumerate}
\item For each tree edge $t'$, there is a unique path $P_{t'}$ where all the edges $t$, where $\cut(t,t')=k$ are in $P_{t'}$ (Lemma \ref{unique_path_lemma}). If the endpoints of $P_{t'}$ are $u$ and $v$, then $P_{t'}$ is composed of parts in the fragments of $u$ and $v$ and complete highways of other fragments. In particular, $P_{t'}$ may contain non-highway paths in just 2 fragments.
\item As shown above, if there are no internal edges between $F$ and $F'$, then for any non-highway edge $t \in F$, either all the edges in $h_{F'}$ with cover value $c_{F'}$ are in a min 2-respecting cut with $t$, or none of the edges in $h_{F'}$ is in such cut.
\end{enumerate}

From the following observations, if there is an edge $t' \in h_{F'}$ where $\cut(t,t')=k$ for a non-highway edge $t \in F$ from this case, there are only 2 possible options $F_1,F_2$ for the identity of the fragment $F$. Moreover, it cannot be that there is another edge $t'' \in h_{F'}$ where $\cut(t,t'')=k$ for a non-highway edge $t$ from a different fragment $F_3 \neq F_1,F_2$  where there are no internal edges between $F'$ and $F_3$, as otherwise it follows from the second observation that also $\cut(t,t')=k$, but for $t'$ the only relevant fragments could be $F_1$ and $F_2$ from the first observation. To conclude, there can only be two fragments $F_1,F_2$ not connected by an internal edge to $F'$ where there are non-highway edges in a min 2-respecting cut with an edge in $h_{F'}$. It follows that we only broadcast $O(\sqrt{n})$ values, as needed.
Note that from the observations above, the non-highway edges in $F_1$ and $F_2$ that are in a min 2-respecting cut with highway edges in $F$ are just on one non-highway path in each of the fragments, we use it to order them.

To complete the proof, we show that from the observations above, both $t$ and $t'$ learn about a unique name of the cut $\{t,t'\}$. Let $A$ be the set of non-highway edges in $F$ in a cut with an edge in $h_{F'}$, and let $B$ be the set of highway edges in $h_{F'}$ with a minimum cover value. Denote by $a_1,a_2,...,a_{|A|}$ the edges in $A$ from the lowest to the highest (from the discussion above they are all in one path). Denote by $b_1,...,b_{|B|}$ the edges of $B$ from the lowest to the highest. Each edge in $A$ or $B$ can learn its name in the sequence, by computing the number of non-highway edges below them in $A$ or $B$, respectively (each edge already knows if it is in one of the sets). We denote by $(n,F,F',i,j)$ the name of the cut $\{a_i,b_j\}$, where the fragment $F$ with non-highway edges in the cut appears first. These names are unique for different cuts, and all tree edges in $A$ and $B$ can compute the names of the cuts they participate in. This follows as edges in $A$ know their number $i$ in the sequence and the size of $B$, and edges in $B$ know their number $j$ in the sequence and the size of $A$ (they learned it from the broadcast). This completes the proof. 
\end{proof}

\subsection{Putting Everything Together} \label{sec:puttingEverything}

We next prove the main results of the paper.

\CostEfThm*

\begin{proof}
From Lemma \ref{lem:reduction}, it is enough to explain how each edge $e \not \in H$ learns $O(1)$-approximation for the number of minimum 2-respecting cuts $e$ covers with respect to a spanning tree $T$ of $H$.
The number of min 1-respecting cuts can be computed exactly in $\tilde{O}(D+\sqrt{n})$ time as explained at the beginning of Section \ref{sec:cost_ef}. The number of min 2-respecting cuts with exactly 2 edges can be approximated using Lemmas \ref{lemma_sketches} and \ref{lem:names}. More concretely, they allow the computation of the sketches $\sketch(t)$. If each edge $e$ learns $\oplus_{t \in P_e} \sketch(t)$ this gives $O(1)$-approximation to the number of min 2-respecting cuts covered by $e$ from Lemma \ref{lemma_sketches} and Claim \ref{claim_cut_cover}. Computing the values $\oplus_{t \in P_e} \sketch(t)$ for all edges $e$ can be done in $\tilde{O}(D+\sqrt{n})$ time using Claim \ref{aggregate_Pe}.
All the computations in the algorithm take at most $\tilde{O}(D+k\sqrt{n})$ time and work with high probability as discussed throughout.
This gives $O(1)$-approximation to the number of min 2-respecting cuts covered by each edge $e$, and from Lemma \ref{lem:reduction} also an $O(1)$-approximation to the number $|C_e|$ of minimum cuts of size $k$ in $H$ covered by $e$. In particular, this gives an $O(1)$-approximation to the cost-effectiveness of edges $\rho(e) = \frac{|C_e|}{w(e)}$.
\end{proof}

As an immediate corollary of Theorems \ref{thm_greedy_augi} and \ref{thm_cost_ef} we get the following.

\mainThm*

\remove{
\begin{theorem}
There is a distributed $O(\log{k}\log{n})$-approximation algorithm for the minimum cost $k$-edge-connected spanning subgraph problem that takes $\tilde{O}(k(D+k\sqrt{n}))$ time in the \congest model. The algorithm works with high probability.
\end{theorem}
}

We next prove \cref{crl:bicriteria}.

\corBic*

\begin{proof}[Proof of \Cref{crl:bicriteria}]
Fix a value $\epsilon=\Theta(1/\log n)$. For $k=O(\log n/\epsilon^2)$, the result follows directly from \Cref{thm:main}. Suppose $k=\Omega(\log n/\epsilon^2)$. Let $\rho=C\log n/\epsilon^2$, where $C\geq 1$ is a large enough constant. Partition the edges of the graph randomly into $\frac{k}{\rho}$ parts, where each edge is placed in one of these parts randomly. In each part, we use \Cref{thm:main} to compute an $O(\log n\log\log n)$ approximation of $(\rho(1-\epsilon))$-connected spanning subgraph. Since $\rho=C\log n/\epsilon^2$, the algorithm for each part would run in $\tilde{O}(D+\sqrt{n})$ rounds, and since there are no more than $k$ algorithms, we can run them all in $\tilde{O}(k(D+\sqrt{n}))$ rounds. Overall, this is a graph with $\frac{k}{\rho} \cdot \rho(1-\epsilon)= k(1-\epsilon)$ connectivity. Let us now argue about the cost. Let $H$ be the minimum-cost $k$-edge-connected spanning subgraph. Then, for each part $i\in [k/\rho]$, the edges of $H$ in part $i$ form a $\rho(1-\epsilon)$-edge connected spanning subrgaph. This follows from Karger's random edge sampling result~\cite{karger1993global} which shows that if we sample edges randomly and independently such that the expected minimum-cut size is at least $C\log n/\epsilon^2$, then the deviation from expectation in the size of each cut is at most a $1\pm \epsilon$ factor. Hence, the algorithm that we had pays a cost within an $O(\log{n}\log{\rho})=O(\log n\log\log n)$ factor of the portion of $H$ in part $i$. This means, overall, our algorithm pays a cost within an $O(\log n\log\log n)$ factor of the cost of $H$.
\end{proof}

%\footnote{\label{note1} We note that this complexity can be improved to  $\tilde{O}(D+\sqrt{nk}))$ using a simple rebalancing of the parameters of underlying MST algorithm\cite{kutten1998fast}. This is not easy to elaborate succinctly as it requires recalling the algorithm of \cite{kutten1998fast}. We thus defer the formal description to a full version of this paper. For the reader familiar with that algorithm, we note the change is to simply define \textit{large} fragments to be those with size greater than $\sqrt{nk}$ (instead of size greater than $\sqrt{n}$, as standard). Then, small fragments communicate only inside small fragments, which are in edge disjoint parts. There are overall $\tilde{O}(k) \cdot \frac{n}{\sqrt{nk}}=\tilde{O}(\sqrt{nk})$ large fragments and their messages can be pipelined using the global BFS tree in $\tilde{O}(D+\sqrt{nk}))$ time.}

\section*{Acknowledgements}

This work was supported in part by funding from the European Research Council (ERC) under the European Union’s Horizon 2020 research and innovation programme (grant agreement No. 853109), and the Swiss National Foundation (project grant 200021-184735).

\bibliographystyle{alpha}
\bibliography{kECSS}

\appendix

\section{Lower Bounds} \label{sec:lowerBounds}

In this section, we show how to extend lower bounds from \cite{sarma2012distributed,censor2020fast,DBLP:conf/wdag/GhaffariK13}, to prove a lower bound for minimum $k$-ECSS.
First, it is easy to show an $\Omega(D)$ lower bound, even in the \local model where the size of messages is unbounded.

\begin{lemma} \label{lemma_LB_D}
Any $\alpha$-approximation algorithm for minimum $k$-ECSS requires $\Omega(D)$ rounds, even in the \local model.
\end{lemma}

\begin{proof}
To illustrate the proof we first focus on the case that $k=1$ (we just want to find a minimum spanning tree). Look at a graph that is a cycle with $n$ vertices, all edges of the cycle have weight 0, except of two edges $e,e'$ of distance $n/2 -1 = \Omega(D)$ from each other. The edge $e$ has weight $\alpha + 1$, where the edge $e'$ can have weight either $1$ or $\alpha(\alpha+1) + 1$. Note that in the first case the edge $e'$ is a part of the MST and $e$ is not, and adding $e$ would increase the cost of the MST by more than $\alpha$ factor. On the other hand, in the second graph $e$ is in the MST and $e'$ is not. Here, if we do not take $e$ we either do not get a spanning tree, or have a spanning tree with cost more than $\alpha$ times the optimal if we take $e'$. However, in less than $n/2-1$ rounds, $e$ cannot distinguish between the two cases, which means that in both cases we either take $e$ or do not take $e$ to the solution, hence in at least one of the graphs we do not get $\alpha$-approximation for the MST.

To extend this proof idea to general $k$, we replace the vertices of the cycle by cliques of $k$ vertices. Here we have $n/k$ cliques of size $k$. In each clique all vertices are connected by zero weight edges. We have a cycle structure over the $n/k$ cliques, where two adjacent cliques are connected by a matching of $k$ edges, where the $k$ vertices of one clique are connected to the $k$ vertices of the second clique. All the matching edges have weight 0, except of the edges of two matchings that are at distance $\Omega(n/k)=\Omega(D)$ from each other. In one of these matchings the cost of all edges is $\alpha k + 1$, where in the second matching the cost of all edges is either 1 or $(\alpha k + 1) \alpha k + 1$. Again, it is easy to check that in the second graph we need to take all the edges of the first matching to get an $\alpha$-approximate minimum $k$-ECSS, where in the first graph we should avoid taking these edges, but in less than $\Omega(D)$ rounds, the vertices adjacent to these edges cannot distinguish between the two cases.
\end{proof}

We next show stronger lower bounds in the \congest model. We start by proving an $\tilde{\Omega}(\sqrt{n})$ lower bound for graphs with parallel edges, and then adapt the construction to simple graphs.

\subsection{Construction with Parallel Edges}

We follow a graph construction from \cite{DBLP:journals/corr/abs-1305-5520}, that is a variant of the constructions used in \cite{sarma2012distributed,censor2020fast}.
The basic idea is to reduce our problem to a problem in communication complexity, and use lower bounds from communication complexity to obtain a lower bound. We focus on the set disjointness problem, in which two players, Alice and Bob, receive input strings $a=(a_1,...,a_{\ell}),b=(b_1,...,b_{\ell})$ of $\ell$ bits and their goal is to determine if their inputs are disjoint, i.e., to determine if there is an index $j$ such that $a_j=b_j=1$ or not. It is well-known that the communication complexity of set disjointness is $\Omega(\ell)$, meaning that Alice and Bob should exchange at least $\Omega(\ell)$ bits to solve the problem, even if they are allowed to use randomization.

\begin{lemma}[\cite{razborov1992distributional}]
The communication complexity of set disjointness is $\Omega(\ell)$, even using randomized protocols.
\end{lemma}

To show a reduction, we define a family of graphs with the following structure. 
Given a parameter $x$, we define a graph $H$ that has $x$ paths of length $y=n/x$, and additional edges that reduce the diameter of the graph to $O(\log{n})$.
Formally, the graph has $n=xy$ vertices $v_{i,j}$ for $0 \leq i \leq x-1, 0 \leq j \leq y-1$. For each vertex we assign a unique id from $\{0,...,n-1\}$ where for $v_{i,j}$ the id is $i+xj$.
The graph has 3 sets of edges $E_{H,1}, E_{H,2}, E_{H,3}$, where 
$$E_{H,1} = \{\{v_{i,j}, v_{i,j+1}\} | 0 \leq i \leq x-1, 0 \leq j \leq y-2\},$$
$$E_{H,2} = \{ \{v_{0,j}, v_{0,j'} \} |\exists s \in  \mathbb{N} \ s.t. \ j \equiv 0 \ (mod\ 2^s) \ and \ j' = j+2^s \},$$
$$E_{H,3} = \{ \{v_{0,j},v_{i,j}\} | 1 \leq i \leq x-1, 0 \leq j \leq y-1\}.$$

The edges $E_{H,1}$ are the edges of the $x$ paths, the edges of $E_{H,2}$ reduce the diameter of the first path to $O(\log{n})$ (similarly to the role of the tree in the constructions in \cite{sarma2012distributed,censor2020fast}), where the edges of $E_{H,3}$ are star edges that connect the $j$'th vertex of the first path with the $j$'th vertices of other paths. See Figure \ref{LBpic} for illustration.

\setlength{\intextsep}{0pt}
\begin{figure}[h]
\centering
\setlength{\abovecaptionskip}{-4pt}
\setlength{\belowcaptionskip}{4pt}
\includegraphics[scale=0.6]{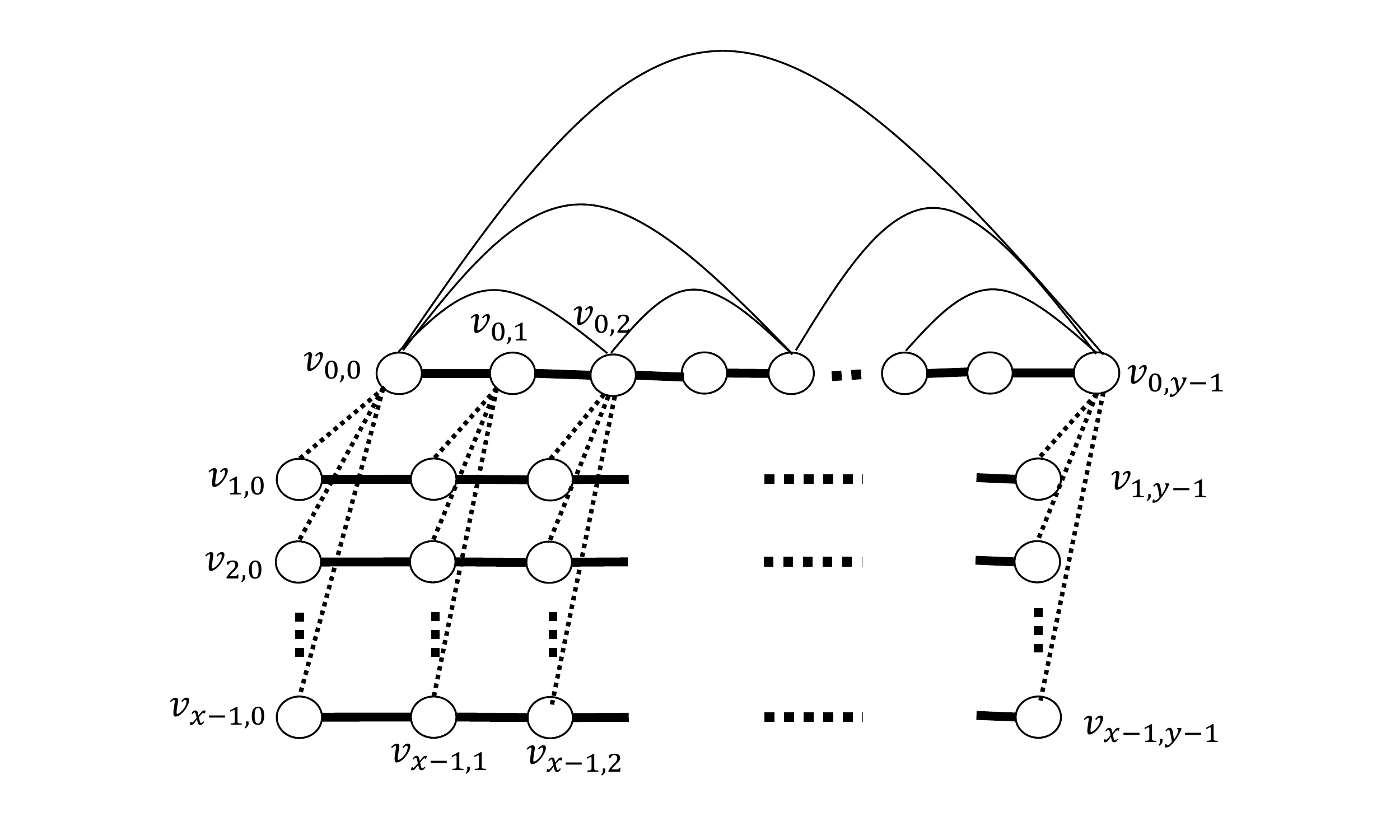}
 \caption{The graph $H$. The bold edges are the edges of the paths, the dashed edges are star edges, and the edges on the top are edges of $E_{H,2}$ that reduce the diameter of the graph.}
\label{LBpic}
\end{figure}

To use the above graph to show a lower bound, we first make sure that $H$ is $k$-edge-connected. Additionally, we define a family of graphs $H_{a,b}$ that are obtained from $H$ by adapting the weights of the edges in $H$ according to two input strings $a,b$ of length $x-1$.
For this, we first replace each one of the edges in  $E_{H,1}$ (the edges of the paths) by $k$ parallel edges. All these edges have weight 0. We leave the edges of $E_{H,2}$ as is, and give them weight 0 as well. We leave all the edges of $E_{H,3}$ as is and give them weight $\alpha x  k+ 1$, except of the edges of the first and last stars that depend on the inputs. The edges of the first and last star are determined by two inputs of $x-1$ bits, denoted by $a=(a_1,...,a_{x-1}), b=(b_1,...,b_{x-1})$, as follows. The edge $\{v_{0,0},v_{i,0}\}$ of the first star is replaced by $k$ parallel edges. These edges have weight 1 if $a_i=0$ and have weight $\alpha x k + 1$ if $a_i=1$. Similarly, each edge $\{v_{0,y-1},v_{i,y-1}\}$ of the last star is replaced by $k$ parallel edges with weight 1 or $\alpha x k + 1 $ depending if $b_i$ is equal to 0 or 1.
We start with the following observation. 

\begin{claim} \label{claim_disjoint}
The graph $H_{a,b}$ is $k$-edge-connected. The cost of the minimum cost $k$-ECSS of $H_{a,b}$ is at most $x k$ if $a$ and $b$ are disjoint, and is at least $\alpha x k + 1$ if they are not disjoint.
\end{claim}

\begin{proof}
Since all the edges of the paths were replaced by $k$ parallel edges, then the paths are $k$-edge-connected. To make sure that the whole graph is $k$-edge-connected note that there are more than $k$ star edges that connect any path to the first path, which guarantees the connectivity.
To construct a $k$-edge-connected spanning subgraph, first we take all the (parallel edges) of the paths, they all have weight 0. Next we should connect the different paths. If the inputs are disjoint, we know that for any $i$ at least one of $a_i,b_i$ is equal to 0. If $a_i=0$, we can take the $k$ parallel edges between $v_{0,0}$ and $v_{i,0}$ to connect the $i$'th path with the first path, similarly if $b_i=0$ we can use the $k$ parallel edges between $v_{0,y-1}$ and $v_{i,y-1}$ to connect the $i$'th path with the first path. The total weight of these edges is $k$ as $a_i$ or $b_i$ are equal to 0. Hence, connecting all the $x-1$ paths to the first path results in a $k$-ECSS with cost $(x-1)k < xk$. On the other hand, if the input strings are disjoint there is an index $i$ where both $a_i=b_i=1$. In this case, all the edges that connect the $i$'th path to the rest of the graph have weight $\alpha x k + 1$, hence any $k$-ECSS has weight at least $\alpha x k + 1$. 
\end{proof}

Claim \ref{claim_disjoint} suggests that an $\alpha$-approximation algorithm for minimum $k$-ECSS can be used to solve set disjointness. Alice and Bob can simulate the algorithm, deduce an $\alpha$-approximation of the minimum cost $k$-ECSS, and according to it learn if their input strings are disjoint. To formulate this intuition, we use the following Lemma based on \cite{DBLP:journals/corr/abs-1305-5520}. 

\begin{lemma} \label{lemma:simulation_cc}
If there is a randomized $\alpha$-approximation algorithm for $k$-ECSS that runs in $T \leq (n-2x)/2x$ rounds in the \congest model on the family of graphs $H_{a,b}$, then there is a randomized protocol for solving set disjointness with communication complexity at most $2B \log{n} T$, where $B=\Theta(\log{n})$ is the size of messages sent in the \congest algorithm. 
\end{lemma}

The proof of Lemma \ref{lemma:simulation_cc} follows from \cite{DBLP:journals/corr/abs-1305-5520}, where a more general simulation theorem is shown for proving lower bounds in \congest based on lower bounds in communication complexity (see Theorem 6.1 in \cite{DBLP:journals/corr/abs-1305-5520}). The proof requires the graph family to satisfy some special properties, and in Lemma 6.2 in \cite{DBLP:journals/corr/abs-1305-5520} it is shown that the family of graphs $H_{a,b}$ satisfies the required properties. In more detail, Lemma 6.2 focuses on graphs that have the exact same structure of $H_{a,b}$, with the only difference that we replaced some edges in $E_{H,1}$ and $E_{H,3}$ with parallel edges, but the exact same proof works also if there are parallel edges replacing them. The main important property shown in the proof is that there are small number of edges that connect vertices $v,v'$ where $|Id_v - Id_{v'}| > x$. Recall that the id of a vertex $v_{i,j}$ is $i + xj$. It is easy to see (and shown in \cite{DBLP:journals/corr/abs-1305-5520}) that all edges in $E_{H,1} \cup E_{H,3}$ connect vertices $v,v'$ where $|Id_v-Id_{v'}| \leq x$, so adding more edges between them does not change the proof. The edges in $E_{H,2}$ are exactly the same in our construction and in \cite{DBLP:journals/corr/abs-1305-5520}, hence the proof bounding their number works also in our case.  
Using Lemma \ref{lemma:simulation_cc} we show the following.

\begin{theorem} \label{thm_LB_parallel}
Any (possibly randomized) $\alpha$-approximation \congest algorithm for the minimum cost $k$-ECSS problem in multi-graphs requires $\Omega(D+\sqrt{n}/\log{n})$ rounds, for any polynomial $\alpha \geq 1$. The lower bound holds even in graphs with diameter $D=O(\log{n})$. 
\end{theorem}

\begin{proof}
The $\Omega(D)$ lower bound follows from Lemma \ref{lemma_LB_D}, we next show the second part.
From Lemma \ref{lemma:simulation_cc}, if there is a distributed algorithm for minimum $k$-ECSS that takes less than $(n-2x)/2x$ rounds on the graphs $H_{a,b}$, Alice and Bob can solve set disjointness by exchanging $2B \log{n} T$ bits for $B=\Theta(\log{n})$. Since solving disjointness on inputs of size $x-1$ requires exchanging $\Omega(x)$ bits, we get that $T \geq \min \{\frac{n-2x}{2x},\Omega(\frac{x}{\log^2{n}}\}$. Choosing $x=\sqrt{n}\log{n}$ gives an $\Omega(\sqrt{n}/\log{n})$ lower bound. The diameter of the graphs $H_{a,b}$ is $O(\log{n})$ as shown in \cite{DBLP:journals/corr/abs-1305-5520}.
\end{proof}

\subsection{Construction without Parallel Edges}

We next extend our construction to simple graphs. Using parallel edges was helpful to make the graphs $H_{a,b}$ $k$-edge-connected. To obtain the same goal in simple graphs, we slightly change the construction, by replacing the original vertices by cliques of $k$ vertices. A similar approach was used in \cite{DBLP:journals/corr/abs-1305-5520}.
Formally, each vertex $v_{i,j} \in H_{a,b}$ is replaced by $k$ vertices $v_{i,j,\ell}$ for $0 \leq \ell \leq k-1$. The graph now has $n=xyk$ vertices.
For all $i,j$, we have a clique between the vertices $v_{i,j,\ell}$ with 0 weight edges. 
The edges $E_{H,1}$ of the paths are replaced by $k$ zero weight edges between the corresponding cliques. Concretely, the $k$ parallel edges $\{v_{i,j}, v_{i,j+1}\}$ are replaced by the $k$ edges $\{v_{i,j,\ell}, v_{i,j+1,\ell}\}$. For any edge $\{v_{0,j},v_{0,j'}\} \in E_{H,2}$ we just have one zero weight edge $\{v_{0,j,0},v_{0,j',0}\}$ (this is enough to keep the diameter small, which is the only role of these edges). For the star edges $\{v_{0,j},v_{i,j}\} \in E_{H,3}$, if $j \neq 0, y-1$, we have one edge $\{v_{0,j,0},v_{i,j,0}\}$ of weight $\alpha x k + 1$. The edges of the first and last star depend on the inputs similarly to before. We have $k$ edges $\{v_{0,0,\ell},v_{i,0,\ell}\}$ each has weight $1$ or $\alpha x k + 1$ depending if $a_i$ is equal to 0 or 1. In the last star we have the $k$ edges $\{v_{0,y-1,\ell},v_{i,y-1,\ell}\}$, with weights depending on $b_i$ in the same way.
We call the new graph $H'_{a,b}$. It is easy to see that the graph is $k$-edge-connected. Moreover, if we contract all the cliques we get the graph $H_{a,b}$ we had before. We use it to show that an algorithm for $H'_{a,b}$ gives an algorithm for $H_{a,b}$, or equivalently a lower bound for $H_{a,b}$ gives a lower bound for $H'_{a,b}$. The lower bound obtained is $\tilde{\Omega}(\sqrt{n/k})$ since the number of vertices in $H'_{a,b}$ is $n=Nk$, where $N$ is the number of vertices in $H_{a,b}$.

\begin{theorem}
Any (possibly randomized) $\alpha$-approximation \congest algorithm for the minimum cost $k$-ECSS problem in simple graphs requires $\Omega(D+\sqrt{n/k}/\log{n})$ rounds, for any polynomial $\alpha \geq 1$. The lower bound holds even in graphs with diameter $D=O(\log{n})$. 
\end{theorem}

\begin{proof}
Again, the $\Omega(D)$ lower bound follows from Lemma \ref{lemma_LB_D}, we next prove the second part.
First, we show that a distributed algorithm for the graphs $H'_{a,b}$ gives a distributed algorithm for the graphs $H_{a,b}$. Say that $A$ is a distributed $\alpha$-approximation algorithm for simple graphs that takes $T$ rounds on the graph $H'_{a,b}$. We show a distributed $\alpha$-approximation algorithm that takes $T$ rounds on the graph $H_{a,b}$. For this, any vertex $v_{i,j} \in H_{a,b}$ simulates the clique of vertices $v_{i,j,\ell} \in H'_{a,b}$. This is done as follows. In each round, each vertex sends and receives messages from its neighbours. Simulating messages over the clique edges can be done internally by the vertex that simulates the clique. For any other edge in $H'_{a,b}$ there exists an edge in $H_{a,b}$ by the construction, and the vertices of $H_{a,b}$ can send messages over these edges. For example, for the path edges $\{v_{i,j,\ell}, v_{i,j+1,\ell}\} \in H'_{a,b}$ there exist $k$ parallel edges between $v_{i,j}$ and $v_{i,j+1}$ in $H_{a,b}$, so these vertices can send all the $k$ messages between them in one round. This allows simulating the whole algorithm in $T$ rounds. 

We next show that an $\alpha$-approximation for $k$-ECSS in $H'_{a,b}$ gives an $\alpha$-approximation for $k$-ECSS in $H_{a,b}$. This follows from the fact that there is a correspondence between $k$-ECSS in the two graphs. Let $K$ be a $k$-ECSS in $H'_{a,b}$, we can assume that $K$ has all the clique edges, as we can add all of them to $K$ without changing its cost (as they have weight 0). Note that if we contract all the cliques in $H'_{a,b}$, allowing parallel edges, the resulting graph is exactly identical to the graph $H_{a,b}$, we have the exact same edges with the exact same weights. Since $K$ is $k$-edge-connected, the edges of $K$ in the contracted graph give a $k$-ECSS for $H_{a,b}$. Similarly, any $k$-ECSS graph in $H_{a,b}$ gives a $k$-ECSS graph in $H'_{a,b}$ if we add to it all the zero weight clique edges. To conclude, any $\alpha$-approximate $k$-ECSS graph in $H'_{a,b}$, corresponds to an $\alpha$-approximate $k$-ECSS graph in $H_{a,b}$, as needed. 

Finally, as a $T$-round $\alpha$-approximation for $H'_{a,b}$ gives a $T$-round $\alpha$-approximation for $H_{a,b}$, a lower bound for $H_{a,b}$ gives a lower bound for $H'_{a,b}$. Let $N$ be the number of vertices in $H_{a,b}$. By construction, the number of edges in $H'_{a,b}$ is $n=Nk$. From Theorem \ref{thm_LB_parallel}, there is a lower bound of $\Omega(\sqrt{N}/\log{N})$ for approximating $k$-ECSS in $H_{a,b}$, this leads to a lower bound of $\Omega(\sqrt{n/k}/\log{n})$ in the simple graphs $H'_{a,b}$. Note that the diameter in $H'_{a,b}$ can only decrease with respect to $H_{a,b}$, hence it is still bounded by $O(\log{n})$.
\end{proof}

\section{Missing Proofs from Section \ref{sec:setcover}} \label{sec:app_setcover}

\ApproxClaim*

\begin{proof}
Let $e \not \in H$, and denote by $C_1,C_2,...,C_{\ell}$ the cuts of size $k-1$ in $H$ covered by $e$ according to the order they are covered in the algorithm. In the iteration that $C_1$ is covered, $\rho(e)=\frac{\ell}{w(e)}$. Let $i'$ be the first index such that $\frac{\ell}{w(e)} > \frac{M}{2^{i'}}$. Then, from Claim \ref{claim:cost_ef_decrease}, $C_1$ is covered in an epoch $i < i'$ (as after this epoch, the cost-effectiveness of edges is smaller than $\frac{M}{2^{i'}} < \frac{\ell}{w(e)}$). Hence, in the epoch $i$ that $C_1$ is covered we have  $\frac{\ell}{w(e)} \leq \frac{M}{2^{i}}$.  By the definition of the costs, it implies that $cost(C_1) = \frac{2^i}{M} \leq \frac{w(e)}{\ell}$. 
Similarly, in the iteration where the cut $C_j$ is covered, $\rho(e) \geq \frac{\ell-j+1}{w(e)}$, as the cuts $C_1,...,C_{j-1}$ could already be covered, but the cuts $C_j,...,C_{\ell}$ are not covered at the beginning of the iteration. The same arguments show that $cost(C_j) \leq \frac{w(e)}{\ell-j+1}$. This gives,
$$\sum_{j=1}^{\ell} cost(C_i) \leq \sum_{j=1}^{\ell} \frac{w(e)}{\ell-j+1} = O(\log \ell) w(e) = O(\log n) w(e),$$
where the last inequality holds since $\ell \leq n^2$, as the number of minimum cuts in a graph is at most $O(n^2)$. 

We use this to bound the sum of costs of all cuts of size $k-1$ in $H$. Denote by $S_e$ all the cuts of size $k-1$ in $H$ covered by $e$, and by $E_C$ the set of edges that cover a cut $C$. Let $x(e)$ be the values given to edges in the optimal fractional solution $A^*$. Then we have
$$\sum_C cost(C) \leq \sum_C cost(C) \sum_{e \in E_C} x(e) = \sum_{e \not \in H} x(e) \sum_{C \in S_e} cost(C) \leq O(\log{n}) \sum_{e \not \in H} x(e) w(e) = O(\log{n}) w(A^*).$$
The first inequality uses the fact that $\sum_{e \in E_C} x(e) \geq 1$ by the definition of fractional solution, the third inequality uses the inequality above, and the last equality is by the definition of $w(A^*).$ This completes the proof.
\end{proof}

\ApproxExpect*

To prove Claim \ref{claim_approx2}, first the following is shown (see Lemma 4.5 in \cite{dory2018distributed}). Recall that $deg(C)$ is the number of candidates that cover the cut $C$.

\begin{claim}
Fix an epoch $i$. At the beginning of phase $j$ of epoch $i$, $deg(C) \leq m/2^j$ with high probability. 
\end{claim}

\begin{proof}[Proof sketch] Note that for $j=0$ this clearly holds, as there are at most $m$ candidates. Note also that $deg(C)$ can only decrease during the epoch, as edges can only be removed from $Candidates_i$ during the epoch. Now if the claim holds for phase $j$, we know that at the beginning of phase $j$ we have $deg(C) \leq m/2^j$, and we would like to show that by the end of the phase $deg(C) \leq m/2^{j+1}$. For this, take any cut where $m/2^{j+1} \leq deg(C) \leq m/2^j$. As all candidates are added to $A$ with probability $2^j/m$ during phase $j$, we have a constant probability to cover $C$ in each iteration of the phase. After $O(\log{n})$ iterations $C$ is covered w.h.p. A union bound shows that it holds for all cuts w.h.p. For full details see \cite{dory2018distributed}.
\end{proof}

The proof of Claim \ref{claim_approx2} is based on the following idea. Fix an iteration $\ell$, let $A_{\ell}$ be the candidates added to $A$ in iteration $\ell$, and let $CUT_{\ell}$ be the cuts of size $k-1$ in $H$ that were first covered in iteration $\ell$. The goal is to show that $w(A_{\ell}) \approx \sum_{C \in CUT_{\ell}} cost(C)$. If this holds, then summing over all iterations gives $w(A) \approx \sum_C cost(C)$ as needed. To prove it we take a closer look at these expressions.
First, let $i$ be the epoch of iteration $\ell$, then for all cuts $C \in CUT_{\ell}$ we have $cost(C)= 2^i/M$ by definition. 
Hence $\sum_{C \in CUT_{\ell}} cost(C) = \frac{2^i}{M} |CUT_{\ell}|$. 
Second, $$w(A_{\ell})= \sum_{e \in A_{\ell}} w(e) = \sum_{e \in A_{\ell}} |C_e| \cdot \frac{w(e)}{|C_e|} \leq \alpha^2 \frac{2^i}{M} \sum_{e \in A_{\ell}}|C_e|.$$
Here $C_e$ are the cuts in $H \cup A$ covered by $e$ (with $A$ at the beginning of iteration $\ell$). The last inequality follows from the fact that $\rho(e)=\frac{|C_e|}{w(e)}$, $\rho'(e)$ is an $\alpha$-approximation of $\rho(e)$, and we have $\rho'(e) \geq \frac{M}{\alpha 2^i}$. This implies that $\frac{w(e)}{|C_e|}= \frac{1}{\rho(e)} \leq \frac{\alpha}{\rho'(e)}\leq \frac{\alpha^2 2^i}{M}$, as needed.

In \cite{dory2018distributed} it is shown that $E[\sum_{e \in A_{\ell}} |C_e|] \leq 3E[|CUT_{\ell}|]$. Intuitively, the proof is based on showing that each cut is covered by a constant number of edges in expectation. This gives 
$$E[w(A_{\ell})] \leq \alpha^2 \frac{2^i}{M} E[\sum_{e \in A_{\ell}}|C_e|] \leq 3 \alpha^2 \frac{2^i}{M} E[|CUT_{\ell}|] = 3 \alpha^2 E[\sum_{C \in CUT_{\ell}} cost(C)].$$
For full details of the proof see\cite{dory2018distributed}. 

\section{Missing Proofs from Section \ref{sec:succinct_min_cuts}} \label{app_proofs}

\uniquePath*

The proof follows from \cite{mukhopadhyay2020weighted}, we add a proof for completeness.  We first show the following.

\begin{claim}\label{claimInteresting}
If $\{t,t'\}$ is a min 2-respecting cut, then $\cov(t,t') \geq \cov(t)/2$.
\end{claim}

\begin{proof}
Recall that from Claim \ref{covClaim}, $\cut(t,t')=\cov(t)+\cov(t') - 2 \cov(t,t')$. If $\{t,t'\}$ is a min 2-respecting cut, we have $\cut(t,t') \leq \cov(t')$, since $\cov(t')$ is the value of the 1-respecting cut defined by $t'$, hence the minimum cut value is bounded by $\cov(t')$. We get that $\cov(t)+\cov(t') - 2 \cov(t,t') \leq \cov(t')$, or equivalently, $\cov(t,t') \geq \cov(t)/2$, as needed. 
\end{proof}

\begin{proof}[Proof of Lemma \ref{unique_path_lemma}]
Assume to the contrary that there is no tree path that contains all edges $\{t'|\cut(t,t')=k\}$. In particular, it follows that there are two tree edges $t_1,t_2$ such that $\{t,t_1\}$ and $\{t,t_2\}$ are min 2-respecting cuts, but there is no one tree path that contains $t,t_1,t_2$. This means that there is no edge that covers all the 3 edges $t,t_1,t_2$, as an edge $\{u,v\}$ covers the tree path between $u$ and $v$, that can include at most 2 of these edges. In other words, the edges that cover $\{t,t_1\}$ and the edges that cover $\{t,t_2\}$ are disjoint. However, from Claim \ref{claimInteresting}, we have at least $\cov(t)/2$ edges that cover $t$ and $t_1$, and at least $\cov(t)/2$ disjoint edges that cover $t$ and $t_2$. Additionally, the edge $t$ covers $t$ but none of $t_1,t_2$. This means that there are at least $\cov(t)+1$ edges that cover $t$, a contradiction.
\end{proof}

\end{document}